%% file: main_arxiv.tex
\documentclass[12pt]{article}

\usepackage[a4paper,margin=1in]{geometry}
\usepackage{amsmath,amssymb}
\usepackage{natbib}
\usepackage{hyperref}
\usepackage{orcidlink}
\usepackage{enumitem}
\usepackage{booktabs}
\usepackage{microtype}
\usepackage{algorithm}
\usepackage{algpseudocode}
\usepackage{graphicx}
\graphicspath{{Figures/}}
\usepackage{amsthm}
\usepackage{tikz}
\usetikzlibrary{arrows.meta,positioning,shapes.geometric,calc}

\DeclareMathOperator*{\argmax}{argmax}

\usepackage{dsfont}
\usepackage{mathtools}

\newtheorem{theorem}{Theorem}[section]
\newtheorem{corollary}[theorem]{Corollary}
\newtheorem{lemma}[theorem]{Lemma}
\newtheorem{definition}[theorem]{Definition}
\newtheorem{assumption}[theorem]{Assumption}

\newtheorem{prop}[theorem]{Proposition}
\newtheorem{remark}[theorem]{Remark}

\title{BOOST: Power-Optimal Strong-FWER Testing for Block-Structured Multiplicity}
\author{%
  Prasanjit Dubey\,\orcidlink{0000-0002-3667-5507}%
  \qquad
  Xiaoming Huo\,\orcidlink{0000-0003-0101-1206}\\
  H.~Milton Stewart School of Industrial and Systems Engineering,\\
  Georgia Institute of Technology, Atlanta, GA 30332, U.S.A.
}
\date{}

\begin{document}
\let\oldthefootnote\thefootnote
\renewcommand{\thefootnote}{}
\maketitle
\let\thefootnote\oldthefootnote
\setcounter{footnote}{0}

\begin{abstract}
Structured multiple-testing problems (gatekeeping trials, dose-finding, multi-tissue eQTL mapping, bundled-challenger A/B experiments) organize hypotheses into design-imposed blocks and demand strong family-wise error rate (FWER) control for confirmatory claims. Practitioners currently use objective-agnostic stepwise rules (Bonferroni, Holm, Hochberg, Hommel), closed-testing and graphical extensions, or hierarchical and resampling methods; none is power-optimal within the block-separable class these designs induce. We introduce \textsc{BOOST} (Block-Optimal Objective-driven Strong-FWER Testing), the power-optimal strong-FWER procedure for block size three, with three guarantees: (i)~finite-sample strong-FWER validity at $O(K)$ cost (versus $O(K^2)$ for general closed testing) without independence assumptions, with a strict \v{S}id\'ak improvement under cross-block independence; (ii)~power-optimal allocation across heterogeneous blocks via an equalized-marginal KKT condition, solvable by bisection in $O(B\log(1/\varepsilon))$; and (iii)~a sample-split plug-in variant for unknown alternative density $g$, attaining $\alpha$-control up to $O(B_{\mathcal T}\,\mathbb E\|g-\widehat g\|_\infty)$ inflation with per-hypothesis power deficit \emph{independent of $B_{\mathcal T}$}. Simulations across independent, equicorrelated, sparse, and mis-specified regimes show $1.4$--$1.7\times$ power gains over the strongest existing baseline at calibrated FWER. On two published datasets (BLUEPRINT cross-lineage cis-eQTL and Upworthy bundled-challenger A/B experiments), \textsc{BOOST} certifies an order of magnitude more full-block discoveries than existing baselines at controlled FWER.
\end{abstract}

\noindent\textit{Keywords:} multiple hypotheses testing; objective-driven inference; strong family-wise error rate; equalized-marginal allocation; power-optimal procedures.

\section{Introduction}
\label{sec:intro}

Confirmatory multiple-testing analyses rarely treat hypotheses as an unstructured list. Clinical trials test primary endpoints grouped by outcome domain and dose level; gatekeeping designs partition hypotheses into pre-specified families with logical precedence; multi-tissue eQTL studies ask, for each gene, whether a candidate SNP modulates expression across a biologically related panel of tissues; online publishers run bundled-challenger A/B tests asking whether all challenger arms beat a shared baseline~\citep{Pocock87,dmitrienko2007gatekeeping,bretz2011graphical,chen2016blueprint,matias2021upworthy}. All share an operational form: each design-imposed \emph{block} is tested using only its own data and the global error budget is apportioned across blocks. The error metric of record is strong family-wise error rate (strong FWER), the probability of at least one false rejection, uniformly over null/alternative configurations~\citep{holm,hommel,marcus1976closed}, because a single false confirmatory claim can trigger a costly follow-up; the weaker FDR control of~\citet{BH95} is appropriate for exploratory screens rather than confirmatory decisions.

Strong-FWER procedures fall into four categories, each with a concrete limitation in structured settings.
\emph{(i)~Classical stepwise rules} (Bonferroni, Holm, Hochberg, Hommel, \v{S}id\'ak)~\citep{holm,hommel,hochberg1988sharper,sidak1967,romano,HochbergTamhane1987,LehmannRomano2005} apply to any $K$ and tolerate arbitrary or positive dependence, but are objective-agnostic and provably sub-optimal in power.
\emph{(ii)~Closed testing and partitioning}~\citep{marcus1976closed,goeman2010sequential,dobriban2020fact,karmakar2025bottomup,hartog2025evalues} is the dominant general framework, but local tests are themselves classical stepwise rules, worst-case evaluation scales as $O(K^2)$~\citep{dobriban2020fact,karmakar2025bottomup}, and optimality within the class is not claimed.
\emph{(iii)~Graphical / weighted gatekeeping}~\citep{Bretz09,bretz2011graphical,dmitrienko2007gatekeeping,dmitrienko2008multistage,dmitrienko2010book} honors design-imposed block structure and encodes regulatory precedence, but power optimality within the admissible graph is not a design criterion and attainable weights are fixed a priori.
\emph{(iv)~Hierarchical and resampling methods}~\citep{meinshausen2008hierarchical,mandozzi2016hierarchical,WestfallYoung1993} exploit tree structure or empirical dependence, but remain objective-agnostic.

A parallel \emph{objective-driven} line~\citep{RHPA22,dubey25,dubey26esp} maximizes discovery subject to strong-FWER constraints via a dual program with strong duality and is power-optimal, but existing constructive solutions assume three conditions: full exchangeability of the joint $p$-value vector, a single common alternative density $g$ shared across all hypotheses, and a known $g$. Structured multiplicity routinely violates all three.

Existing procedures do not jointly deliver (a) power optimality, (b) block separability, i.e., operational alignment with the designs practitioners deploy, and (c) validity under unknown, heterogeneous alternative densities. Structured multiplicity needs all three: block-level signal strengths differ (endpoints, tissues, doses), the joint law is exchangeable only within blocks, and $g$ must be estimated from data.

\paragraph{Our contributions.}
We give a procedure that delivers all three, for block size three.

\begin{itemize}[leftmargin=1.25em,itemsep=2pt,topsep=2pt]
\item \textbf{Power-optimal block-separable procedure at $O(K)$ cost.}
    \textsc{BOOST} (Algorithm~\ref{alg:blockwise_k3}) maximizes the average-power objective $\Pi_K$ (the ``all-$K$-reject'' probability, defined in Section~\ref{sec:formulation}) within the block-separable class $\mathfrak D_{\mathrm{sep}}$, which contains the fixed-weight-allocation closed-testing and graphical procedures deployed in clinical and genomic pipelines (alpha-propagating schemes that transfer weight across families lie outside $\mathfrak D_{\mathrm{sep}}$), at $O(K)$ evaluation, versus $O(K^2)$ worst case for general closed testing~\citep{dobriban2020fact}. The optimum holds for heterogeneous blocks with block-specific alternative laws.
\item \textbf{Finite-sample FWER with no independence.}
    Theorem~\ref{thm:blockwise_strong_fwer_clean} gives $\mathrm{FWER}\le\alpha$ under only block-measurability and local validity against the true block marginal; no cross- or within-block independence is assumed. Under cross-block independence, a \v{S}id\'ak tightening strictly improves optimal power (Theorem~\ref{thm:global_opt_sep_allocation_sidak}, Corollary~\ref{cor:sidak_blockwise_optimality}).
\item \textbf{Allocation by equalized-marginal KKT.}
    Theorems~\ref{thm:global_opt_sep_within_block}--\ref{thm:global_opt_sep_allocation} characterize the power-optimal allocation across heterogeneous blocks; the optimizer is computable by bisection in $O(B\log(1/\varepsilon))$.
\item \textbf{Plug-in theory for unknown $g$.}
    A sample-split plug-in variant (Algorithm~\ref{alg:plugin_blockwise_k3}) attains $\alpha$-control up to $O(B_{\mathcal T}\,\mathbb E\|g-\widehat g\|_\infty)$ inflation (Theorem~\ref{thm:plugin_fwer}), with per-hypothesis power deficit \emph{independent of $B_{\mathcal T}$} (Theorem~\ref{thm:plugin_power}, Appendix~\ref{app:deferred_main}), addressing limitation (c).
\end{itemize}

Empirically (Section~\ref{sec:simulations}), \textsc{BOOST} leads Bonferroni, Holm, Hochberg, Hommel, and \v{S}id\'ak step-down on $\Pi_K$ across all signal settings, with relative gains widening at moderate signal strengths and at stringent significance levels, the joint regime where confirmatory structured-multiplicity studies operate.

Optimality is established within the block-separable class $\mathfrak D_{\mathrm{sep}}$, the class practitioners already use, and is not claimed against all monotone strong-FWER rules. Block size three is the smallest non-trivial case (block size two reduces to per-marginal Neyman--Pearson, since the only structured configurations are $\vec h_0$ and $\vec h_1$); the framework is modular in the block-level atom (any power-optimal $K$-block solver can be substituted, Section~\ref{subsec:blockwise_method}). The partition is taken as dictated by design; principled partition selection is left for future work.

\section{Preliminaries and Problem Formulation}
\label{sec:formulation}

\subsection{Multiple hypothesis testing model and decision rules}
\label{sec:mht}
For $K$ hypotheses $\{H_k\}_{k=1}^K$ with data $\vec X=(X_1,\ldots,X_K)\in\mathcal X$, a testing policy $\vec D(\vec X)\in\{0,1\}^K$ rejects $H_k$ when $D_k=1$. The truth configuration $h\in\{0,1\}^K$ encodes $h_k=0$ (null) vs.\ $h_k=1$ (alternative), with marginal densities $f_{k,0},f_{k,1}$ and, under independence,
\begin{equation}
\label{eq:independence}
f_h(x)=\prod_{k=1}^K f_{k,h_k}(x_k).
\end{equation}
The structured configurations $h_\ell=(1,\ldots,1,0,\ldots,0)^\top$ with $\ell$ alternatives and $K-\ell$ nulls drive the strong-FWER constraints below.

\paragraph{Power and error functions.}
Let $V=\sum_k(1-h_k)D_k$ be the number of false rejections. We work with \emph{average power} and the per-configuration family-wise error rate
\begin{equation}
\label{eq:pil_final}
\Pi_\ell(\vec{D}) = \frac{1}{\ell}\,\mathbb{E}_{\vec{h}_\ell}\!\left[\sum_{k=1}^{\ell} D_k(\vec{X})\right],\qquad 1 \le \ell \le K,
\end{equation}
\begin{equation}
\label{eq:fwer_final}
\mathrm{FWER}_\ell(\vec D) = \mathbb{P}_{\vec h_\ell}(V>0), \qquad 0\le \ell < K.
\end{equation}
The summation in \eqref{eq:pil_final} runs over the first $\ell$ coordinates because the canonical configuration $\vec h_\ell$ places its alternatives there; under $\vec h$-exchangeability of the joint law (asserted later in this subsection), this is without loss of generality.
Strong FWER control at level $\alpha$ asks $\mathrm{FWER}_\ell\le\alpha$ for all $\ell\in\{0,\ldots,K-1\}$.
The objective-driven multiple-testing problem is then
\begin{equation}
\label{eq:general_objective_driven}
\max_{\vec D}\ \Pi(\vec{D})
\qquad \text{s.t.}\qquad
\mathrm{FWER}_\ell(\vec{D})\le \alpha,\ \ \ell=0,1,\dots,K-1,
\end{equation}
where $\Pi(\vec{D})$ is a chosen power functional (e.g., $\Pi_K(\vec D)$ or $\Pi_{\mathrm{any}}(\vec D)$).
Choosing the average power objective $\Pi_3(\vec{D})$ and the strong-FWER constraints $\mathrm{FWER}_\ell(\vec{D})\le \alpha$ leads to the central $K=3$ optimization problem studied in~\citet{dubey25}:
\begin{equation} \label{eq:MHTK3} \begin{aligned} \max_{\vec{D} \in \{0,1\}^3}\ \ & \Pi_{3}(\vec{D}) \\ \text{s.t.}\ \ & \mathrm{FWER}_{\ell}(\vec{D})\leq \alpha, \qquad 0 \leq \ell\leq 2. \end{aligned} \end{equation}
For $K>3$, directly solving \eqref{eq:general_objective_driven} over all monotone procedures is in general computationally infeasible.
Our strategy is to build a scalable $K>3$ procedure using an exact $K=3$ strong-FWER optimizer as a building block.

Under standard $\vec h$-exchangeability and arrangement-increasing assumptions~\citep{RHPA22}, we restrict attention without loss of generality to symmetric, likelihood-ratio-ordered rules and work in this class throughout. Write $u_k=F_{0,k}(X_k)$ for the $p$-value, with $u_k\sim\mathrm{Unif}(0,1)$ under $H_{0k}$ and $u_k\sim g$ under $H_{Ak}$ for a common non-increasing alternative density $g$~\citep[cf.~eq.~\eqref{eq:pvalue_model}]{RHPA22}; Appendix~\ref{app:k3_full} gives the linear-functional representation of $\Pi_K$ and $\mathrm{FWER}_\ell$.

\subsection{Target problem: block-separable strong-FWER design}
\label{sec:target_problem_generalK}

The unrestricted global program is the strong-FWER analogue of \eqref{eq:general_objective_driven} with the global average power objective:
\begin{equation}
\label{eq:MHTK}
\begin{aligned}
\max_{\vec D\in\mathcal D_K}\quad & \Pi_K(\vec D)\\
\text{s.t.}\quad & \mathrm{FWER}_\ell(\vec{D})\le \alpha,\qquad \ell=0,1,\dots,K-1.
\end{aligned}
\end{equation}
We target \eqref{eq:MHTK} within the \emph{block-separable} class $\mathfrak D_{\mathrm{sep}}$ (Definition~\ref{def:block_separable}), under which each block is tested using only its own data and the global budget is apportioned across blocks. The next section presents the procedure, reusing the~\citet{dubey25} $K=3$ solver as the per-block atom.

\section{Main Result}
\label{sec:main}

The design of a strong-FWER procedure for structured multiplicity must answer three coupled questions: \emph{(i)} how to decide within each block, \emph{(ii)} how to apportion the global error budget across blocks, and \emph{(iii)} how to do so when the alternative density driving within-block power is not known exactly. This section gives a unified answer within the block-separable class.

\subsection{A blockwise computational strategy for finding the optimizer for \texorpdfstring{$K>3$}{K>3}}
\label{subsec:blockwise_method}

Within $\mathfrak D_{\mathrm{sep}}$, the per-block value function $\pi_3^{(b)}$ is attained by any power-optimal $K=3$ strong-FWER solver applied to the per-block law; we use the closed-form solver of~\citet{dubey25} as the block-level atom (the framework is modular: the elementary-symmetric-polynomial solver of~\citet{dubey26esp} for block sizes $\ge 4$ can be substituted without changing the allocation theory). The remaining design choice is the allocation of $\alpha$ across blocks.

Fix $K=3B$ and a partition $\{\mathcal B_b\}_{b=1}^B$ of $[K]$ with $|\mathcal B_b|=3$. Write $u^{(b)}$ for the within-block $p$-values and $u^{(b)}_{(1)}\le u^{(b)}_{(2)}\le u^{(b)}_{(3)}$ for their order statistics. Per-block levels $\alpha_{\mathrm{blk}}^{(b)}$ satisfy $\sum_b\alpha_{\mathrm{blk}}^{(b)}\le\alpha$; a natural choice is $\alpha_{\mathrm{blk}}^{(b)}=\alpha/B$.

Per block, Algorithm~\ref{alg:blockwise_k3} calls \texttt{ComputeOptimalMu} (Algorithm~\ref{alg:compute_optimal_mu_K3_main}) at level $\alpha_{\mathrm{blk}}^{(b)}$ to obtain multipliers $\hat{\vec\mu}^{(b)}$, applies the induced $K=3$ rule $\vec D^{\hat{\vec\mu}^{(b)}}$ to the ordered triple, and unions the within-block rejections. When $\alpha_{\mathrm{blk}}^{(1)}=\cdots=\alpha_{\mathrm{blk}}^{(B)}$ (e.g., the uniform split $\alpha/B$), a single $\hat{\vec\mu}$ is precomputed once and reused across all blocks.

\begin{algorithm}[htbp]
\caption{\textsc{BOOST}: Block-Optimal Objective-driven Strong-FWER Testing}
\label{alg:blockwise_k3}
\begin{algorithmic}[1]
\Statex \textbf{Input:} $p$-values $(u_1,\ldots,u_K)$, $K=3B$; partition $\{\mathcal B_b\}_{b=1}^B$; levels $\alpha_{\mathrm{blk}}^{(b)}$; solver tolerances $(\delta,\varepsilon,T_{\max},U_{\max})$.
\State $\mathcal R\gets\emptyset$.
\For{$b=1$ to $B$}
\State Sort $u^{(b)}=(u_i)_{i\in\mathcal B_b}$ to $u^{(b)}_{(1)}\le u^{(b)}_{(2)}\le u^{(b)}_{(3)}$.
\State $\hat{\vec\mu}^{(b)}\gets\texttt{ComputeOptimalMu}(\alpha_{\mathrm{blk}}^{(b)},\delta,\varepsilon,T_{\max},U_{\max})$ \Comment{Algorithm~\ref{alg:compute_optimal_mu_K3_main}}
\State $R^{(b)}\gets\sum_{i=1}^3 D_i^{\hat{\vec\mu}^{(b)}}(u^{(b)}_{(1)},u^{(b)}_{(2)},u^{(b)}_{(3)})$ via \eqref{eq:dimus2}.
\State Add to $\mathcal R$ the indices in $\mathcal B_b$ corresponding to the $R^{(b)}$ smallest $p$-values.
\EndFor
\State \textbf{Output:} $\mathcal R$.
\end{algorithmic}
\end{algorithm}

With equal per-block levels, Algorithm~\ref{alg:blockwise_k3} has $O(K)$ evaluation cost after a one-time $K=3$ precompute.

\subsection{Finite-sample strong-FWER validity}
\label{subsec:theorems}

We establish finite-sample strong-FWER validity of Algorithm~\ref{alg:blockwise_k3} under only block-measurability and local validity against the block marginal; no independence is required.

Fix a partition $\{\mathcal B_1,\ldots,\mathcal B_B\}$ of $[K]$ with $|\mathcal B_b|=3$ for all $b\in[B]$,
and per-block levels $\alpha_{\mathrm{blk}}^{(1)},\ldots,\alpha_{\mathrm{blk}}^{(B)}\in(0,1)$.
For each block $b\in[B]$, fix an ordering $\mathcal B_b=\{k_{b,1},k_{b,2},k_{b,3}\}$ and write
\[
X^{(b)} := (X_{k_{b,1}},X_{k_{b,2}},X_{k_{b,3}})\in\mathcal X^{(b)}.
\]
Let $\vec D^{(b)}:\mathcal X^{(b)}\to\{0,1\}^3$ denote the local decision rule implemented by
Algorithm~\ref{alg:blockwise_k3} within block $b$.

For any within-block configuration $\eta=(\eta_1,\eta_2,\eta_3)\in\{0,1\}^3$, define the within-block product law
$\mathbb P_{b,\eta}$ as the probability measure on $\mathcal X^{(b)}$ with density
\begin{equation}
\label{eq:block_product_law_density}
f_{b,\eta}(x^{(b)}) := \prod_{i=1}^3 f_{k_{b,i},\eta_i}(x_{k_{b,i}}),
\qquad x^{(b)}\in\mathcal X^{(b)}.
\end{equation}
Define the within-block number of false rejections under configuration $\eta$ by
\begin{equation}
\label{eq:Vb_def_clean}
V^{(b)}(X^{(b)};\eta)
\;:=\;
\sum_{i=1}^3 (1-\eta_i)\,D^{(b)}_i(X^{(b)}).
\end{equation}

The first ingredient is structural: each block's decision must use only that block's data, a property that holds by construction for Algorithm~\ref{alg:blockwise_k3}.

\begin{assumption}[Block-measurability of the local rule]
\label{as:block_measurability}
For each $b\in[B]$, there exists a measurable map $\phi_b:\mathcal X^{(b)}\to\{0,1\}^3$ such that
\begin{equation}
\label{eq:block_measurability}
\vec D^{(b)}(\vec X)=\phi_b\!\left(X^{(b)}\right)\qquad\text{for all }\vec X\in\mathcal X,
\end{equation}
i.e., $\vec D^{(b)}$ is $\sigma(X^{(b)})$-measurable.
\end{assumption}

The second ingredient is a local strong-FWER guarantee stated against the \emph{actual} marginal law of the block data under the global model, rather than against a product-law ideal. For $b\in[B]$ and $\ell\in\{0,1,\dots,K-1\}$, let
\begin{equation}
\label{eq:block_marginal_law}
\mathbb Q^{(b)}_{h_\ell}\;:=\;\mathcal L_{h_\ell}\!\bigl(X^{(b)}\bigr)
\end{equation}
denote the marginal distribution of $X^{(b)}$ induced by $\vec X\sim\mathbb P_{h_\ell}$, and write
\[
\eta^{(b)}(\ell)
\;:=\;
\bigl((h_\ell)_{k_{b,1}},(h_\ell)_{k_{b,2}},(h_\ell)_{k_{b,3}}\bigr)\in\{0,1\}^3
\]
for the within-block configuration induced by $h_\ell$.

\begin{assumption}[Local strong-FWER validity against the block marginal]
\label{as:local_validity_marginal}
For each $b\in[B]$ and each $\ell\in\{0,1,\dots,K-1\}$ such that $\eta^{(b)}(\ell)$ has at least one true null (i.e., $\sum_{i=1}^{3}\eta^{(b)}(\ell)_i \le 2$),
\begin{equation}
\label{eq:block_strong_fwer_assump_clean}
\mathbb Q^{(b)}_{h_\ell}\!\Bigl(V^{(b)}\bigl(X^{(b)};\eta^{(b)}(\ell)\bigr)>0\Bigr)\;\le\; \alpha_{\mathrm{blk}}^{(b)}.
\end{equation}
\end{assumption}

Assumption~\ref{as:local_validity_marginal} imposes no condition on the global model beyond what it forces on the marginal of each block: whatever law $X^{(b)}$ actually has under $h_\ell$, the local rule must control its own strong FWER at level $\alpha_{\mathrm{blk}}^{(b)}$ against that law. In particular, it does \emph{not} require the block marginal to factorize, either within a block or across blocks.

These two assumptions, combined with a Bonferroni budget across blocks, deliver global strong-FWER control with no independence requirement on the joint law.

\begin{theorem}[Strong-FWER control of the blockwise procedure]
\label{thm:blockwise_strong_fwer_clean}
Suppose Assumptions~\ref{as:block_measurability} and \ref{as:local_validity_marginal} hold, and that the per-block levels satisfy
\[
\sum_{b=1}^B \alpha_{\mathrm{blk}}^{(b)}\;\le\; \alpha.
\]
Let $\vec D^{\mathrm{blk}}:\mathcal X\to\{0,1\}^K$ denote the global decision rule produced by Algorithm~\ref{alg:blockwise_k3}, and let $\mathcal R(\vec X):=\{k\in[K]:D^{\mathrm{blk}}_k(\vec X)=1\}$ be its rejection set. Then, for every $\ell\in\{0,1,\ldots,K-1\}$,
\[
\mathrm{FWER}_\ell(\vec D^{\mathrm{blk}})
\;=\;
\mathbb P_{h_\ell}\!\left(V>0\right)
\;\le\;
\alpha,
\]
so Algorithm~\ref{alg:blockwise_k3} is feasible for \eqref{eq:MHTK}; the choice $\alpha_{\mathrm{blk}}^{(b)}=\alpha/B$ saturates the budget.
\end{theorem}

The proof is provided in Appendix~\ref{app:proof_blockwise_strong_fwer_clean}.

Independence enters only to verify Assumption~\ref{as:local_validity_marginal} in concrete models: under \eqref{eq:independence} the $K=3$ optimizer of~\citet{dubey25} discharges it, with weakenings (cross-block-only independence, $\vec h$-exchangeability reductions) in Appendix~\ref{app:thm1_scope_instantiations}.

\subsection{Plug-in estimation of the alternative density}
\label{subsec:plugin}

Algorithm~\ref{alg:blockwise_k3} treats $g$ as known; in practice $g$ is estimated from external or held-out data. Replacing $g$ by a data-independent estimate $\widehat g$ preserves strong-FWER validity up to a $\|g-\widehat g\|_\infty$-driven perturbation, with the sample-splitting instantiation given by Algorithm~\ref{alg:plugin_blockwise_k3}. The plug-in analysis below operates against the within-block product-law block marginal $\mathbb Q^{(b)}_{h_\ell,g}$. This matches the actual block marginal under within-block independence with a common alternative density $g$, which is the standard setting of the canonical model class~\eqref{eq:pvalue_model}.

We adopt the following terminology. For a candidate density $\tilde g$ on $[0,1]$, write $\hat{\vec\mu}(\tilde g;\alpha')$ for the output of Algorithm~\ref{alg:compute_optimal_mu_K3_main} applied with $\tilde g$ as the alternative density at block level $\alpha'$. The solver guarantees that
\begin{equation}
\label{eq:solver_guarantee_against_ghat}
\mathbb Q^{(b)}_{\!h_\ell,\tilde g}\!\Bigl(V^{(b)}\!\bigl(X^{(b)};\eta^{(b)}(\ell)\bigr)>0\Bigr)\;\le\;\alpha',
\qquad \ell=0,\dots,K-1,
\end{equation}
where $\mathbb Q^{(b)}_{h_\ell,\tilde g}$ denotes the per-block law built with $\tilde g$ in place of $g$. The \emph{plug-in blockwise procedure} computes $\hat{\vec\mu}^{(b)}=\hat{\vec\mu}(\widehat g;\alpha_{\mathrm{blk}}^{(b)})$ with a data-independent $\widehat g$ and then proceeds exactly as Algorithm~\ref{alg:blockwise_k3}.

The analysis rests on three perturbation lemmas with constants $L_3=2M$ and $L_3^{\mathrm{pow}}=3M^2$, $M:=\max(\|g\|_\infty,\|\widehat g\|_\infty)$, stated in Appendix~\ref{app:plugin_lemmas}.

The FWER side first: an additive inflation linear in $B_{\mathcal T}$ and $\mathbb E\|g-\widehat g\|_\infty$.

\begin{theorem}[Plug-in blockwise procedure: approximate strong-FWER validity]
\label{thm:plugin_fwer}
Assume Assumption~\ref{as:block_measurability} and fix a partition of the blocks $[B]=\mathcal E\sqcup\mathcal T$, $|\mathcal T|=B_{\mathcal T}\ge 1$. Suppose the following hold:
\begin{enumerate}
\item[(i)] $\widehat g$ is a $[0,1]$-valued density function that is measurable with respect to a $\sigma$-algebra $\mathcal G$ independent of $(X^{(b)})_{b\in\mathcal T}$.
\item[(ii)] Per-block levels $\{\alpha_{\mathrm{blk}}^{(b)}\}_{b\in\mathcal T}$ satisfy the Bonferroni budget $\sum_{b\in\mathcal T}\alpha_{\mathrm{blk}}^{(b)}\le\alpha$.
\item[(iii)] For each $b\in\mathcal T$, $\hat{\vec\mu}^{(b)}:=\hat{\vec\mu}(\widehat g;\alpha_{\mathrm{blk}}^{(b)})$ from Algorithm~\ref{alg:compute_optimal_mu_K3_main} applied with $\widehat g$, and $\vec D^{(b)}$ is the induced block rule of Algorithm~\ref{alg:blockwise_k3}.
\item[(iv)] For every $b\in\mathcal T$ and every $\ell\in\{0,\ldots,K\}$, the within-block law $\mathbb Q^{(b)}_{h_\ell,g}$ factorizes as the within-block product law with marginal $g$ on alternative coordinates and uniform on null coordinates (within-block independence with common alternative density $g$; automatic under~\eqref{eq:pvalue_model}). The case $\ell=K$ specializes to $g^{\otimes 3}$ and is what Theorem~\ref{thm:plugin_power} uses; the cases $\ell\in\{0,\ldots,K-1\}$ are what the FWER bound below uses.
\end{enumerate}
Let $\vec D^{\mathrm{blk},\mathcal T}$ denote the plug-in blockwise rule restricted to blocks in $\mathcal T$. Then, for every $\ell\in\{0,\dots,K-1\}$,
\begin{equation}
\label{eq:plugin_fwer_bound}
\mathrm{FWER}_\ell(\vec D^{\mathrm{blk},\mathcal T})
\;\le\;\alpha\;+\;L_3\,B_{\mathcal T}\,\mathbb E\!\left[\|g-\widehat g\|_{\infty}\right],
\end{equation}
with $L_3=2M$ and $M:=\max(\|g\|_\infty,\|\widehat g\|_\infty)$ as in Lemma~\ref{lem:fwer_perturbation}. When $\widehat g$ is deterministic (e.g., pre-specified from external data), $\mathbb E\|g-\widehat g\|_\infty$ is the deterministic distance $\|g-\widehat g\|_\infty$. When $\widehat g$ is random, $M$ is random and $L_3\,\mathbb E\|g-\widehat g\|_\infty$ is read as $2\,\mathbb E[M\,\|g-\widehat g\|_\infty]$; a $\mathcal G$-measurable uniform upper bound $\|\widehat g\|_\infty\le \overline M$ a.s.\ (automatic for Grenander estimators of bounded $g$) reduces this to the stated form with $L_3=2\max(\|g\|_\infty,\overline M)$.
\end{theorem}

The proof is provided in Appendix~\ref{app:proof_plugin_fwer}.

The cleanest way to obtain a $\widehat g$ independent of the testing fold is to fit it on a held-out estimation fold; under cross-block independence, a deterministic estimation/testing split delivers a $\widehat g$ satisfying condition~(i) of Theorem~\ref{thm:plugin_fwer} (Corollary~\ref{cor:plugin_sample_splitting} in Appendix~\ref{app:plugin_corollaries}; the bound is unchanged under the \v{S}id\'ak budget).

The power side is asymmetric: per-hypothesis power loss does not pick up the $B_{\mathcal T}$ factor.

\begin{theorem}[Plug-in blockwise procedure: realized-power lower bound]
\label{thm:plugin_power}
Under the setup of Theorem~\ref{thm:plugin_fwer}, suppose additionally:
\begin{enumerate}
\item[(v)] All $K=3B$ hypotheses on the testing fold are alternatives, i.e., the within-block configuration is $\eta^{(b)}=(1,1,1)$ for every $b\in\mathcal T$.
\end{enumerate}
Let
\[
\mathrm{Power}_{g}^{(b)}\!\bigl(\widehat{\vec D}^{(b)}\bigr)
\;:=\;\Pi_{3}^{g}\!\bigl(\widehat{\vec D}^{(b)}\bigr)
\;=\;\tfrac{1}{3}\,\mathbb E_{g}\!\bigl[\textstyle\sum_{i=1}^{3}D_{i}^{\hat{\vec\mu}^{(b)}}\!\bigl(U^{(b)}\bigr)\bigr]
\]
denote the per-hypothesis average power of the plug-in rule on block $b$ under the within-block product law $U^{(b)}\sim g^{\otimes 3}$, normalized as in \eqref{eq:pil_final}. By Theorem~\ref{thm:plugin_fwer}(iv) and~(v), the testing-fold block law under $\vec h_K$ is $g^{\otimes 3}$, so $\mathrm{Power}_{g}^{(b)}$ equals the realized per-hypothesis power under the global model. Then, for the average per-hypothesis power on the testing fold,
\begin{equation}
\label{eq:plugin_power_bound}
\frac{1}{B_{\mathcal T}}\sum_{b\in\mathcal T}\mathbb E\!\left[\mathrm{Power}_{g}^{(b)}\!\bigl(\widehat{\vec D}^{(b)}\bigr)\right]
\;\ge\;\frac{1}{B_{\mathcal T}}\sum_{b\in\mathcal T}\mathbb E\!\left[\pi_{3}^{\widehat g}\!\bigl(\alpha_{\mathrm{blk}}^{(b)}\bigr)\right]
\;-\;L_{3}^{\mathrm{pow}}\,\mathbb E\!\left[\|g-\widehat g\|_{\infty}\right],
\end{equation}
with $L_{3}^{\mathrm{pow}}=3M^{2}$ as in Lemma~\ref{lem:power_perturbation}. Equivalently, writing $N_{\mathrm{rej}}^{(b)}:=\sum_{i=1}^{3}D_{i}^{\hat{\vec\mu}^{(b)}}\!(U^{(b)})$ for the rejection count on block $b$, the total expected number of correct rejections on the testing fold satisfies
\begin{equation}
\label{eq:plugin_power_bound_sum}
\sum_{b\in\mathcal T}\mathbb E\!\left[N_{\mathrm{rej}}^{(b)}\right]
\;\ge\;3\sum_{b\in\mathcal T}\mathbb E\!\left[\pi_{3}^{\widehat g}\!\bigl(\alpha_{\mathrm{blk}}^{(b)}\bigr)\right]
\;-\;3B_{\mathcal T}\,L_{3}^{\mathrm{pow}}\,\mathbb E\!\left[\|g-\widehat g\|_{\infty}\right].
\end{equation}
\end{theorem}

The proof is provided in Appendix~\ref{app:proof_plugin_power}.

At per-block levels not too small relative to the estimation error, the plug-in power lower bound can be expressed against the \emph{true-$g$} oracle at a slightly deflated level $\alpha_{\mathrm{blk}}^{(b)}-L_3\|g-\widehat g\|_\infty$, with the same $B_{\mathcal T}$-free penalty (Corollary~\ref{cor:plugin_power_oracle} in Appendix~\ref{app:plugin_corollaries}).

\begin{algorithm}[htbp]
\caption{plug-in \textsc{BOOST}: plug-in variant for unknown $g$}
\label{alg:plugin_blockwise_k3}
\begin{algorithmic}[1]
\Statex \textbf{Input:} $p$-values $(u_1,\ldots,u_K)$, $K=3B$; partition $\{\mathcal B_b\}$; fold split $[B]=\mathcal E\sqcup\mathcal T$; density estimator $\widehat g$ (e.g., Grenander); levels $\{\alpha_{\mathrm{blk}}^{(b)}\}_{b\in\mathcal T}$ with $\sum_{b\in\mathcal T}\alpha_{\mathrm{blk}}^{(b)}\le\alpha$ (or \v{S}id\'ak); tolerances $(\delta,\varepsilon,T_{\max},U_{\max})$.
\State $\widehat g\gets\widehat g(\{u_i:i\in\mathcal B_b,\,b\in\mathcal E\})$. \Comment{fit on estimation fold only}
\State $\mathcal R\gets\emptyset$.
\For{$b\in\mathcal T$}
\State Sort $u^{(b)}=(u_i)_{i\in\mathcal B_b}$ to $u^{(b)}_{(1)}\le u^{(b)}_{(2)}\le u^{(b)}_{(3)}$.
\State $\hat{\vec\mu}^{(b)}\gets\texttt{ComputeOptimalMu}(\alpha_{\mathrm{blk}}^{(b)},\delta,\varepsilon,T_{\max},U_{\max};\widehat g)$.
\State $R^{(b)}\gets\sum_{i=1}^3 D_i^{\hat{\vec\mu}^{(b)}}(u^{(b)}_{(1)},u^{(b)}_{(2)},u^{(b)}_{(3)})$ via \eqref{eq:dimus2}.
\State Add to $\mathcal R$ the indices in $\mathcal B_b$ corresponding to the $R^{(b)}$ smallest $p$-values.
\EndFor
\State \textbf{Output:} $\mathcal R\subseteq\bigcup_{b\in\mathcal T}\mathcal B_b$.
\end{algorithmic}
\end{algorithm}

Estimation-fold blocks are not tested; if hypotheses on $\mathcal E$ also need to be tested, re-run with the roles of $\mathcal E$ and $\mathcal T$ swapped at half level.

Specializing to canonical density estimators turns both bounds into explicit polynomial-in-$n$ rates. Under monotone $g$ and the sample-splitting Grenander estimator on $n_{\mathcal E}$ estimation-fold $p$-values, $\mathbb E\|g-\widehat g\|_\infty=O((n_{\mathcal E}^{-1}\log n_{\mathcal E})^{1/3})$~\citep{groeneboom2014nonparametric}, giving FWER excess $O(B_{\mathcal T}(n_{\mathcal E}^{-1}\log n_{\mathcal E})^{1/3})$ and average-power deficit $O((n_{\mathcal E}^{-1}\log n_{\mathcal E})^{1/3})$; under H\"older-$s$ alternatives and a bandwidth-optimal kernel density estimator, both rates become $(n_{\mathcal E}^{-1}\log n_{\mathcal E})^{s/(2s+1)}$~\citep{tsybakov2009book} (Corollary~\ref{cor:plugin_rates} in Appendix~\ref{app:plugin_corollaries}). The procedure is asymptotically $\alpha$-level exact whenever $B_{\mathcal T}\,r_{n_{\mathcal E}}\to 0$.

\subsection{Power-optimal allocation within $\mathfrak D_{\mathrm{sep}}$}
\label{subsec:allocation}

We state the optimality result within the block-separable class $\mathfrak D_{\mathrm{sep}}$ (Definition~\ref{def:block_separable}), the operational form of procedures deployed in gatekeeping, multi-endpoint trials, and per-locus genomic tests, which decide within a block using only that block's data and apportion the global budget across blocks. The analysis decouples: Theorem~\ref{thm:global_opt_sep_within_block} upper-bounds each block's contribution by a $K=3$ value function $\pi_3^{(b)}$; Theorem~\ref{thm:global_opt_sep_allocation} characterizes the power-optimal allocation under concavity (Lemma~\ref{lem:pi3_concavity}) via an equalized-marginal KKT condition; Corollary~\ref{cor:blockwise_sep_optimality} combines them. When cross-block independence is additionally assumed, Theorem~\ref{thm:global_opt_sep_allocation_sidak} tightens the Bonferroni budget to the \v{S}id\'ak budget $\prod_b(1-\alpha_b)\ge 1-\alpha$, strictly improving the attainable power.

For each block $b\in[B]$, write $\mathcal B_b=\{k_{b,1},k_{b,2},k_{b,3}\}$ and let
\[
P^{(b)} \;:=\; \bigl(u^{(b)}_{(1)},u^{(b)}_{(2)},u^{(b)}_{(3)}\bigr)\in Q,
\qquad 0\le u^{(b)}_{(1)}\le u^{(b)}_{(2)}\le u^{(b)}_{(3)}\le 1,
\]
denote the order statistics of the block $p$-values $\{u_{k_{b,1}},u_{k_{b,2}},u_{k_{b,3}}\}$ obtained via \eqref{eq:pvalue_def}. Under $\vec h_K$ we have $h_k=1$ for all $k\in[K]$, so the induced within-block configuration is $\eta^{(b)}=(1,1,1)$ for every $b\in[B]$. Write
\begin{equation}
\label{eq:per_block_alt_law}
\mathbb Q^{(b)}_{\!\vec h_K}\;:=\;\mathcal L_{\vec h_K}\!\bigl(P^{(b)}\bigr)
\end{equation}
for the within-block marginal law of $P^{(b)}$ under the global alternative; in general these laws may differ across blocks.

We first define a $K=3$ value function $\pi_3^{\mathbb Q}$ associated with an arbitrary $K=3$ generative law $\mathbb Q$ on $Q$. For a fixed level $\alpha\in(0,1)$, recall that $\Pi_3(\vec{D})$ denotes the average power of a given $K=3$ rule $\vec{D}$ as in \eqref{eq:pil_final}. We define
\begin{equation}
\label{eq:pi3_value_def}
\begin{aligned}
\pi_3^{\mathbb Q}(\alpha)
\;:=\;\sup\bigl\{\Pi_3^{\mathbb Q}(\vec{D}) \,:\, &\vec{D}:Q\to\{0,1\}^3 \text{ measurable, }\\
&\mathrm{FWER}_l^{\mathbb Q}(\vec{D})\le \alpha \text{ for } l=0,1,2\bigr\},
\end{aligned}
\end{equation}
where $\Pi_3^{\mathbb Q}$ and $\mathrm{FWER}_l^{\mathbb Q}$ denote the average power and per-configuration FWER computed under $P\sim\mathbb Q$. We write $\pi_3:=\pi_3^{\mathbb P_{h_3}}$ when $\mathbb Q=\mathbb P_{h_3}$ is the canonical $K=3$ generative law of~\citet{dubey25}, and define the per-block value function as the value function induced by $\mathbb Q^{(b)}_{\!\vec h_K}$:
\begin{equation}
\label{eq:pi3_per_block_def}
\pi_3^{(b)}(\alpha)\;:=\;\pi_3^{\mathbb Q^{(b)}_{\!\vec h_K}}(\alpha),\qquad b\in[B].
\end{equation}
Under any law $\mathbb Q$ for which the $K=3$ FWER--power tradeoff admits the structural conditions of~\citet{dubey25}, the supremum in \eqref{eq:pi3_value_def} is attained by the corresponding $K=3$ strong-FWER optimizer.

The next lemma gives the monotonicity and concavity of $\pi_3^{\mathbb Q}$ used in the allocation analysis.

\begin{lemma}[Monotonicity and concavity of the $K=3$ value function]
\label{lem:pi3_concavity}
For each \emph{within-block alternative count} $l\in\{0,1,2,3\}$, let $\mathbb Q_l$ denote the joint law on $Q$ of the within-block $p$-value vector under any configuration with exactly $l$ within-block alternatives (and $3-l$ within-block nulls), with $\mathbb Q_3=\mathbb Q$ the complete-alternative law. Suppose each $\mathbb Q_l$ is absolutely continuous with respect to Lebesgue measure on $[0,1]^3$ and invariant under coordinate permutations of the within-block-alternative indices (i.e., $\vec h$-exchangeable within $\mathbb Q_l$). Then $\pi_3^{\mathbb Q}:[0,1]\to[0,1]$ is non-decreasing and concave. In particular, under the model class \eqref{eq:pvalue_model}, every per-block configuration law has a Lebesgue density and is exchangeable within the alternative indices, so each per-block value function $\pi_3^{(b)}$ in \eqref{eq:pi3_per_block_def} is non-decreasing and concave on $[0,1]$.
\end{lemma}

The proof is provided in Appendix~\ref{app:proof_pi3_concavity}.

We name the concavity hypothesis used in Theorems~\ref{thm:global_opt_sep_allocation}--\ref{thm:global_opt_sep_allocation_sidak} below. Cross-block homogeneity is used only to specialize the heterogeneous results to the uniform split $\alpha/B$; it is deferred to Appendix~\ref{app:deferred_main} (Assumption~\ref{assump:exchangeable_blocks}).

\begin{assumption}[Concavity of the per-block optimal power curves]
\label{assump:concavity}
Each per-block value function $\pi_3^{(b)}:[0,1]\to[0,1]$ defined in \eqref{eq:pi3_per_block_def} is concave on $[0,1]$. By Lemma~\ref{lem:pi3_concavity}, this assumption is satisfied automatically whenever each within-block configuration law is absolutely continuous and exchangeable within the alternative indices, in particular under \eqref{eq:pvalue_model}; we retain it as a named hypothesis only for use in Theorem~\ref{thm:global_opt_sep_allocation} when working outside this regularity class.
\end{assumption}

\begin{definition}[Block-separable procedures]
\label{def:block_separable}
Let $\mathfrak{D}_{\mathrm{sep}}$ denote the class of procedures that (i) choose a non-negative allocation $(\alpha_1,\ldots,\alpha_B)$ with $\sum_{b=1}^B \alpha_b\le \alpha$ and (ii) for each block $b\in[B]$, apply a within-block rule $\vec D^{(b)}$ that depends only on $P^{(b)}$ and is feasible for the $K=3$ strong-FWER problem at level $\alpha_b$ defining $\pi_3^{(b)}(\alpha_b)$ in \eqref{eq:pi3_per_block_def}, i.e., its within-block configuration FWERs are controlled at level $\alpha_b$ under the family $\{\mathbb Q^{(b)}_{\eta}:\eta\in\{0,1\}^3, \sum_i\eta_i\le 2\}$.
\end{definition}

For a fixed allocation, the within-block subproblem decouples: the global average power factors into the sum of per-block average powers, each upper-bounded by its $K=3$ value function.

\begin{theorem}[Within-block optimality]
\label{thm:global_opt_sep_within_block}
Fix any allocation $(\alpha_1,\ldots,\alpha_B)$ with $\sum_{b=1}^B\alpha_b\le\alpha$. For every $D^{\mathrm{sep}}\in\mathfrak D_{\mathrm{sep}}$ that uses this allocation,
\begin{equation}
\label{eq:thm_within_block_bound}
\Pi_K\!\bigl(D^{\mathrm{sep}}\bigr)
\;=\;\frac{1}{B}\sum_{b=1}^B \Pi_3^{\mathbb Q^{(b)}_{\!\vec h_K}}\!\bigl(\vec D^{(b)}\bigr)
\;\le\;\frac{1}{B}\sum_{b=1}^B \pi_3^{(b)}(\alpha_b),
\end{equation}
and the upper bound is attained when each within-block rule $\vec D^{(b)}$ attains $\pi_3^{(b)}(\alpha_b)$ in \eqref{eq:pi3_value_def}; in particular, when $\vec D^{(b)}$ is the $K=3$ strong-FWER optimizer of~\citet{dubey25} (Algorithm~\ref{alg:compute_optimal_mu_K3_main}) at level $\alpha_b$, applied to the per-block law $\mathbb Q^{(b)}_{\!\vec h_K}$. Under cross-block homogeneity ($\pi_3^{(b)}\equiv\pi_3$), the bound reduces to $\frac1B\sum_b\pi_3(\alpha_b)$.
\end{theorem}

The proof is provided in Appendix~\ref{app:proof_global_opt_sep_within_block}.

With the within-block ceilings in hand, the outer problem is to allocate $\alpha$ across $B$ concave one-dimensional gains. Concavity forces every active block to share a common marginal-power slope at the optimum, the multi-block analogue of the equimarginal principle.

\begin{theorem}[Allocation optimality under concavity: equalized-marginal characterization]
\label{thm:global_opt_sep_allocation}
Assume Assumption~\ref{assump:concavity} (in particular, this holds under \eqref{eq:pvalue_model} by Lemma~\ref{lem:pi3_concavity}). Let $\partial_+\pi_3^{(b)}$ denote the (non-negative, non-increasing) right derivative of the concave, non-decreasing $\pi_3^{(b)}$ on $[0,\alpha]$. Among all allocations $(\alpha_1,\ldots,\alpha_B)\in[0,\alpha]^B$ with $\sum_{b=1}^B\alpha_b\le\alpha$, the maximum
\begin{equation}
\label{eq:hetero_alloc_max}
M(\alpha;\,\pi_3^{(1)},\ldots,\pi_3^{(B)})
\;:=\;
\max_{\substack{\alpha_b\ge 0,\\ \sum_b\alpha_b\le\alpha}}\;\frac{1}{B}\sum_{b=1}^B\pi_3^{(b)}(\alpha_b)
\end{equation}
is attained, and any maximizer $\vec\alpha^{\,*}=(\alpha_1^{*},\ldots,\alpha_B^{*})$ satisfies the equalized-marginal (KKT) condition: there exists $\mu^{*}\ge 0$ such that for all $b\in[B]$,
\begin{equation}
\label{eq:KKT_condition}
\partial_+\pi_3^{(b)}(\alpha_b^{*})\;=\;\mu^{*}\;\text{ if }\alpha_b^{*}\in(0,\alpha),\qquad
\partial_+\pi_3^{(b)}(0^{+})\;\le\;\mu^{*}\;\text{ if }\alpha_b^{*}=0,
\end{equation}
together with $\sum_{b=1}^{B}\alpha_b^{*}=\alpha$ (a binding-budget maximizer exists by monotonicity of each $\pi_3^{(b)}$; when every active marginal is zero, $\mu^{*}=0$ and any extension to a binding-budget allocation is also a maximizer). Under cross-block homogeneity ($\pi_3^{(b)}\equiv\pi_3$), \eqref{eq:KKT_condition} forces $\alpha_b^{*}=\alpha/B$ for every $b$, and \eqref{eq:hetero_alloc_max} reduces to
\begin{equation}
\label{eq:uniform_is_optimal}
M(\alpha;\,\pi_3,\ldots,\pi_3)\;=\;\pi_3(\alpha/B).
\end{equation}
\end{theorem}

The proof is provided in Appendix~\ref{app:proof_global_opt_sep_allocation}.

Under cross-block independence, the Bonferroni budget relaxes to the larger \v{S}id\'ak budget $\prod_b(1-\alpha_b)\ge 1-\alpha$, strictly improving the optimum.

\begin{theorem}[Allocation under cross-block independence: \v{S}id\'ak tightening]
\label{thm:global_opt_sep_allocation_sidak}
Suppose the $p$-value vector $P=(P^{(1)},\ldots,P^{(B)})$ is \emph{cross-block independent}: for every configuration $h\in\{0,1\}^K$ the block vectors $P^{(1)},\ldots,P^{(B)}$ are mutually independent under $\mathcal L_h(P)$. Define the enlarged block-separable class
\begin{equation}
\label{eq:sidak_class}
\mathfrak D_{\mathrm{sep}}^{\mathrm{ind}}
\;:=\;\Bigl\{D^{\mathrm{sep}}:\text{within-block rules as in Definition~\ref{def:block_separable} with }\textstyle\prod_{b=1}^{B}(1-\alpha_b)\ge 1-\alpha\Bigr\},
\end{equation}
so that $\mathfrak D_{\mathrm{sep}}\subseteq\mathfrak D_{\mathrm{sep}}^{\mathrm{ind}}$ (every Bonferroni allocation satisfies the \v{S}id\'ak constraint). Every $D^{\mathrm{sep}}\in\mathfrak D_{\mathrm{sep}}^{\mathrm{ind}}$ is feasible for the global problem \eqref{eq:MHTK}. Under Assumption~\ref{assump:concavity}, the maximum
\begin{equation}
\label{eq:sidak_alloc_max}
M^{\mathrm{ind}}\!\bigl(\alpha;\pi_3^{(1)},\ldots,\pi_3^{(B)}\bigr)
\;:=\;
\max_{\substack{\alpha_b\in[0,1]\\ \prod_{b=1}^{B}(1-\alpha_b)\ge 1-\alpha}}
\;\frac{1}{B}\sum_{b=1}^{B}\pi_3^{(b)}(\alpha_b)
\end{equation}
is attained, and any maximizer $\vec\alpha^{*}=(\alpha_1^{*},\ldots,\alpha_B^{*})$ satisfies the \emph{weighted} equalized-marginal KKT condition: there exists $\mu^{*}\ge 0$ such that for every $b\in[B]$ with $\alpha_b^{*}\in(0,1)$,
\begin{equation}
\label{eq:sidak_KKT}
(1-\alpha_b^{*})\,\partial_+\pi_3^{(b)}(\alpha_b^{*})\;=\;B\mu^{*},
\end{equation}
together with $\prod_{b=1}^{B}(1-\alpha_b^{*})=1-\alpha$ (the \v{S}id\'ak constraint binds). Under cross-block homogeneity ($\pi_3^{(b)}\equiv\pi_3$), the optimum is the uniform \v{S}id\'ak split $\alpha_b^{*}\equiv 1-(1-\alpha)^{1/B}$, and
\begin{equation}
\label{eq:sidak_uniform_is_optimal}
M^{\mathrm{ind}}\!\bigl(\alpha;\pi_3,\ldots,\pi_3\bigr)
\;=\;\pi_3\!\bigl(1-(1-\alpha)^{1/B}\bigr)
\;\ge\;\pi_3(\alpha/B)
\;=\;M\!\bigl(\alpha;\pi_3,\ldots,\pi_3\bigr),
\end{equation}
with strict inequality whenever $B>1$, $\alpha\in(0,1)$, and $\pi_3$ is strictly increasing on $[\alpha/B,\,1-(1-\alpha)^{1/B}]$.
\end{theorem}

The proof is provided in Appendix~\ref{app:proof_global_opt_sep_allocation_sidak}.

The equalized-marginal characterization reduces the $B$-dimensional constrained search to a univariate bisection on the shared Lagrange multiplier $\mu^{*}$, with right-inverses $g_b(\mu):=\sup\{\alpha'\in[0,\alpha]:\partial_+\pi_3^{(b)}(\alpha')\ge\mu\}$ giving the per-block allocations $\alpha_b^{*}=g_b(\mu^{*})$. Bisection on $\mu\in[0,\overline\mu]$ with $\overline\mu:=\max_b\partial_+\pi_3^{(b)}(0^{+})<\infty$ (finiteness automatic when $\|g\|_\infty<\infty$, cf.~Prop.~\ref{prop:canonical_model_verification}(iv)) converges to precision $\varepsilon$ in $\lceil\log_2(\overline\mu/\varepsilon)\rceil$ outer iterations with $O(B)$ per-iteration cost, total complexity $O(B\log(1/\varepsilon))$ (Proposition~\ref{prop:bisection_allocation} in Appendix~\ref{app:bisection_allocation_statement}).

Combining the within-block ceiling of Theorem~\ref{thm:global_opt_sep_within_block} with the equalized-marginal allocator of Theorem~\ref{thm:global_opt_sep_allocation} characterizes the global $\mathfrak D_{\mathrm{sep}}$-optimum and the procedure that achieves it.

\begin{corollary}[Global optimality within $\mathfrak D_{\mathrm{sep}}$]
\label{cor:blockwise_sep_optimality}
Under Assumption~\ref{assump:concavity},
\begin{equation}
\label{eq:hetero_sup_value}
\sup_{D^{\mathrm{sep}}\in\mathfrak D_{\mathrm{sep}}}\Pi_K\!\bigl(D^{\mathrm{sep}}\bigr)
\;=\;
\frac{1}{B}\sum_{b=1}^{B}\pi_3^{(b)}(\alpha_b^{*})
\;=:\;
\Pi_K^{*},
\end{equation}
where $\vec\alpha^{*}$ is the equalized-marginal allocation of Theorem~\ref{thm:global_opt_sep_allocation}, computable by the bisection of Proposition~\ref{prop:bisection_allocation}. Provided each per-block law $\mathbb Q^{(b)}_{\!\vec h_K}$ satisfies the dubey25 regularity conditions (Assumptions~\ref{as:assumption3}--\ref{as:assumption5}; automatic in the canonical regime~\eqref{eq:pvalue_model}, cf.~Prop.~\ref{prop:canonical_model_verification}), the supremum is attained by the procedure that uses allocation $\vec\alpha^{*}$ together with the $K=3$ strong-FWER optimizer of~\citet{dubey25} applied at level $\alpha_b^{*}$ to $\mathbb Q^{(b)}_{\!\vec h_K}$ in each block $b\in[B]$.

Under cross-block homogeneity ($\pi_3^{(b)}\equiv\pi_3$, formalized as Assumption~\ref{assump:exchangeable_blocks} in Appendix~\ref{app:deferred_main}), $\Pi_K^{*}=\pi_3(\alpha/B)$ and the optimal allocation reduces to the uniform split $\alpha_{\mathrm{blk}}^{(b)}=\alpha/B$, recovering Algorithm~\ref{alg:blockwise_k3} as currently stated. Feasibility for \eqref{eq:MHTK} follows in both cases from Theorem~\ref{thm:blockwise_strong_fwer_clean} via Remark~\ref{cor:k3_optimizer_instantiation}.
\end{corollary}

The proof is provided in Appendix~\ref{app:proof_blockwise_sep_optimality}.

\begin{remark}[Extensions and structural verification in Appendix~\ref{app:deferred_main}]
\label{rem:extensions_deferred}
Under cross-block independence, a \v{S}id\'ak tightening strictly improves Cor.~\ref{cor:blockwise_sep_optimality} (Appendix~\ref{app:deferred_main}, Cor.~\ref{cor:sidak_blockwise_optimality}); the attainable value is sandwiched between the global \v{S}id\'ak floor and the per-marginal Neyman--Pearson ceiling (Prop.~\ref{prop:containment}); and the canonical truncated-normal model discharges every structural hypothesis used above (Prop.~\ref{prop:canonical_model_verification}).
\end{remark}

\subsubsection*{Pipeline schematic}
\label{subsec:boost_illustration}

Figure~\ref{fig:boost_pipeline} summarizes the structure of \textsc{BOOST}: the $K=3$ atom appears as both value-function oracle and decision rule, with the outer KKT bisection on $\mu^{\ast}$ tying them together.

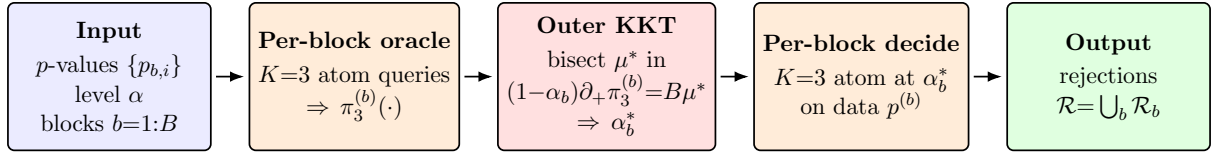
\begin{figure}[!htbp]
\centering
\resizebox{\linewidth}{!}{%
\begin{tikzpicture}[
  every node/.style={font=\small},
  box/.style={rectangle, rounded corners=3pt, draw, thick,
              minimum width=33mm, minimum height=24mm, align=center,
              inner sep=4pt},
  inbox/.style={box, fill=blue!8},
  atombox/.style={box, fill=orange!15},
  outbox/.style={box, fill=red!12},
  finalbox/.style={box, fill=green!12},
  arr/.style={-Latex, thick, shorten >=2pt, shorten <=2pt}
]
\node[inbox] (input) at (0,0)
  {\textbf{Input}\\[2pt] $p$-values $\{p_{b,i}\}$\\ level $\alpha$\\ blocks $b{=}1{:}B$};
\node[atombox, right=6mm of input] (curves)
  {\textbf{Per-block oracle}\\[2pt] $K{=}3$ atom queries\\ $\Rightarrow\,\pi_3^{(b)}(\cdot)$};
\node[outbox, right=6mm of curves] (kkt)
  {\textbf{Outer KKT}\\[2pt] bisect $\mu^\ast$ in\\ $(1{-}\alpha_b)\partial_+\pi_3^{(b)}{=}B\mu^\ast$\\ $\Rightarrow\,\alpha_b^\ast$};
\node[atombox, right=6mm of kkt] (apply)
  {\textbf{Per-block decide}\\[2pt] $K{=}3$ atom at $\alpha_b^\ast$\\ on data $p^{(b)}$};
\node[finalbox, right=6mm of apply] (out)
  {\textbf{Output}\\[2pt] rejections\\ $\mathcal R{=}\bigcup_b\mathcal R_b$};
\draw[arr] (input) -- (curves);
\draw[arr] (curves) -- (kkt);
\draw[arr] (kkt) -- (apply);
\draw[arr] (apply) -- (out);
\end{tikzpicture}%
}
\caption{\textsc{BOOST} pipeline. The $K=3$ optimizer of~\citet{dubey25} (orange boxes) appears twice: as a value-function oracle producing the concave curves $\pi_3^{(b)}(\cdot)$, and as a decision rule applied at the optimized allocation $\alpha_b^\ast$. The outer stage (red) solves the equalized-marginal KKT system of Theorem~\ref{thm:global_opt_sep_allocation} by bisection on the shared Lagrange multiplier $\mu^\ast$ (Proposition~\ref{prop:bisection_allocation}), reducing a $B$-dimensional constrained optimization to a single univariate search. The KKT box shows the \v{S}id\'ak form (Theorem~\ref{thm:global_opt_sep_allocation_sidak}); the Bonferroni form drops the $(1-\alpha_b)$ factor and uses $\sum_b\alpha_b\le\alpha$.}
\label{fig:boost_pipeline}
\end{figure}

\section{Simulations}
\label{sec:simulations}

We evaluate \textsc{BOOST} (Algorithm~\ref{alg:blockwise_k3}) at $\alpha=0.05$, $\alpha_{\mathrm{blk}}=\alpha/B$, against ten strong-FWER baselines spanning five categories (stepwise, dependence-aware stepwise, graphical-gatekeeping, combination closure, resampling); baselines and experimental details are in Appendix~\ref{app:sim_baselines}. Monte-Carlo estimates use $3\times 10^4$ replicates unless noted (MC-SE $\approx 0.002$ for power, $\approx 0.001$ for FWER).

\subsection{Comparison across $p$-value families}
\label{subsec:sim1_truncnorm}

At $K=30$, \textsc{BOOST} dominates every strong-FWER baseline on $\Pi_K$ across truncnorm, tdist, and sparse families at calibrated FWER (Table~\ref{tab:sim_master}, Fig.~\ref{fig:sim_families}), with $1.4$--$1.9\times$ the best stepwise competitor at moderate signals. Closed-Fisher ($\notin\mathfrak D_{\mathrm{sep}}$) leads only on in-class truncnorm at strong signals and collapses elsewhere; the ranking is preserved at $K\in\{6,15,60\}$ (Appendix~\ref{app:scaleK_full}). \emph{Finding}: \textsc{BOOST}'s $\Pi_K$-lead is robust across family, dependence, and block count.

\begin{figure}[htbp]
\centering

\includegraphics[width=\linewidth]{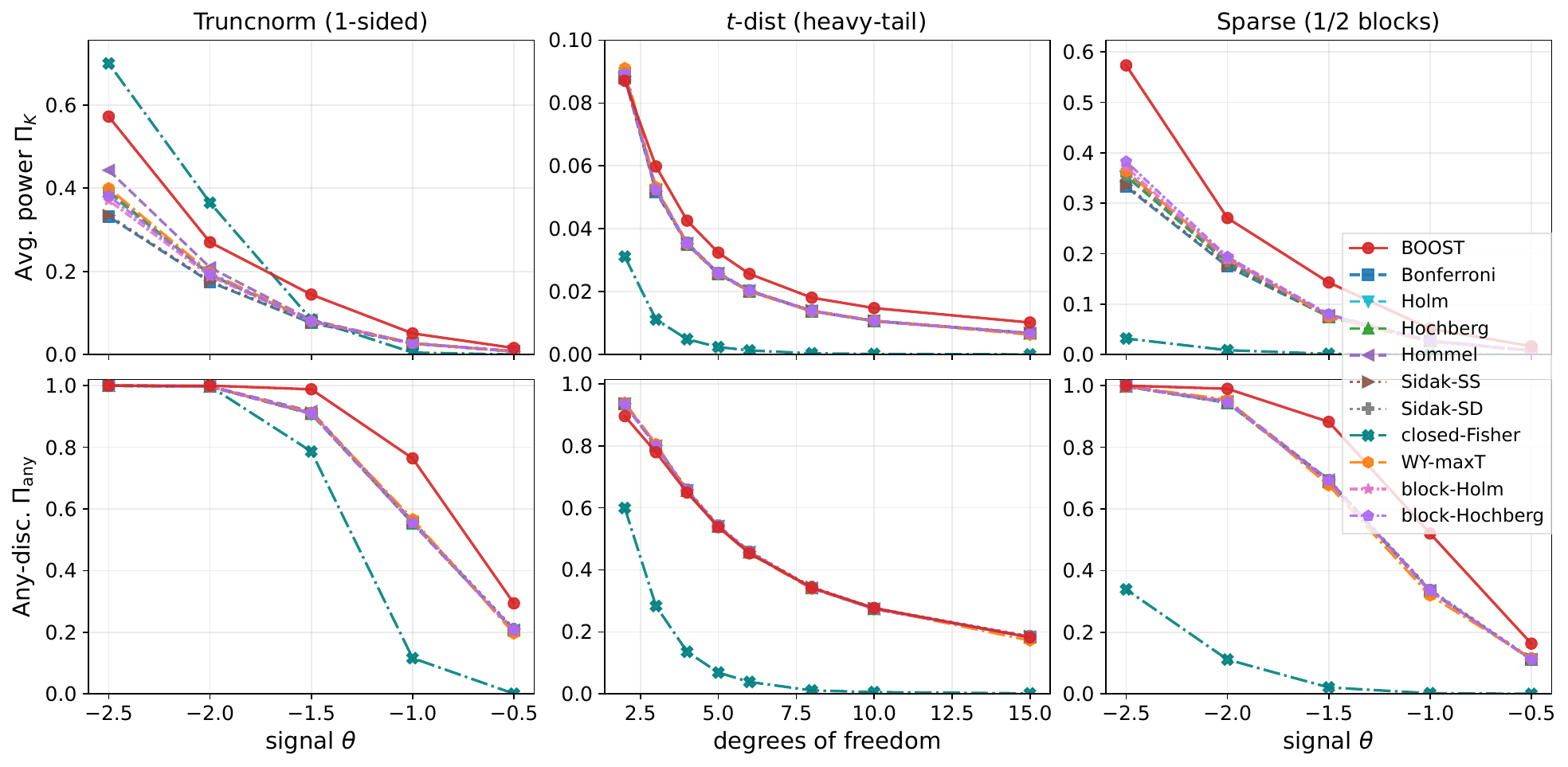}
\caption{\textsc{BOOST} (red) dominates every stepwise, graphical, and closed-testing baseline at moderate signals across three $p$-value families at $K=30$, complementing Table~\ref{tab:sim_master}; closed-Fisher ($\notin\mathfrak D_{\mathrm{sep}}$) leads only on in-class truncnorm at strong signals and collapses elsewhere. $B=10$, $\alpha=0.05$, $\alpha_{\mathrm{blk}}=\alpha/B$, $3\times 10^4$ replicates; top row $\Pi_K$, bottom row $\Pi_{\mathrm{any}}$; truncnorm and sparse cropped to $\theta\in[-2.5,-0.5]$. $K\in\{6,15,60\}$ panels in Appendix~\ref{app:scaleK_full}.}

\label{fig:sim_families}
\end{figure}

\begin{table}[htbp]
\centering
\small
\caption{Master comparison across families at representative signals. $B=10$, $K=30$, $\alpha=0.05$, $2\times 10^4$ replicates (MC-SE $\le 0.0035$). Bold is per-row max (ties share bold). Both $\Pi_K$ (top block) and $\Pi_{\mathrm{any}}$ (bottom block) reported; $K\in\{6,15,60\}$ sweeps in Appendix~\ref{app:scaleK_full}.}
\label{tab:sim_master}
\resizebox{\linewidth}{!}{%
\begin{tabular}{l*{11}{r}}
\toprule
& Bonf & \v{S}id-SS & Holm & Hoch & Hommel & \v{S}id-SD & WY-max$T$ & bk-Holm & bk-Hoch & cl-Fisher & \textbf{\textsc{BOOST}} \\
\midrule
\multicolumn{12}{l}{\emph{Average power $\Pi_K$}} \\
Truncnorm, $\theta=-1.5$ & 0.076 & 0.077 & 0.080 & 0.080 & 0.083 & 0.081 & 0.081 & 0.079 & 0.080 & 0.084 & \textbf{0.144} \\
tdist, $\mathrm{df}=10$  & 0.011 & 0.011 & 0.011 & 0.011 & 0.011 & 0.011 & 0.011 & 0.011 & 0.011 & 0.000 & \textbf{0.015} \\
Sparse, $\theta=-2$      & 0.175 & 0.177 & 0.182 & 0.182 & 0.187 & 0.185 & 0.192 & 0.189 & 0.194 & 0.009 & \textbf{0.271} \\
\midrule
\multicolumn{12}{l}{\emph{Any-discovery power $\Pi_{\mathrm{any}}$}} \\
Truncnorm, $\theta=-1.5$ & 0.909 & 0.911 & 0.909 & 0.909 & 0.917 & 0.911 & 0.912 & 0.909 & 0.910 & 0.786 & \textbf{0.988} \\
tdist, $\mathrm{df}=10$  & 0.274 & 0.276 & 0.274 & 0.274 & 0.275 & 0.276 & \textbf{0.277} & 0.274 & 0.274 & 0.006 & \textbf{0.277} \\
Sparse, $\theta=-2$      & 0.944 & 0.947 & 0.945 & 0.945 & 0.950 & 0.947 & 0.954 & 0.944 & 0.946 & 0.111 & \textbf{0.990} \\
\bottomrule
\end{tabular}%
}
\end{table}

\subsection{Comparison with the general-$K$ ESP solver}
\label{subsec:sim_esp_h2h}

At small $K\in\{9,12\}$ where the general-$K$ ESP solver of~\citet{dubey26esp} is computable on the truncnorm family, ESP attains $1.3$--$1.8\times$ \textsc{BOOST}'s $\Pi_K$ but under slack FWER and at $25$--$260\times$ the wall-clock; it is infeasible at the $K\ge 30$ reference (Table~\ref{tab:sim_esp_h2h}). \textsc{BOOST} recovers $55\%$--$79\%$ of ESP's $\Pi_K$ at FWER-tight calibration. Heterogeneous and misspecified regimes are in Appendix~\ref{app:esp_hetero_check}. \emph{Finding}: at the operating $K$ where ESP is infeasible, \textsc{BOOST} recovers most of the unrestricted ceiling at applications scale.

\begin{table}[htbp]
\centering
\small

\caption{Head-to-head with the general-$K$ ESP solver~\citep{dubey26esp} (truncnorm, $\alpha=0.05$, $\alpha_{\mathrm{blk}}=\alpha/B$, $2\times 10^4$ replicates, single core). ESP attains higher $\Pi_K$ but under slack FWER and at $25$--$260\times$ the wall-clock; infeasible for the $3\times 10^4$-replicate sweeps of Fig.~\ref{fig:sim_families}. \textsc{BOOST} wall-clock is the $K=3$ preprocessing cost (once per block-level $\alpha_{\mathrm{blk}}$; downstream per-block decision is $O(1)$); ESP wall-clock is the iterative primal--dual solve per replicate.}

\label{tab:sim_esp_h2h}
\begin{tabular}{rr rrr rr rr}
\toprule
 & & \multicolumn{3}{c}{$\Pi_K$} & \multicolumn{2}{c}{FWER} & \multicolumn{2}{c}{wall-clock (s)} \\
\cmidrule(lr){3-5} \cmidrule(lr){6-7} \cmidrule(lr){8-9}
$K$ & $\theta$ & Bonf & \textbf{\textsc{BOOST}} & ESP & \textbf{\textsc{BOOST}} & ESP & \textbf{\textsc{BOOST}} & ESP \\
\midrule
9  & $-1.5$ & 0.149 & \textbf{0.236} & 0.401 & 0.046 & 0.017 & \textbf{3.3} & 212 \\
9  & $-2.5$ & 0.485 & \textbf{0.631} & 0.799 & 0.047 & 0.015 & \textbf{0.9} & 37  \\
12 & $-1.5$ & 0.126 & \textbf{0.213} & 0.390 & 0.046 & 0.014 & \textbf{1.8} & 45  \\
12 & $-2.5$ & 0.445 & \textbf{0.606} & 0.801 & 0.049 & 0.009 & \textbf{2.6} & 675 \\
\bottomrule
\end{tabular}
\end{table}

\subsection{Heterogeneous block families}
\label{subsec:sim_hetero}

When alternative densities differ by block, per-block fitting converts heterogeneity into a power source. At $B=6$, \textsc{BOOST-adaptive} edges \textsc{BOOST} on a truncnorm$+$mixnorm mix (both $1.6$--$1.7\times$ stepwise baselines), and recovers $+39\%$ on the tdist half of a truncnorm$+$tdist$_{\mathrm{df}=4}$ mix (Fig.~\ref{fig:sim_hetero}). The lift compounds with mix diversity (three-family mix in Appendix~\ref{app:hetero_blocks}). \emph{Finding}: the more diverse the block mix, the larger the per-block-fit advantage.

\begin{figure}[htbp]
\centering
\includegraphics[width=\linewidth]{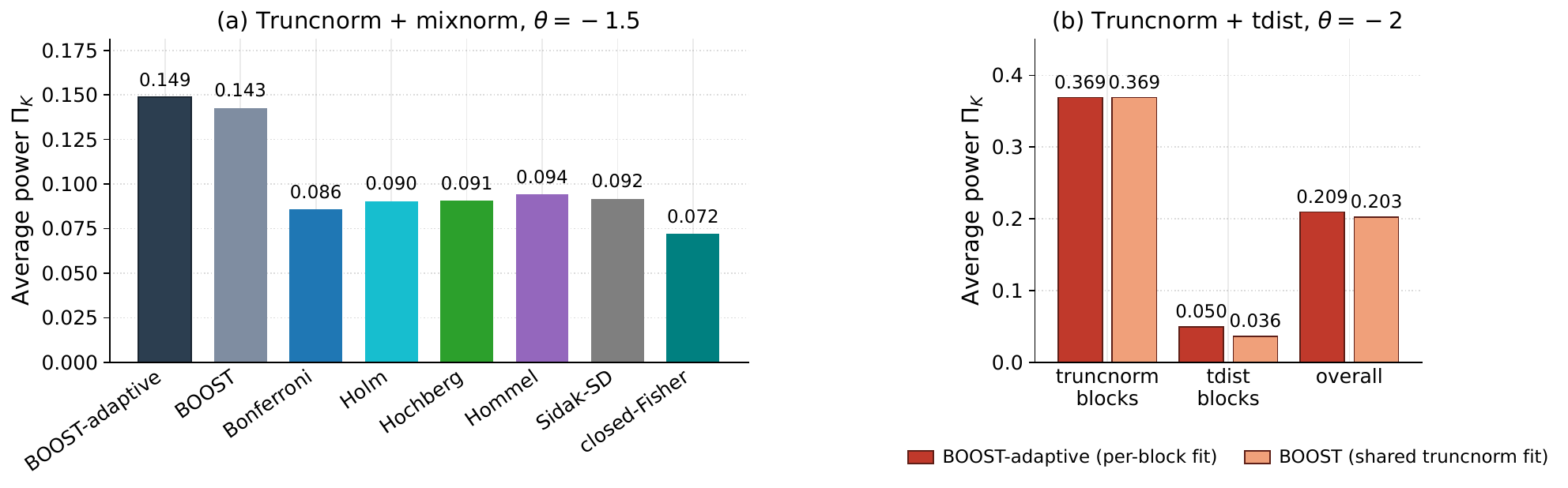}

\caption{Per-block fitting converts block heterogeneity into a power source ($B=6$, $K=18$, $\alpha=0.05$, $\alpha_{\mathrm{blk}}=\alpha/B$; $3\times 10^4$ replicates). (a)~Truncnorm$+$mixnorm at $\theta=-1.5$: \textsc{BOOST-adaptive} edges \textsc{BOOST}; both $1.6$--$1.7\times$ baselines. (b)~Truncnorm$+$tdist($\mathrm{df}=4$) at $\theta=-2$: per-block fitting recovers $+39\%$ on the tdist half; truncnorm-block power is identical by construction. First two groups in each panel are per-family $\Pi_K$ breakdowns; right is overall.}

\label{fig:sim_hetero}
\end{figure}

\subsection{Hierarchical Simes and e-value closed tests}
\label{subsec:sim_modern}

On the natural block-tree, hierarchical Simes~\citep{meinshausen2008hierarchical} coincides with Bonferroni, and the Vovk--Wang-calibrated e-closure of~\citet{hartog2025evalues} trails Bonferroni by up to an order of magnitude at moderate signals, a calibration artifact for adversarial composition rather than the strong-FWER block-fixed regime (Table~\ref{tab:sim_modern_main}, interpretation in Appendix~\ref{app:sim_modern_baselines}). \textsc{BOOST} leads both across the full $\theta$ grid. \emph{Finding}: \textsc{BOOST} sustains a meaningful $\Pi_K$ margin against both recent closed-testing benchmarks at every signal strength.

\begin{table}[htbp]
\centering
\caption{$\Pi_K$ and $\Pi_{\mathrm{any}}$ vs.\ hierarchical Simes~\citep{meinshausen2008hierarchical} and Vovk--Wang-calibrated e-closure~\citep{hartog2025evalues}. Truncnorm, $B=10$, $K=30$, $\alpha=0.05$, $2\times 10^4$ replicates. FWER$_0$ is the complete-null rejection rate; bold is per-column max. Meinshausen coincides with Bonferroni on the block-tree; the Vovk--Wang-calibrated Hartog closure is up to an order of magnitude weaker at moderate signals ($\theta\ge -2$) in this benchmark (interpretation in Appendix~\ref{app:sim_modern_baselines}).}
\label{tab:sim_modern_main}
\resizebox{\textwidth}{!}{%
\begin{tabular}{lc rr rr rr rr rr}
\toprule
 & & \multicolumn{2}{c}{$\theta=-1$} & \multicolumn{2}{c}{$\theta=-1.5$} & \multicolumn{2}{c}{$\theta=-2$} & \multicolumn{2}{c}{$\theta=-2.5$} & \multicolumn{2}{c}{$\theta=-3$} \\
\cmidrule(lr){3-4}\cmidrule(lr){5-6}\cmidrule(lr){7-8}\cmidrule(lr){9-10}\cmidrule(lr){11-12}
Method & FWER$_0$ & $\Pi_K$ & $\Pi_{\mathrm{any}}$ & $\Pi_K$ & $\Pi_{\mathrm{any}}$ & $\Pi_K$ & $\Pi_{\mathrm{any}}$ & $\Pi_K$ & $\Pi_{\mathrm{any}}$ & $\Pi_K$ & $\Pi_{\mathrm{any}}$ \\
\midrule
Bonferroni          & 0.049 & 0.026 & 0.554 & 0.076 & 0.906 & 0.174 & 0.997 & 0.332 & 1.000 & 0.525 & 1.000 \\
Holm                & 0.049 & 0.027 & 0.554 & 0.079 & 0.906 & 0.192 & 0.997 & 0.390 & 1.000 & 0.655 & 1.000 \\
Hommel              & 0.049 & 0.027 & 0.559 & 0.083 & 0.915 & 0.209 & 0.999 & 0.446 & 1.000 & 0.745 & 1.000 \\
Meinshausen (2008)  & 0.049 & 0.026 & 0.554 & 0.076 & 0.906 & 0.174 & 0.997 & 0.332 & 1.000 & 0.525 & 1.000 \\
Hartog e-val (2025) & 0.000 & 0.000 & 0.003 & 0.001 & 0.020 & 0.015 & 0.153 & 0.124 & 0.665 & 0.345 & 0.988 \\
\textbf{\textsc{BOOST}}      & 0.045 & \textbf{0.051} & \textbf{0.764} & \textbf{0.143} & \textbf{0.987} & \textbf{0.270} & \textbf{1.000} & \textbf{0.573} & \textbf{1.000} & \textbf{0.765} & \textbf{1.000} \\
\bottomrule
\end{tabular}}
\end{table}

\subsection{Robustness: stringent $\alpha$, dependence, plug-in}
\label{subsec:sim_robustness}
\label{subsec:sim_alpha}

Three robustness sweeps preserve the lead. Across $\alpha\in\{0.01,0.05,0.10\}$, the $\Pi_K$-gap over Hommel widens from $1.20\times$ at $\alpha=0.10$ to $1.99\times$ at $\alpha=0.01$ (Table~\ref{tab:sim_alpha_sensitivity}). Under cross-block equicorrelation and 1-factor dependence, global FWER stays within $0.02$ of nominal (Appendix~\ref{app:sim_dep_fwer}). For Grenander-estimated $\hat g$, FWER is controlled and the per-hypothesis oracle deficit does not grow with $B$ (Fig.~\ref{fig:sim_plugin_combined}), consistent with Theorem~\ref{thm:plugin_power}. \emph{Finding}: \textsc{BOOST}'s advantage is largest in the confirmatory-inference regime, where stringent $\alpha$, structured dependence, and finite-sample density estimation coincide.

\begin{table}[htbp]
\centering
\small
\caption{Sensitivity to $\alpha$ (truncnorm, $\theta=-2$, $K=30$, $B=10$; $3\times 10^4$ Monte-Carlo replicates, experiment E$\alpha$). \textsc{BOOST}'s $\Pi_K$-advantage over Hommel widens from $1.20\times$ at $\alpha=0.10$ to $1.99\times$ at $\alpha=0.01$; $\Pi_{\mathrm{any}}$ saturates near $1$ for every method, so the relevant differentiator is $\Pi_K$.}
\label{tab:sim_alpha_sensitivity}
\begin{tabular}{l cccc cccc}
\toprule
 & \multicolumn{4}{c}{Average power $\Pi_K$} & \multicolumn{4}{c}{Any-discovery power $\Pi_{\mathrm{any}}$} \\
\cmidrule(lr){2-5} \cmidrule(lr){6-9}
$\alpha$ & Bonf & Holm & Hommel & \textbf{\textsc{BOOST}} & Bonf & Holm & Hommel & \textbf{\textsc{BOOST}} \\
\midrule
0.01 & 0.080 & 0.084 & 0.086 & \textbf{0.171} & 0.919 & 0.919 & 0.925 & \textbf{0.995} \\
0.05 & 0.175 & 0.193 & 0.210 & \textbf{0.271} & 0.997 & 0.997 & 0.998 & \textbf{1.000} \\
0.10 & 0.238 & 0.272 & 0.315 & \textbf{0.379} & 1.000 & 1.000 & 1.000 & \textbf{1.000} \\
\bottomrule
\end{tabular}
\end{table}

\begin{figure}[htbp]
\centering
\includegraphics[width=\linewidth]{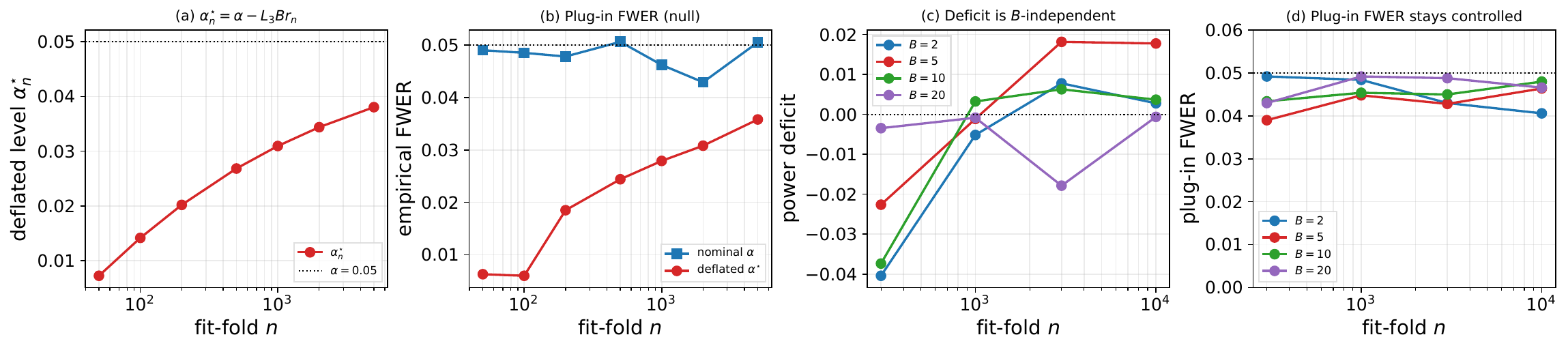}
\caption{Plug-in FWER is controlled and the per-hypothesis oracle deficit is $B$-independent for $n\ge 300$ (plug-in \textsc{BOOST}, truncnorm, $\theta_{\mathrm{fit}}=-2$, $L_3=10^{-2}$, $\alpha=0.05$). \emph{Top row, FWER calibration at $B=10$:} (a) deflation $\alpha^{\star}_n=\alpha-L_3 B r_n$ with $r_n=(n^{-1}\log n)^{1/3}$; (b) empirical complete-null FWER under nominal vs.\ deflated levels ($10^4$ replicates per $n$). \emph{Bottom row, plug-in--oracle gap across $B$:} (c) per-hypothesis oracle deficit vs.\ $n$, grouped by $B\in\{2,5,10,20\}$; (d) complete-null plug-in FWER, below nominal throughout. Panels (c)--(d) use $5\times 10^3$ replicates per cell.}
\label{fig:sim_plugin_combined}
\end{figure}

\section{Real-World Data Applications}
\label{sec:applications}

We validate \textsc{BOOST}, the $\Pi_K$-optimized (``all-$K$-reject'') policy of Theorem~\ref{thm:global_opt_sep_allocation}, on two published datasets with a domain-given $K=3$ block structure: (i)~multi-tissue cis-eQTL across the canonical immune-cell trio (BLUEPRINT, Section~\ref{subsec:blueprint_app}), and (ii)~bundled-challenger A/B tests from the Upworthy Research Archive (Section~\ref{subsec:upworthy_app}). In both applications \textsc{BOOST} delivers order-of-magnitude gains on $\Pi_K$ at controlled FWER, concentrating its rejection mass on full triples rather than scattered any-coordinate hits, the behavior its design is tuned for. For each application the results table juxtaposes rejection counts on the real test half at $\alpha=0.05$ with an oracle Monte-Carlo validation at $n_{\mathrm{rep}}=5000$ replicates drawn from the fitted Grenander density $\hat g^{\mathrm{real}}$ that \textsc{BOOST} uses internally (Monte-Carlo design and baseline scoping in Appendix~\ref{app:apps_setup}). Both applications run \textsc{BOOST} at uniform $\alpha/B$ allocation with a sample-split pooled $\hat g$ (Algorithm~\ref{alg:plugin_blockwise_k3}); the heterogeneous KKT allocators of Theorem~\ref{thm:global_opt_sep_allocation} (Bonferroni budget) and Theorem~\ref{thm:global_opt_sep_allocation_sidak} (\v{S}id\'ak budget, when cross-block independence is available) are exercised in the simulated studies of Section~\ref{sec:simulations}.

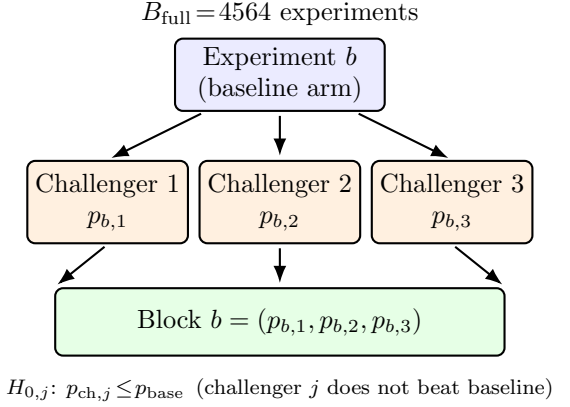
\begin{figure}[htbp]
\centering
\resizebox{\linewidth}{!}{%
\begin{tikzpicture}[
  every node/.style={font=\footnotesize},
  unit/.style={rectangle, rounded corners=3pt, draw, thick,
              align=center, inner sep=3pt,
              minimum width=28mm, minimum height=9mm, fill=blue!8},
  child/.style={rectangle, rounded corners=3pt, draw, thick,
               align=center, inner sep=3pt,
               minimum width=20mm, minimum height=11mm, fill=orange!12},
  block/.style={rectangle, rounded corners=3pt, draw, thick,
               align=center, inner sep=3pt,
               minimum width=60mm, minimum height=9mm, fill=green!10},
  arr/.style={-Latex, thick, shorten >=1.5pt, shorten <=1.5pt}
]
\node[font=\small\bfseries, align=center] at (-42mm, 22mm) {BLUEPRINT (cross-lineage cis-eQTL)\\[1pt]\normalfont\footnotesize $B_{\mathrm{full}}\!=\!11{,}260$ genes};
\node[unit] (geneA) at (-42mm, 11mm) {Gene $b$\\ (lead cis-SNP)};
\node[child] (cA1) at (-65mm, -6mm) {CD14$^{+}$ monocyte\\ $p_{b,1}$};
\node[child] (cA2) at (-42mm, -6mm) {Neutrophil\\ $p_{b,2}$};
\node[child] (cA3) at (-19mm, -6mm) {CD4$^{+}$ T cell\\ $p_{b,3}$};
\draw[arr] (geneA) -- (cA1.north);
\draw[arr] (geneA) -- (cA2);
\draw[arr] (geneA) -- (cA3.north);
\node[block] (blockA) at (-42mm, -22mm) {Block $b=(p_{b,1},p_{b,2},p_{b,3})$};
\draw[arr] (cA1.south) -- (blockA.north west);
\draw[arr] (cA2.south) -- (blockA);
\draw[arr] (cA3.south) -- (blockA.north east);
\node[font=\scriptsize, align=center] at (-42mm, -31mm) {$H_{0,j}$: $\beta^{(j)}_{\mathrm{SNP}}\!=\!0$ \,(no cis-association in lineage $j$)};
\node[font=\small\bfseries, align=center] at (42mm, 22mm) {Upworthy (bundled-challenger A/B)\\[1pt]\normalfont\footnotesize $B_{\mathrm{full}}\!=\!4564$ experiments};
\node[unit] (expB) at (42mm, 11mm) {Experiment $b$\\ (baseline arm)};
\node[child] (cB1) at (19mm, -6mm) {Challenger $1$\\ $p_{b,1}$};
\node[child] (cB2) at (42mm, -6mm) {Challenger $2$\\ $p_{b,2}$};
\node[child] (cB3) at (65mm, -6mm) {Challenger $3$\\ $p_{b,3}$};
\draw[arr] (expB) -- (cB1.north);
\draw[arr] (expB) -- (cB2);
\draw[arr] (expB) -- (cB3.north);
\node[block] (blockB) at (42mm, -22mm) {Block $b=(p_{b,1},p_{b,2},p_{b,3})$};
\draw[arr] (cB1.south) -- (blockB.north west);
\draw[arr] (cB2.south) -- (blockB);
\draw[arr] (cB3.south) -- (blockB.north east);
\node[font=\scriptsize, align=center] at (42mm, -31mm) {$H_{0,j}$: $p_{\mathrm{ch},j}\!\le\!p_{\mathrm{base}}$ \,(challenger $j$ does not beat baseline)};
\draw[dashed, gray!55] (0, 30mm) -- (0, -34mm);
\end{tikzpicture}%
}
\caption{Domain-dictated $K=3$ block structure for the two applications. Each domain unit (a gene in BLUEPRINT, an experiment in Upworthy) generates three coordinate tests against a common internal reference (three immune lineages sharing the same cis-SNP; three challengers against a shared baseline), whose $p$-values form one block. \textsc{BOOST} certifies the block-level event $\bigcap_{j=1}^{3}\{p_{b,j}\text{ rejected}\}$ (``all three reject''), the inferential target in both applications.}
\label{fig:app_schematic}
\end{figure}

\subsection{Multi-tissue cis-eQTL across the BLUEPRINT immune trio}
\label{subsec:blueprint_app}

The \textsc{BLUEPRINT} consortium~\citep{chen2016blueprint} mapped cis-eQTLs in three primary hematopoietic cell types: CD14$^{+}$ monocytes, neutrophils, and CD4$^{+}$ T cells. Block~$b$ is one gene, and the $K=3$ hypotheses test whether its lead-SNP cis-association is non-null in each lineage (Fig.~\ref{fig:app_schematic}, left). Biological symmetry of the triad (innate myeloid, innate granulocyte, adaptive lymphoid) gives within-block $h$-exchangeability. Genes with cis-eQTL signal in all three lineages are candidates for cell-type-shared regulatory variants (prioritized in functional follow-up over single-lineage hits), so the full-triple rejection $\Pi_K$ is the inferential target. Intersecting genes across all three cell types gives $B_{\mathrm{full}}=11{,}260$. We sample $B=2000$ and split fit/test evenly, giving $B_{\mathrm{test}}=1000$. Clustered gene neighborhoods leave residual cross-block dependence, so we forgo the \v{S}id\'ak tightening and retain the uniform $\alpha/B$ allocation. Data source in Appendix~\ref{app:apps_setup}.

Table~\ref{tab:app_blueprint}: \textsc{BOOST} certifies $\mathbf{419}$ genes with simultaneous cis-eQTL signal across the three immune lineages, beating Hartog e-value ($123$) by $3.4\times$, block-Hochberg ($100$) by $4.2\times$, and the single-step baselines ($89$) by $4.7\times$; the oracle Monte-Carlo replicates the same ordering: $\Pi_K=0.244$ for \textsc{BOOST} vs $0.031$ Hartog ($7.9\times$) and $0.019$ block-Hochberg ($13\times$), with strong FWER at $0.026$, well below the guarantee $\le\alpha$ from Theorem~\ref{thm:blockwise_strong_fwer_clean}. \textsc{BOOST} also exceeds the BH-FDR row (a strictly looser error criterion) by $2.4\times$ on the full-triple target: BH controls each column's discovery rate independently, whereas \textsc{BOOST} targets the joint event directly.

\begin{table}[htbp]
\centering
\small
\setlength{\tabcolsep}{3pt}
\begin{tabular}{lrrrrr}
\toprule
       & \multicolumn{2}{c}{Test half (counts)} & \multicolumn{3}{c}{Oracle Monte-Carlo}\\
\cmidrule(lr){2-3}\cmidrule(lr){4-6}
Method & Total & All $3$ & $\mathrm{FWER}_0$ & $\Pi_{\mathrm{any}}$ & $\Pi_K$ \\
\midrule
\textsc{BOOST}                   & 1292 & \textbf{419} & 0.0264 & 0.5836 & \textbf{0.2443}\\
Bonferroni                       & 709  & 89           & 0.0544 & 0.5764 & 0.0155\\
\v{S}id\'ak SS                   & 710  & 89           & 0.0548 & 0.5770 & 0.0155\\
Holm                             & 722  & 92           & 0.0544 & 0.5843 & 0.0164\\
Hochberg                         & 722  & 92           & 0.0544 & 0.5843 & 0.0164\\
Hommel                           & 726  & 94           & 0.0544 & 0.5861 & 0.0166\\
\v{S}id\'ak SD                   & 725  & 94           & 0.0548 & 0.5849 & 0.0165\\
Block-Holm                       & 726  & 100          & 0.0544 & 0.5764 & 0.0186\\
Block-Hochberg                   & 726  & 100          & 0.0544 & 0.5765 & 0.0187\\
Meinshausen hier.                & 709  & 89           & 0.0544 & 0.5764 & 0.0155\\
Hartog e-value                   & 862  & 123          & 0.0000 & 0.6171 & 0.0308\\
\midrule
BH per cell type (FDR)           & 1209 & 172          & n/a    & n/a    & n/a   \\
\bottomrule
\end{tabular}
\caption{BLUEPRINT cis-eQTL results at $\alpha=0.05$, $B_{\mathrm{test}}=1000$, $K=3$. \emph{All-3 rej.}: full-triple count on the test half. \emph{Oracle MC}: $n_{\mathrm{rep}}=5000$, null iid $\mathrm{Unif}[0,1]$ ($\mathrm{FWER}_0$), alt iid $\hat g^{\mathrm{real}}$ ($\Pi_{\mathrm{any}}=\Pr(\ge\!1\text{ rej})$, $\Pi_K=\Pr(\text{all }3\text{ rej})$). \textsc{BOOST} certifies $\mathbf{419}$ simultaneous-signal genes, $4.7\times$ the Bonferroni-family baselines ($89$) and $3.4\times$ Hartog ($123$); oracle $\Pi_K=0.244$ at $\mathrm{FWER}_0=0.026<\alpha$. BH-FDR is a looser-criterion ceiling; MC omitted.}
\label{tab:app_blueprint}
\end{table}

\subsection{Online experimentation: the Upworthy headline archive}
\label{subsec:upworthy_app}

The Upworthy Research Archive~\citep{matias2021upworthy} releases a major U.S.\ publisher's headline-and-image randomized-experiment ledger ($32{,}487$ experiments, ${\sim}10^{8}$ impressions). Restricting to at least four arms yields $B=4564$ blocks; within each experiment, the highest-impressions arm is the baseline and the next three are challengers, so block~$b$ is the experiment and the $K=3$ hypotheses are the one-sided two-proportion $z$-tests $H_{0,j}\!:\!p_{\mathrm{ch},j}\!\le\!p_{\mathrm{base}}$, $j\!=\!1,2,3$ (Fig.~\ref{fig:app_schematic}, right). Certifying that all three challengers beat the baseline is editorial evidence of a systematic content gap (not merely a single-arm win), so $\Pi_K$ is the business-relevant object. Challengers are randomized within an experiment ($h$-exchangeability) and experiments ran on disjoint impression windows (cross-block independence); $B_{\mathrm{test}}=2282$.

Table~\ref{tab:app_upworthy}: \textsc{BOOST} certifies $\mathbf{99}$ experiments where all three challengers beat the baseline, $14\times$ Hartog's $7$ and $33\times$ every non-Hartog closed-testing competitor (tied at $3$); the oracle Monte-Carlo reinforces the scale: $\Pi_K=0.0095$ for \textsc{BOOST} vs $\le 10^{-4}$ for every strong-FWER competitor (two-to-three orders of magnitude: Hartog at $10^{-4}$; Bonferroni family effectively zero), with FWER controlled at $0.033<\alpha$. The Upworthy signal is sparse (pooled effect-size estimate $\hat\theta_{\mathrm{pool}}\!\approx\!-0.05$): the ${\approx}4\%$ of experiments carrying a genuine cross-challenger effect sit in a thin near-zero tail of $\hat g$, which a level-set boundary accesses but stepwise uniform slopes do not.

\begin{table}[htbp]
\centering
\small
\setlength{\tabcolsep}{3pt}
\begin{tabular}{lrrrrr}
\toprule
       & \multicolumn{2}{c}{Test half (counts)} & \multicolumn{3}{c}{Oracle Monte-Carlo}\\
\cmidrule(lr){2-3}\cmidrule(lr){4-6}
Method & Total & All $3$ & $\mathrm{FWER}_0$ & $\Pi_{\mathrm{any}}$ & $\Pi_K$ \\
\midrule
\textsc{BOOST}                   & 297 & \textbf{99} & 0.0328 & 0.0446 & \textbf{0.0095}\\
Bonferroni                       & 99  & 3           & 0.0504 & 0.0448 & 0.0000\\
\v{S}id\'ak SS                   & 100 & 3           & 0.0512 & 0.0450 & 0.0000\\
Holm                             & 100 & 3           & 0.0504 & 0.0449 & 0.0000\\
Hochberg                         & 100 & 3           & 0.0504 & 0.0449 & 0.0000\\
Hommel                           & 100 & 3           & 0.0504 & 0.0451 & 0.0000\\
\v{S}id\'ak SD                   & 100 & 3           & 0.0512 & 0.0452 & 0.0000\\
Block-Holm                       & 100 & 3           & 0.0504 & 0.0448 & 0.0000\\
Block-Hochberg                   & 100 & 3           & 0.0504 & 0.0448 & 0.0000\\
Meinshausen hier.                & 99  & 3           & 0.0504 & 0.0448 & 0.0000\\
Hartog e-value                   & 189 & 7           & 0.0000 & 0.0757 & 0.0001\\
\midrule
BH per challenger (FDR)          & 382 & 15          & n/a    & n/a    & n/a   \\
\bottomrule
\end{tabular}
\caption{Upworthy bundled-challenger results at $\alpha=0.05$, $B_{\mathrm{test}}=2282$, $K=3$. Column conventions as in Table~\ref{tab:app_blueprint}. \textsc{BOOST} certifies $\mathbf{99}$ experiments where all three challengers beat the baseline: $33\times$ the closed-testing competitors (tied at $3$) and $14\times$ Hartog ($7$); oracle $\Pi_K=0.0095$ is two-to-three orders of magnitude above every strong-FWER competitor (Hartog $10^{-4}$; Bonferroni family effectively $0$) at $\mathrm{FWER}_0=0.033<\alpha$.}
\label{tab:app_upworthy}
\end{table}

\paragraph{Takeaway.}
Across both applications \textsc{BOOST} dominates the full-triple target by $3$--$33\times$ on the test half and by an order of magnitude or more on oracle $\Pi_K$ at $\mathrm{FWER}_0<\alpha$. The lift comes from shaping the within-block boundary to the level sets of $\widehat g$ where stepwise rules apply a uniform slope; Hartog's e-closure is the only strong-FWER baseline that materially approaches \textsc{BOOST} on $\Pi_K$, but its leaf cutoff $p\le(\alpha/2)^2$ is tighter than $\alpha/K$ and trails by $8$--$100\times$. Operationally, \textsc{BOOST} can be inserted into any closed-testing or graphical pipeline already in use, with a sample-split $\widehat g$ as the only added input; Theorem~\ref{thm:plugin_power} certifies that the per-hypothesis power deficit is $B_{\mathcal T}$-free (further expansions in Appendix~\ref{app:sim_modern_baselines}).

\section{Discussion and Outlook}
\label{sec:conclusion}

We introduced \textsc{BOOST} (Algorithm~\ref{alg:blockwise_k3}), a power-optimal procedure within the block-separable class $\mathfrak D_{\mathrm{sep}}$. Theorem~\ref{thm:blockwise_strong_fwer_clean} establishes strong FWER control at level $\alpha$ under block-measurability and local validity alone, with no independence assumption and $O(K)$ evaluation. Theorems~\ref{thm:global_opt_sep_within_block}--\ref{thm:global_opt_sep_allocation} reduce the global optimum to per-block value functions $\pi_3^{(b)}$ with an equalized-marginal KKT allocator computable in $O(B\log(1/\varepsilon))$. Under cross-block independence the \v{S}id\'ak tightening (Theorem~\ref{thm:global_opt_sep_allocation_sidak}) strictly improves the optimum, which Proposition~\ref{prop:containment} pins between the global \v{S}id\'ak floor and the per-marginal Neyman--Pearson ceiling. A sample-split plug-in variant (Algorithm~\ref{alg:plugin_blockwise_k3}) handles unknown $g$ with FWER inflation $O(B_{\mathcal T}\,\mathbb E\|g-\widehat g\|_\infty)$ (Theorem~\ref{thm:plugin_fwer}) and per-hypothesis power deficit $B_{\mathcal T}$-free (Theorem~\ref{thm:plugin_power}). Proofs, extension statements, AI-tool usage, and acknowledgements are in Appendices~\ref{app:k3_full}--\ref{app:meta}.

\paragraph{Limitations.}
Optimality is \emph{within} $\mathfrak D_{\mathrm{sep}}$, not against all monotone strong-FWER rules: the unrestricted solver of~\citet{dubey26esp} attains $1.3$--$1.8\times$ higher $\Pi_K$ at small $K$, but at slack FWER, at $25$--$260\times$ the wall-clock, and infeasibly at the $K\ge 30$ scales we target (Section~\ref{subsec:sim_esp_h2h}). The per-block atom is the $K{=}3$ solver of~\citet{dubey25}; the framework is modular in the atom (Section~\ref{subsec:blockwise_method}), so any $K\ge 4$ atom is a drop-in replacement. The plug-in FWER bound grows linearly in $B_{\mathcal T}$ while per-hypothesis power deficit remains $B_{\mathcal T}$-free (Theorem~\ref{thm:plugin_power}). The partition is taken as design-dictated; principled partition selection is left open.

\paragraph{Extensions.}
Four directions stand out. \emph{Block-formation optimality}: joint optimization over (partition, allocation, atom), with the concavity of each $\pi_3^{(b)}$ and the allocator of Theorem~\ref{thm:global_opt_sep_allocation} coupling block-size trade-offs to the allocation penalty. \emph{Higher block sizes}: $K\ge 4$ atoms from~\citet{dubey26esp} can be substituted where dimensionality permits; certified-suboptimal approximations preserve the $O(K)$ outer-loop scaling at larger $K$. \emph{FDR analog}: blockwise allocation lifts to BH-type step-up outer rules paired with per-block expected-discovery atoms. \emph{Dependence-adaptive allocation}: resampling-calibrated allocation in the spirit of~\citet{WestfallYoung1993} applied at block level should extract sharper power under structured cross-block dependence.

\newpage
\bibliographystyle{plainnat}
\bibliography{cite}

\newpage
\appendix
\renewcommand{\thesection}{S\arabic{section}}
\renewcommand{\thesubsection}{S\arabic{section}.\arabic{subsection}}
\setcounter{section}{0}

\begin{center}
{\Large\bfseries Supplementary Material}
\end{center}

\medskip
This supplement contains proofs of all theoretical results stated in the main paper, full $p$-value formulation and $K=3$ atom internals, deferred remarks and corollaries, plug-in perturbation lemmas, the full simulation grid (scaling, convergence, timing, dependence stress, plug-in rates, ESP and modern closed-testing baselines), real-data application setup, reproducibility notes, the AI-tool-usage declaration, and acknowledgements. We use the notation of the main paper throughout.

\input{appendix_arxiv}

\end{document}

%% file: appendix_arxiv.tex
\section{Detailed $p$-value formulation and $K=3$ building block}
\label{app:k3_full}

This appendix reproduces the $p$-value formulation, linear representation, dual program, and coordinate-update optimizer for $K=3$ from~\citet{dubey25}, following their notation, so that the present paper is self-contained.
The main-body construction of Section~\ref{sec:main} invokes the $K=3$ optimizer as the per-block atom without unpacking it; this appendix is for readers who want to follow the block solver's internals without opening a separate manuscript.

\subsection{$p$-value formulation and linear representation}
\label{sec:pvalue_formulation}

We now move to a $p$-value representation, which yields a tractable ordered-domain formulation.
For each test $k$, define the (one-sided) $p$-value induced by the likelihood ratio threshold:
\begin{equation}
\label{eq:pvalue_def}
u_k
=\int \mathbf 1\!\left\{ \frac{f_{k,1}(x)}{f_{k,0}(x)} \ge \Lambda_k(X_k)\right\} f_{k,0}(x)\,dx,
\end{equation}
so that $u_k\sim \mathrm{Unif}(0,1)$ under $H_{0k}$.
Under the alternative, assume the $p$-value admits a density $g$ on $[0,1]$ (common across $k$ under symmetry):
\begin{equation}
\label{eq:pvalue_model}
H_{0k}:u_k\sim \mathrm{Unif}(0,1)
\qquad\text{vs.}\qquad
H_{Ak}:u_k\sim g,
\qquad k=1,\dots,K.
\end{equation}
Write $u_{(1)}\le \cdots\le u_{(K)}$ for the order statistics and define the ordered simplex
$
Q=\Bigl\{u\in[0,1]^K:0\le u_1\le \cdots\le u_K\le 1\Bigr\}.
$
Any symmetric LR-ordered decision rule can be represented as a measurable map $D:Q\to\{0,1\}^K$ applied to $(u_{(1)},\dots,u_{(K)})$.
In the LR-based construction above, the alternative density $g$ is non-increasing on $[0,1]$ (see~\citet{RHPA22,dubey25}).

\paragraph{Linear representation of power and strong-FWER constraints.}
Both the objective and each strong-FWER constraint of \eqref{eq:general_objective_driven} admit a \emph{linear} representation over $Q$~\citep{RHPA22}.
For each $\ell\in\{0,\dots,K-1\}$, there exist non-negative coefficient functions $a_i,b_{\ell,i}:Q\to\mathbb R_+$ ($i\in[K]$) such that
\begin{equation}
\label{eq:linear_power_general}
\Pi(\vec{D})=\int_{Q}\sum_{i=1}^K a_i(u)\,D_i(u)\,du,
\qquad
\mathrm{FWER}_\ell(\vec{D})=\int_{Q}\sum_{i=1}^K b_{\ell,i}(u)\,D_i(u)\,du.
\end{equation}
The \textit{objective-driven} strong-FWER problem therefore takes the canonical form
\begin{equation}
\label{eq:canonical_primal}
\begin{aligned}
\max_{\vec{D}\in\mathcal D_K}\quad &
\int_{Q}\sum_{i=1}^K a_i(u)\,D_i(u)\,du\\
\text{s.t.}\quad &
\int_{Q}\sum_{i=1}^K b_{\ell,i}(u)\,D_i(u)\,du \le \alpha,
\qquad \ell=0,1,\dots,K-1.
\end{aligned}
\end{equation}

\subsection{Computation of the optimal $K=3$ multipliers}
\label{sec:dual_general}

We state the strong-duality result that justifies solving the primal via the dual.
\begin{lemma}[Dual problem, \protect\cite{RHPA22}]
\label{lem:dual_problem}
Let the Lagrangian associated with \eqref{eq:canonical_primal} be
\begin{equation}\label{eq:lagrangian1}
  L(\vec D,\mu)=
    \sum_{\ell=0}^{K-1}\mu_{\ell}\alpha
    +\int_{Q}\sum_{i=1}^{K}D_i(\vec u)R_i(\mu,\vec u)d\vec u,
\end{equation}
where the dual vector $\mu=(\mu_{0},\dots,\mu_{K-1})^{\top}\ge 0$ and
\begin{equation}\label{eq:rimu1}
  R_{i}(\mu,\vec u) = a_i(\vec u)-\sum_{\ell=0}^{K-1}\mu_\ell b_{\ell,i}(\vec u),\qquad
  1\le i\le K.
\end{equation}
For every fixed $\mu\ge 0$, define
\begin{equation}\label{eq:l*}
  \ell^{\ast}(\mu,\vec u) = \argmax_{1\le \ell\le K} \Bigl\{0,\sum_{i=1}^{\ell}R_i(\mu,\vec u)\Bigr\},
\end{equation}
and the induced LR-ordered policy
\begin{equation}\label{eq:optdecision1}
  D^{\mu}_{i}(\vec u) = \mathbf 1\!\bigl\{1\le i\le \ell^{\ast}(\mu,\vec u)\bigr\},\qquad 1\le i\le K.
\end{equation}
Then $\max_{\vec D\in\mathcal D} L(\vec D,\mu) = L(\vec D^{\mu},\mu)$, and the dual program reduces to
\begin{equation}\label{eq:ld_mu}
\min_{\mu\geq 0} \max_{\vec{D} \in \mathcal{D}} L(\vec{D}, \mu)
\;=\; \min_{\mu\geq 0}  L(\vec{D}^{\mu}, \mu).
\end{equation}
\end{lemma}
For $K=3$,~\citet[Lemma~3]{dubey25} derive the closed-form integrals below, and $\vec D^\mu$ has the closed form:
\begin{lemma}[\protect\cite{dubey25}, Lemma 3]
\label{lemma:optimization_problem_k_3}
With $\vec{D} = (D_1, D_2, D_3): Q \to [0,1]^3$ on $Q = \{\vec{u} \in [0,1]^3 : u_1 \le u_2 \le u_3\}$, average power and strong FWER admit the closed-form integral representations
\begin{equation}
\label{eq:objconst}
\begin{aligned}
\max _{\vec{D}: Q \rightarrow[0,1]^{3}} & 2 \int_{Q}\left(D_{1}+D_{2}+D_{3}\right) g(u_1) g(u_2) g(u_3) d \vec{u} \\
\text { s.t. } & \mathrm{FWER}_{0}=6 \int_{Q} D_{1}\,d \vec{u} \le \alpha, \\
& \mathrm{FWER}_{1}=2 \int_{Q}\!\left[D_{1}(g(u_2)+g(u_3))+D_{2} g(u_1) \right] d \vec{u} \le \alpha, \\
& \mathrm{FWER}_{2}=2 \int_{Q}\!\left[D_{1} g(u_2) g(u_3)+ D_{2} g(u_1) g(u_3) + D_{3} g(u_1) g(u_2)\right] d \vec{u} \le \alpha, \\
& 0 \le D_{3} \le D_{2} \le D_{1} \le 1.
\end{aligned}
\end{equation}
\end{lemma}
\begin{lemma}[\protect\cite{dubey25}, Lemma 4]
\label{lem:opt_decision_mu}
For fixed $\mu=(\mu_0,\mu_1,\mu_2)$, the maximizer $\vec{D}^\mu=(D_1^\mu,D_2^\mu,D_3^\mu)$ of \eqref{eq:ld_mu} is
\begin{equation}\label{eq:dimus2}
D_{1}^{\mu}(\vec{u})=\alpha_1^{\mu}(\vec{u}),\quad
D_{2}^{\mu}(\vec{u})=\alpha_{1}^{\mu}(\vec{u}) \alpha_2^{\mu}(\vec{u}),\quad
D_{3}^{\mu}(\vec{u})=\alpha_{1}^{\mu}(\vec{u}) \alpha_{2}^{\mu}(\vec{u}) \alpha_3^{\mu}(\vec{u}),
\end{equation}
with
\begin{equation}\label{eq:alpha_mu_u}
\begin{aligned}
\alpha_1^{\mu}(\vec{u}) &= \mathds{1}\{R_{1}>0 \cup R_{1}+R_{2}>0 \cup R_{1}+R_{2}+R_{3}>0\}, \\
\alpha_2^{\mu}(\vec{u}) &= \mathds{1}\{R_{2}>0 \cup R_{2}+R_{3}>0\}, \\
\alpha_3^{\mu}(\vec{u}) &= \mathds{1}\{R_{3}>0\},
\end{aligned}
\end{equation}
where $R_{i}(\mu,\vec{u})=a_i(\vec u)-\sum_{\ell=0}^{2}\mu_\ell b_{\ell,i}(\vec u)$ for $i=1,2,3$.
\end{lemma}
We define $\beta_i^{\mu}(\vec{u}) := 1-\alpha_i^{\mu}(\vec{u})$ for use below.

\subsection{The $K=3$ coordinate-update algorithm}
\label{sec:k3_building_block}

The optimal $\vec\mu^\ast$ is characterized coordinate-wise.
\begin{theorem}[\protect{\cite[Theorem~1]{dubey25}}]
\label{thm:mu_optimality_combined}
Under Assumptions~\ref{as:assumption3}, \ref{as:assumption4}, \ref{as:assumption5}, for fixed values of the other coordinates, a local minimizer $\mu^\ast=(\mu_0^\ast,\mu_1^\ast,\mu_2^\ast)$ of \eqref{eq:ld_mu} satisfies
\begin{align}
\label{eq:mu_0_optimality_condition}
\alpha &= 6 \int_Q \alpha_1^{(\mu_0^\ast,\mu_1,\mu_2)}(\vec{u}) d\vec{u},\\
\label{eq:mu_1_optimality_condition}
\alpha &= 1 - 2 \int_Q \beta_2^{(\mu_0,\mu_1^\ast,\mu_2)}(\vec{u})g(u_1) d\vec{u},\\
\label{eq:mu_2_optimality_condition}
\alpha &= 2 \int_Q \alpha_3^{(\mu_0,\mu_1,\mu_2^\ast)}(\vec{u}) g(u_1)g(u_2) d\vec{u}.
\end{align}
\end{theorem}
These conditions drive a coordinate-update algorithm~\citep[Algorithm~1]{dubey25}:
\begin{algorithm}[htbp]
\caption{\texttt{ComputeOptimalMu} (\citealp{dubey25}): coordinate update for optimal $\vec\mu$}
\label{alg:compute_optimal_mu_K3_main}
\begin{algorithmic}[1]
\Statex \textbf{Input:} level $\alpha$, tolerances $\delta,\varepsilon$, iteration cap $T_{\max}$, search ceiling $U_{\max}$.
\Statex \textbf{Init:} $\vec{\mu}^{(0)}\gets(0,0,0)$; $t\gets 0$.
\While{True}
\State $t\gets t+1$.
\State $\mu_0^{(t)}\gets \texttt{ComputeCoordinateMu}(F_0(\cdot;\mu_1^{(t-1)},\mu_2^{(t-1)}),\alpha,\delta,U_{\max})$.
\State $\mu_1^{(t)}\gets \texttt{ComputeCoordinateMu}(F_1(\cdot;\mu_0^{(t)},\mu_2^{(t-1)}),\alpha,\delta,U_{\max})$.
\State $\mu_2^{(t)}\gets \texttt{ComputeCoordinateMu}(F_2(\cdot;\mu_0^{(t)},\mu_1^{(t)}),\alpha,\delta,U_{\max})$.
\State \textbf{break} if $\|\vec\mu^{(t)}-\vec\mu^{(t-1)}\|_2\le\varepsilon$ or $t\ge T_{\max}$.
\EndWhile
\Statex \textbf{Output:} $\hat{\vec\mu}\gets\vec{\mu}^{(t)}$.
\end{algorithmic}
\end{algorithm}
with
\begin{equation}
\label{eq:f_gamma}
\begin{aligned}
F_0(\mu_0;\mu_1,\mu_2) &= 6\!\int_Q \alpha_1^{(\mu_0,\mu_1,\mu_2)}(\vec{u})d\vec{u}, \\
F_1(\mu_1;\mu_0,\mu_2) &= 1-2\!\int_Q \beta_2^{(\mu_0,\mu_1,\mu_2)}(\vec{u})g(u_1)d\vec{u}, \\
F_2(\mu_2;\mu_0,\mu_1) &= 2\!\int_Q \alpha_3^{(\mu_0,\mu_1,\mu_2)}(\vec{u})g(u_1)g(u_2)d\vec{u}.
\end{aligned}
\end{equation}
The subroutine \texttt{ComputeCoordinateMu} is given in Appendix~\ref{app:algos}.

\subsection{Regularity assumptions for Algorithm~\ref{alg:compute_optimal_mu_K3_main}}
\label{app:k3_notation}
\paragraph{Assumptions.}
To derive the optimality conditions for minimizing $L(\vec{D}^\mu, \mu)$ with respect to $\mu$, we first impose several regularity conditions on the alternative density function $g(\cdot)$.
The first assumption, a lower Lipschitz bound, ensures that the density $g(\cdot)$ is strictly monotonic and does not flatten out over any interval, which guarantees that the density responds at a controlled rate to changes in the $p$-value.
\begin{assumption}[Lower Lipschitz Bounds]
\label{as:assumption3}
There exists a constant $c_3 > 0$ such that for all $u, u' \in [0,1]$,
\begin{equation}
|g(u) - g(u')| \ge c_3 |u - u'|.
\end{equation}
\end{assumption}

The second assumption, strict positivity, requires that the alternative density be bounded away from zero. 
This is necessary to ensure that terms involving $g(u)$ in the Lagrangian do not vanish.
\begin{assumption}[Strict Positivity]
\label{as:assumption4}
There exists a constant $c_4 > 0$ such that for all $u \in [0,1]$,
\begin{equation}
g(u) \ge c_4.
\end{equation}
\end{assumption}

The third assumption provides a simple upper bound on the density function, which ensures all integrals are well-defined.
\begin{assumption}[Upper Bound]
\label{as:assumption5}
There exists a constant $c_5 > 0$ such that for all $u \in [0,1]$,
\begin{equation}\label{eq:assumption5}
g(u) \le c_5.
\end{equation}
\end{assumption}

Under these assumptions, we can characterize the minimizer of the Lagrangian $L(\vec{D}^\mu, \mu)$ by analyzing it in a coordinate-wise approach.

\subsection{Root-finding subroutine for Algorithm~\ref{alg:compute_optimal_mu_K3_main}}
\label{app:algos}
We detail the one-dimensional root-finding subroutine underpinning the coordinate-descent scheme of Algorithm~\ref{alg:compute_optimal_mu_K3_main} from~\citep{dubey25}.
Algorithm~\ref{alg:ComputeCoordinateMu} takes as input one constraint map $F_{\gamma}(\cdot;\mu_A,\mu_B)$, the other two multipliers, and a target level $\alpha$, and returns an approximate root of $F_{\gamma}(x;\mu_A,\mu_B)=\alpha$.
The routine combines expanding-interval bracketing with bisection, with explicit failure flags for the pathological cases $F_{\gamma}(0)<\alpha$ or no bracket found below $U_{\max}$.

\begin{algorithm}[htbp]
\caption{\texttt{ComputeCoordinateMu} (Subroutine for Algorithm~\ref{alg:compute_optimal_mu_K3_main})}
\label{alg:ComputeCoordinateMu}
\begin{algorithmic}[1]
    \Statex \textbf{Input:}
    \Statex \hspace{0.5em} 1: Functions $F_{\gamma}(x;\mu_{A},\mu_{B})$ ($\gamma \in \{0,1,2\}$, $x$ is the coordinate being solved for, and $\mu_{A},\mu_{B}$ are its other parameters).
    \Statex \hspace{0.5em} 2: Fixed value for parameter $\mu_{A}$: $\mu_{A}'$
    \Statex \hspace{0.5em} 3: Fixed value for parameter $\mu_{B}$: $\mu_{B}'$
    \Statex \hspace{0.5em} 4: Target FWER level $\alpha$.
    \Statex \hspace{0.5em} 5: Bisection parameters: tolerance $\delta$, max iterations $MaxIter_b$, initial step $U_s$, expansion factor $U_f$, max coordinate value $U_{max}$.

    \State $L \gets 0$.
    \State $\mu_{coord} \gets 0$. \Comment{Value for the current coordinate being computed}
    \State $flag \gets 0$. \Comment{0 for success, 1 for termination}
    \State $msg \gets \text{`'}$.
    \If{$F_{\gamma}(L; \mu_{A}', \mu_{B}') = \alpha$}
        \State $\mu_{coord} \gets L$. \Comment{$H(0)=\alpha$ met, optimal value is $0$.}
    \ElsIf{$F_{\gamma}(L; \mu_{A}', \mu_{B}') < \alpha$}
        \State $flag \gets 1$. \Comment{No optimal solution $H(0) < \alpha$ for current coordinate}
        \State $msg \gets \text{`Consider decreasing FWER level } \alpha\text{.'}$
        \State $\mu_{coord} \gets L$. \Comment{Value assigned, but main algorithm will terminate based on flag.}
    \Else \Comment{$F_{\gamma}(L; \mu_{A}', \mu_{B}') > \alpha$}
        \State $U \gets L + U_s$. 
        \While{$F_{\gamma}(U; \mu_{A}', \mu_{B}') > \alpha$ \textbf{and} $U < U_{max}$}
            \State $U \gets U \times U_f$. 
        \EndWhile
        \If{$F_{\gamma}(U; \mu_{A}', \mu_{B}') > \alpha$} 
            \State $flag \gets 1$. \Comment{No optimal solution: Failed to bracket positive root for current coordinate.}
            \State $msg \gets \text{`Consider increasing } U_{max} \text{ or decreasing FWER level $\alpha$} \text{.'}$
            \State $\mu_{coord} \gets L$. \Comment{Value assigned, but main algorithm will terminate based on flag.}
        \Else
            \For{$j \gets 1 \text{ to } MaxIter_b$} 
                \State $mid \gets L + (U-L)/2$.
                \If{$(U-L)/2 < \delta$} \textbf{break}; \EndIf 
                \If{$F_{\gamma}(mid; \mu_{A}', \mu_{B}') > \alpha$}
                    \State $L \gets mid$.
                \Else
                    \State $U \gets mid$.
                \EndIf
            \EndFor
            \State $\mu_{coord} \gets L + (U-L)/2$.
        \EndIf
    \EndIf
    \Statex \textbf{Output:}
    \Statex \hspace{0.5em} 1: Computed coordinate value $\mu_{coord}$.
    \Statex \hspace{0.5em} 2: Termination flag $flag$.
    \Statex \hspace{0.5em} 3: Termination message $msg$.
\end{algorithmic}
\end{algorithm}

\section{Extensions, instantiations, and the assumption hierarchy}
\label{app:deferred_main}

This appendix collects (i) scope and instantiation remarks for Theorem~\ref{thm:blockwise_strong_fwer_clean} under independence; (ii) the plug-in corollaries for sample splitting, the true-$g$ oracle reformulation, and explicit polynomial-in-$n$ rates; (iii) the bisection algorithm for the equalized-marginal allocation and its \v{S}id\'ak-constrained variant; (iv) the global \v{S}id\'ak-tightening corollary; (v) the containment sandwich on the unrestricted strong-FWER optimum; (vi) the assumption-hierarchy audit; and (vii) the verification of structural hypotheses for the canonical truncated-normal model. Each is referenced but not reproduced in the main body.

\subsection{Cross-block homogeneity assumption}
\label{app:exchangeable_blocks_assumption}

\begin{assumption}[Cross-block homogeneity under the global alternative]
\label{assump:exchangeable_blocks}
The per-block laws $\mathbb Q^{(b)}_{\!\vec h_K}$ in \eqref{eq:per_block_alt_law} all coincide with the canonical $K=3$ generative law $\mathbb P_{h_3}$ of~\citet{dubey25}; equivalently, $\pi_3^{(b)}\equiv\pi_3$ for all $b\in[B]$. (Under the global independence model \eqref{eq:independence} with a common within-block alternative density, this condition holds automatically.) The body theorems are stated for heterogeneous blocks; this assumption is invoked only to specialize their conclusions to the uniform split $\alpha/B$ and the closed-form value $\pi_3(\alpha/B)$.
\end{assumption}

\subsection{Scope of Theorem~\ref{thm:blockwise_strong_fwer_clean} and instantiations under independence}
\label{app:thm1_scope_instantiations}

\begin{remark}[Scope of Theorem~\ref{thm:blockwise_strong_fwer_clean}]
\label{rem:thm1_scope}
The theorem contains no independence assumption. Independence enters only to \emph{verify} Assumption~\ref{as:local_validity_marginal} in concrete models, via the two reductions below.
\end{remark}

\begin{remark}[Instantiation under independence, and weakenings]
\label{cor:indep_to_productlaw}
\label{cor:k3_optimizer_instantiation}
\label{rem:crossblock_only}
\label{rem:block_reduction_optional}
Under \eqref{eq:independence}, each block marginal $\mathbb Q^{(b)}_{h_\ell}$ equals the within-block product law, so the $K=3$ strong-FWER optimizer of~\citet{dubey25} (Algorithm~\ref{alg:compute_optimal_mu_K3_main}) discharges Assumption~\ref{as:local_validity_marginal} and Algorithm~\ref{alg:blockwise_k3} controls strong FWER at level $\alpha$, provided the within-block alternative density $g$ satisfies the dubey25 regularity conditions (Assumptions~\ref{as:assumption3}--\ref{as:assumption5}: lower-Lipschitz, strict positivity, upper bound), verified for the canonical truncated-normal design in Proposition~\ref{prop:canonical_model_verification}(iii)--(iv). Two weakenings: (i) cross-block independence alone preserves Theorem~\ref{thm:blockwise_strong_fwer_clean} provided each $\vec D^{(b)}$ is valid under the actual within-block joint marginal; (ii) under $\vec h$-exchangeability with a symmetric local rule, verification of~\eqref{eq:block_strong_fwer_assump_clean} reduces to the three structured configurations $\eta=h_j$, $j=0,1,2$.
\end{remark}

\subsection{Plug-in corollaries: sample splitting, oracle reformulation, and explicit rates}
\label{app:plugin_corollaries}

\begin{corollary}[Cross-block sample splitting under independence]
\label{cor:plugin_sample_splitting}
Assume Assumption~\ref{as:block_measurability}, the cross-block independence assumption A3 (equivalently, the cross-block half of \eqref{eq:independence}), and that the estimation-fold blocks $\{X^{(b)}\}_{b\in\mathcal E}$ are independent of the testing-fold blocks $\{X^{(b)}\}_{b\in\mathcal T}$. Let $\widehat g=\widehat g(\{X^{(b)}\}_{b\in\mathcal E})$ be any measurable estimator computed solely from the estimation fold. Then $\widehat g$ satisfies condition~(i) of Theorem~\ref{thm:plugin_fwer}, and the conclusion~\eqref{eq:plugin_fwer_bound} holds with the $\sigma$-algebra $\mathcal G:=\sigma(\{X^{(b)}\}_{b\in\mathcal E})$. The \v{S}id\'ak-budget variant $\prod_{b\in\mathcal T}(1-\alpha_{\mathrm{blk}}^{(b)})\ge 1-\alpha$ replaces the Bonferroni budget in (ii) without changing the bound, since the $B_{\mathcal T}$ factor is driven by the Boole/union step over $\mathcal T$ and applies equally to both budgets.
\end{corollary}

The proof is provided in Appendix~\ref{app:proof_plugin_sample_splitting}.

\begin{corollary}[Realized power versus the true-$g$ oracle]
\label{cor:plugin_power_oracle}
Under the setup of Theorem~\ref{thm:plugin_power}, suppose the per-block levels obey $\alpha_{\mathrm{blk}}^{(b)}\ge L_3\,\|g-\widehat g\|_{\infty}$ almost surely for every $b\in\mathcal T$. Then
\begin{equation}
\label{eq:plugin_power_vs_oracle}
\frac{1}{B_{\mathcal T}}\sum_{b\in\mathcal T}\mathbb E\!\left[\mathrm{Power}_{g}^{(b)}\!\bigl(\widehat{\vec D}^{(b)}\bigr)\right]
\;\ge\;\frac{1}{B_{\mathcal T}}\sum_{b\in\mathcal T}\mathbb E\!\left[\pi_{3}^{g}\!\bigl(\alpha_{\mathrm{blk}}^{(b)}-L_{3}\,\|g-\widehat g\|_{\infty}\bigr)\right]
\;-\;2L_{3}^{\mathrm{pow}}\,\mathbb E\!\left[\|g-\widehat g\|_{\infty}\right].
\end{equation}
That is, on average per hypothesis the plug-in procedure attains the \emph{true-$g$ oracle's} value at a deflated level $\alpha_{\mathrm{blk}}^{(b)}-L_{3}\,\|g-\widehat g\|_{\infty}$, with an additional additive penalty linear in $\mathbb E\|g-\widehat g\|_{\infty}$ and free of any $B_{\mathcal T}$ factor.
\end{corollary}

The proof is provided in Appendix~\ref{app:proof_plugin_power_oracle}.

\begin{corollary}[Plug-in rates under classical sup-norm bounds on $\widehat g$]
\label{cor:plugin_rates}
Assume the setup of Theorem~\ref{thm:plugin_fwer} and Theorem~\ref{thm:plugin_power}, and let $n_{\mathcal E}$ denote the total number of estimation-fold $p$-values available to fit $\widehat g$. Suppose the density estimator satisfies the expected sup-norm rate
\begin{equation}
\label{eq:plugin_rate_hyp}
\mathbb E\!\left[\|g-\widehat g\|_\infty\right]\;=\;O(r_{n_{\mathcal E}})\qquad\text{as}\qquad n_{\mathcal E}\to\infty,
\end{equation}
for some deterministic sequence $r_{n}\downarrow 0$. Then the plug-in rule $\vec D^{\mathrm{blk},\mathcal T}$ of Algorithm~\ref{alg:plugin_blockwise_k3} obeys, uniformly in $\ell\in\{0,\ldots,K-1\}$,
\begin{align}
\label{eq:plugin_rates_fwer}
\mathrm{FWER}_\ell\!\bigl(\vec D^{\mathrm{blk},\mathcal T}\bigr) &\;\le\;\alpha\;+\;O\!\bigl(B_{\mathcal T}\,r_{n_{\mathcal E}}\bigr),\\
\label{eq:plugin_rates_pow}
\frac{1}{3B_{\mathcal T}}\sum_{b\in\mathcal T}\mathbb E\!\left[\mathrm{Power}_{g}^{(b)}\!\bigl(\widehat{\vec D}^{(b)}\bigr)\right] &\;\ge\;\frac{1}{3B_{\mathcal T}}\sum_{b\in\mathcal T}\mathbb E\!\left[\pi_3^{\widehat g}\!\bigl(\alpha_{\mathrm{blk}}^{(b)}\bigr)\right]\;-\;O\!\bigl(r_{n_{\mathcal E}}\bigr),
\end{align}
with the constants in $O(\cdot)$ depending only on $M:=\max(\|g\|_\infty,\|\widehat g\|_\infty)$. Two canonical instantiations:
\begin{enumerate}
\item[(i)] \emph{Grenander / monotone density estimator.} If $g$ is non-increasing on $[0,1]$ (automatic under Proposition~\ref{prop:canonical_model_verification}(iii)) and $\widehat g$ is the sample-splitting Grenander estimator, then $r_{n_{\mathcal E}}=\bigl(n_{\mathcal E}^{-1}\log n_{\mathcal E}\bigr)^{1/3}$~\citep{groeneboom2014nonparametric}; the FWER excess is $O\!\bigl(B_{\mathcal T}(n_{\mathcal E}^{-1}\log n_{\mathcal E})^{1/3}\bigr)$ and the average-power deficit is $O\!\bigl((n_{\mathcal E}^{-1}\log n_{\mathcal E})^{1/3}\bigr)$.
\item[(ii)] \emph{H\"older-smooth alternatives via kernel density estimation.} If $g$ belongs to a H\"older class of smoothness $s\ge 1$ on $[0,1]$ and $\widehat g$ is a bandwidth-optimal kernel density estimator fit on the estimation fold, then $r_{n_{\mathcal E}}=\bigl(n_{\mathcal E}^{-1}\log n_{\mathcal E}\bigr)^{s/(2s+1)}$~\citep{tsybakov2009book}; the FWER excess is $O\!\bigl(B_{\mathcal T}(n_{\mathcal E}^{-1}\log n_{\mathcal E})^{s/(2s+1)}\bigr)$ and the average-power deficit is $O\!\bigl((n_{\mathcal E}^{-1}\log n_{\mathcal E})^{s/(2s+1)}\bigr)$.
\end{enumerate}
In either regime, the plug-in procedure is asymptotically $\alpha$-level exact and asymptotically matches the $\widehat g$-oracle value $\pi_3^{\widehat g}(\alpha_{\mathrm{blk}}^{(b)})$ in per-hypothesis power whenever $B_{\mathcal T}\,r_{n_{\mathcal E}}\to 0$.
\end{corollary}

The proof is provided in Appendix~\ref{app:proof_plugin_rates}.

\subsection{Bisection algorithm for the equalized-marginal allocation}
\label{app:bisection_allocation_statement}

\begin{prop}[Bisection algorithm for the equalized-marginal allocation]
\label{prop:bisection_allocation}
Under Assumption~\ref{assump:concavity}, define the right-inverse of $\partial_+\pi_3^{(b)}$ on $[0,\alpha]$ by
\begin{equation}
\label{eq:right_inverse_def}
g_b(\mu)\;:=\;\sup\{\alpha'\in[0,\alpha]:\partial_+\pi_3^{(b)}(\alpha')\ge\mu\},\qquad \mu\ge 0,\quad b\in[B],
\end{equation}
with the convention $\sup\emptyset=0$. Each $g_b$ is non-increasing and right-continuous on $[0,\infty)$, with $g_b(0)=\alpha$ and $g_b(\mu)\to 0$ as $\mu\to\infty$. The function $f(\mu):=\sum_{b=1}^{B}g_b(\mu)$ is non-increasing, $f(0)=B\alpha\ge\alpha$, and $f(\mu)\to 0$, so the equation $f(\mu^{*})=\alpha$ has a (possibly non-unique) solution $\mu^{*}\in[0,\overline\mu]$ with $\overline\mu:=\max_{b}\partial_+\pi_3^{(b)}(0^{+})<\infty$ (finiteness of $\overline\mu$ is automatic under A6: bounded $\|g\|_\infty\le M<\infty$ in Proposition~\ref{prop:canonical_model_verification}(iv) makes each $\pi_3^{(b)}$ Lipschitz on $[0,\alpha]$). Setting $\alpha_b^{*}:=g_b(\mu^{*})$ (with any deterministic tiebreak when $g_b$ is multivalued at $\mu^{*}$, scaled to satisfy $\sum_b\alpha_b^{*}=\alpha$) yields a maximizer of \eqref{eq:hetero_alloc_max}.

A bisection on $\mu\in[0,\overline\mu]$ converges to $\mu^{*}$ to absolute precision $\varepsilon>0$ in $\lceil\log_2(\overline\mu/\varepsilon)\rceil$ iterations, each of which requires $O(B)$ derivative-inversion queries $g_b(\cdot)$. The total complexity is therefore $O(B\log(1/\varepsilon))$, plus the (block-independent) cost of each $g_b$ query, which is itself $O(\log(1/\varepsilon))$ when $\partial_+\pi_3^{(b)}$ is accessed by bisection through evaluations of $\pi_3^{(b)}$ via the $K=3$ optimizer of~\citet{dubey25}.
\end{prop}

The proof is provided in Appendix~\ref{app:proof_bisection_allocation}.

\subsection{\v{S}id\'ak refinements and the containment sandwich}
\label{app:sidak_and_containment}

\begin{remark}[Bisection under the \v{S}id\'ak constraint]
\label{rem:sidak_bisection}
The bisection of Proposition~\ref{prop:bisection_allocation} extends to the \v{S}id\'ak-constrained problem \eqref{eq:sidak_alloc_max} of Theorem~\ref{thm:global_opt_sep_allocation_sidak} without additional cost. Replacing the Bonferroni right-inverse \eqref{eq:right_inverse_def} by
\begin{equation}
\label{eq:right_inverse_sidak}
g_b^{\mathrm{ind}}(\mu)\;:=\;\sup\bigl\{\alpha'\in[0,1]:(1-\alpha')\,\partial_+\pi_3^{(b)}(\alpha')\ge\mu\bigr\},\qquad \mu\ge 0,\;b\in[B],
\end{equation}
and replacing the sum $f(\mu)=\sum_b g_b(\mu)$ by $f^{\mathrm{ind}}(\mu):=\sum_b-\log(1-g_b^{\mathrm{ind}}(\mu))$, the \v{S}id\'ak shadow price $\mu^{*}$ is the unique solution to $f^{\mathrm{ind}}(\mu^{*})=-\log(1-\alpha)$; each $g_b^{\mathrm{ind}}$ is non-increasing and right-continuous on $[0,\infty)$ (product of the strictly decreasing $(1-\alpha)$ and the non-increasing $\partial_+\pi_3^{(b)}$ yields a non-increasing integrand), and bisection on $\mu$ converges in $\lceil\log_2(\overline\mu^{\mathrm{ind}}/\varepsilon)\rceil$ iterations with $\overline\mu^{\mathrm{ind}}:=\max_b\partial_+\pi_3^{(b)}(0^{+})$. The total complexity is $O(B\log(1/\varepsilon))$, identical to the Bonferroni case.
\end{remark}

\begin{corollary}[Global optimality within $\mathfrak D_{\mathrm{sep}}^{\mathrm{ind}}$ under cross-block independence]
\label{cor:sidak_blockwise_optimality}
Suppose cross-block independence holds as in Theorem~\ref{thm:global_opt_sep_allocation_sidak} and Assumption~\ref{assump:concavity} is in force. Then
\begin{equation}
\label{eq:sidak_sup_value}
\sup_{D^{\mathrm{sep}}\in\mathfrak D_{\mathrm{sep}}^{\mathrm{ind}}}\Pi_K\!\bigl(D^{\mathrm{sep}}\bigr)
\;=\;\frac{1}{B}\sum_{b=1}^{B}\pi_3^{(b)}(\alpha_b^{*,\mathrm{ind}})
\;=:\;\Pi_K^{*,\mathrm{ind}}
\;\ge\;\Pi_K^{*},
\end{equation}
where $\vec\alpha^{*,\mathrm{ind}}$ is the weighted equalized-marginal allocation of Theorem~\ref{thm:global_opt_sep_allocation_sidak}, computable via Remark~\ref{rem:sidak_bisection} in $O(B\log(1/\varepsilon))$ operations; the supremum is attained by the procedure that combines $\vec\alpha^{*,\mathrm{ind}}$ with the per-block $K=3$ strong-FWER optimizer of~\citet{dubey25} at level $\alpha_b^{*,\mathrm{ind}}$ applied to $\mathbb Q^{(b)}_{\!\vec h_K}$. In the homogeneous special case (Assumption~\ref{assump:exchangeable_blocks}), $\Pi_K^{*,\mathrm{ind}}=\pi_3(1-(1-\alpha)^{1/B})$ and the optimal allocation is the uniform \v{S}id\'ak split $\alpha_{\mathrm{blk}}^{(b)}=1-(1-\alpha)^{1/B}$, strictly dominating the Bonferroni split $\alpha/B$ for $B>1$ and any $\alpha\in(0,1)$ on which $\pi_3$ is strictly increasing.
\end{corollary}

\begin{proof}
By Theorem~\ref{thm:global_opt_sep_within_block}, every $D^{\mathrm{sep}}\in\mathfrak D_{\mathrm{sep}}^{\mathrm{ind}}$ with allocation $\vec\alpha$ satisfies $\Pi_K(D^{\mathrm{sep}})\le(1/B)\sum_b\pi_3^{(b)}(\alpha_b)$ (the within-block bound uses only per-block strong-FWER feasibility, which is part of the definition of $\mathfrak D_{\mathrm{sep}}^{\mathrm{ind}}$; no use is made of the shape of the allocation constraint). Maximizing the right side under $\prod_b(1-\alpha_b)\ge 1-\alpha$ gives $M^{\mathrm{ind}}$ by Theorem~\ref{thm:global_opt_sep_allocation_sidak}; the procedure described attains this bound, yielding equality. The inequality $\Pi_K^{*,\mathrm{ind}}\ge\Pi_K^{*}$ follows because $\mathfrak D_{\mathrm{sep}}\subseteq\mathfrak D_{\mathrm{sep}}^{\mathrm{ind}}$. The homogeneous reduction is the corresponding case of Theorem~\ref{thm:global_opt_sep_allocation_sidak}.
\end{proof}

\begin{prop}[Containment sandwich on the unrestricted strong-FWER optimum]
\label{prop:containment}
Suppose the global independence model \eqref{eq:independence} holds with within-block iid alternative density $g$ satisfying the model class \eqref{eq:pvalue_model} (in particular $g$ is non-increasing on $[0,1]$, so $G(\alpha):=\int_{0}^{\alpha}g(u)\,du$ is concave on $[0,1]$). Define the unrestricted strong-FWER optimum
\begin{equation}
\label{eq:Pi_unr_def}
\begin{aligned}
\Pi_K^{\mathrm{unr}}(\alpha)
\;:=\;\sup\bigl\{\Pi_K(\vec D)\,:\,&\vec D:\mathcal X\to\{0,1\}^K\text{ measurable},\\
&\mathrm{FWER}_\ell(\vec D)\le\alpha\;\,\forall\,\vec h_\ell\text{ with }|\vec h_\ell|=\ell,\ \ell=0,\ldots,K-1\bigr\},
\end{aligned}
\end{equation}
where the strong-FWER constraint is imposed under \emph{every} configuration $\vec h_\ell\in\{0,1\}^K$ with exactly $\ell$ alternatives (worst-case over the assignment of nulls to coordinates). Then, under Assumption~\ref{assump:exchangeable_blocks} (cross-block homogeneity, automatic under \eqref{eq:independence} with a common within-block alternative density),
\begin{equation}
\label{eq:sandwich}
G\!\bigl(1-(1-\alpha)^{1/K}\bigr)
\;\le\;\Pi_K^{*,\mathrm{ind}}(\alpha)
\;=\;\pi_3\!\bigl(1-(1-\alpha)^{1/B}\bigr)
\;\le\;\Pi_K^{\mathrm{unr}}(\alpha)
\;\le\;G(\alpha),
\end{equation}
with $\Pi_K^{*,\mathrm{ind}}(\alpha)$ as in Corollary~\ref{cor:sidak_blockwise_optimality}. The lower bound is the per-hypothesis power of global \v{S}id\'ak thresholding at level $\alpha$; the upper bound is the Neyman--Pearson per-marginal ceiling at level $\alpha$. The middle equality is Corollary~\ref{cor:sidak_blockwise_optimality} in the homogeneous case.
\end{prop}

\begin{proof}
\emph{Lower bound.} The within-block \v{S}id\'ak rule $\vec D^{\mathrm{Sid}}$ that rejects $u_k\le c$ with $c=1-(1-\alpha_b)^{1/3}$ at block level $\alpha_b=1-(1-\alpha)^{1/B}$ controls strong FWER at level $\alpha_b$ within the block under within-block independence: for every within-block configuration $\eta\in\{0,1\}^3$ with at least one null, $\Pr(\bigcup_{i:\eta_i=0}\{u_i\le c\})\le 1-(1-c)^3=\alpha_b$. It is therefore feasible for the $K=3$ value problem at level $\alpha_b$, so by definition of $\pi_3$, $\pi_3(\alpha_b)\ge\Pi_3^{\mathbb P_{h_3}}(\vec D^{\mathrm{Sid}})=G(c)$. Substituting $c=1-(1-\alpha)^{1/(3B)}=1-(1-\alpha)^{1/K}$ and using Corollary~\ref{cor:sidak_blockwise_optimality} yields $\Pi_K^{*,\mathrm{ind}}(\alpha)=\pi_3(\alpha_b)\ge G(1-(1-\alpha)^{1/K})$. The middle inequality $\Pi_K^{*,\mathrm{ind}}(\alpha)\le\Pi_K^{\mathrm{unr}}(\alpha)$ holds because every $D^{\mathrm{sep}}\in\mathfrak D_{\mathrm{sep}}^{\mathrm{ind}}$ is feasible for \eqref{eq:MHTK} (Theorem~\ref{thm:global_opt_sep_allocation_sidak}, feasibility part).

\emph{Upper bound.} Fix any $\vec D$ feasible for the worst-case version of \eqref{eq:MHTK} (in the sense made explicit in \eqref{eq:Pi_unr_def}) at level $\alpha$ and any $k\in[K]$. Let $\vec h_{(k)}\in\{0,1\}^K$ be the configuration with $(h_{(k)})_k=0$ and $(h_{(k)})_j=1$ for $j\ne k$ (one true null at coordinate $k$, $K-1$ alternatives elsewhere); this configuration has $\ell=K-1$ alternatives, so the strong-FWER constraint at level $\alpha$ binds under it and gives
\[
\Pr_{\vec h_{(k)}}\!\bigl(D_k(\vec X)=1\bigr)
\;\le\;\Pr_{\vec h_{(k)}}\!\Bigl(\bigcup_{j:(h_{(k)})_j=0}\{D_j(\vec X)=1\}\Bigr)
\;=\;\mathrm{FWER}_{K-1}(\vec D;\vec h_{(k)})\;\le\;\alpha.
\]
Under the global independence model \eqref{eq:independence}, the conditional law of $u_k$ given $\vec u_{-k}$ is uniform on $[0,1]$ under $\vec h_{(k)}$ and has density $g$ under $\vec h_K$, while $\vec u_{-k}$ has the same marginal law $g^{\otimes(K-1)}$ under both. Define the conditional Type-I-error function $\beta_k(\vec u_{-k}):=\Pr(D_k(\vec X)=1\mid \vec u_{-k},\,u_k\sim\mathrm{Unif}[0,1])$. Then
\[
\mathbb E_{\vec u_{-k}\sim g^{\otimes(K-1)}}\!\bigl[\beta_k(\vec u_{-k})\bigr]
\;=\;\Pr_{\vec h_{(k)}}\!\bigl(D_k(\vec X)=1\bigr)\;\le\;\alpha.
\]
Conditional on $\vec u_{-k}$, $D_k(\vec X)$ is a (possibly randomized) test of the simple hypothesis $u_k\sim\mathrm{Unif}[0,1]$ versus $u_k\sim g$ at conditional Type I error $\beta_k(\vec u_{-k})$. Since $g$ is non-increasing on $[0,1]$, the likelihood ratio $g(u)/1$ is non-increasing in $u$, so the Neyman--Pearson most powerful test rejects $\{u_k\le \beta_k(\vec u_{-k})\}$ with conditional power $G(\beta_k(\vec u_{-k}))$. Hence
\[
\Pr_{\vec h_K}\!\bigl(D_k(\vec X)=1\mid \vec u_{-k}\bigr)\;\le\;G(\beta_k(\vec u_{-k})).
\]
Marginalizing,
\[
\Pr_{\vec h_K}\!\bigl(D_k(\vec X)=1\bigr)
\;\le\;\mathbb E[G(\beta_k(\vec u_{-k}))]
\;\le\;G\!\bigl(\mathbb E[\beta_k(\vec u_{-k})]\bigr)
\;\le\;G(\alpha),
\]
the second inequality from Jensen applied to the concave $G$, the third from monotonicity of $G$. Averaging over $k\in[K]$ gives $\Pi_K(\vec D)\le G(\alpha)$, and taking the supremum yields $\Pi_K^{\mathrm{unr}}(\alpha)\le G(\alpha)$.
\end{proof}

\begin{remark}[Where Algorithm~\ref{alg:blockwise_k3} sits inside the sandwich]
\label{rem:empirical_containment}
The bound \eqref{eq:sandwich} is informative at both ends: $G(1-(1-\alpha)^{1/K})$ is the per-hypothesis power of global \v{S}id\'ak thresholding (which Algorithm~\ref{alg:blockwise_k3} strictly dominates whenever $\pi_3(\alpha_b)>G(1-(1-\alpha)^{1/K})$, i.e., whenever the $K=3$ strong-FWER optimizer beats \v{S}id\'ak on a triple), while $G(\alpha)$ is the per-marginal Neyman--Pearson ceiling that no strong-FWER procedure can cross. The simulations of Section~\ref{sec:simulations} place Algorithm~\ref{alg:blockwise_k3} clearly above the standard practical strong-FWER baselines (Bonferroni, Holm, Hochberg, Hommel, \v{S}id\'ak step-down) under the truncated-normal model, with the largest separation at moderate signal strengths; the corresponding numerical comparison to the closed-form sandwich bounds is reported alongside each figure. Whether the gap to $\Pi_K^{\mathrm{unr}}(\alpha)$ closes exactly under the model class \eqref{eq:pvalue_model} is left open.
\end{remark}

\begin{remark}[Scope and limits of the $\mathfrak D_{\mathrm{sep}}$-optimality claim]
\label{rem:sep_scope}
Corollary~\ref{cor:blockwise_sep_optimality} establishes optimality \emph{within} $\mathfrak D_{\mathrm{sep}}$, not among all feasible procedures for \eqref{eq:MHTK}: a non-block-separable procedure (e.g., one that pools evidence across blocks) might in principle outperform $\Pi_K^{*}$. Within $\mathfrak D_{\mathrm{sep}}$ itself, Theorem~\ref{thm:global_opt_sep_within_block} is the load-bearing step: for any allocation, the within-block subproblem reduces exactly to the $K=3$ problem against the per-block law $\mathbb Q^{(b)}_{\!\vec h_K}$, with no concavity or homogeneity assumption required at all. Theorem~\ref{thm:global_opt_sep_allocation} then adds the allocation step under concavity (which holds automatically under \eqref{eq:pvalue_model} by Lemma~\ref{lem:pi3_concavity}) and yields an equalized-marginal characterization that handles heterogeneous blocks via a $O(B\log(1/\varepsilon))$ bisection (Proposition~\ref{prop:bisection_allocation}). The homogeneous case (cross-block-identical $\pi_3^{(b)}$) collapses to the uniform split $\alpha/B$ as a special case. When cross-block independence is available, Theorem~\ref{thm:global_opt_sep_allocation_sidak} and Corollary~\ref{cor:sidak_blockwise_optimality} further enlarge the admissible class from $\mathfrak D_{\mathrm{sep}}$ to $\mathfrak D_{\mathrm{sep}}^{\mathrm{ind}}$ via the \v{S}id\'ak budget, strictly improving the attainable power at the same $O(B\log(1/\varepsilon))$ cost (Remark~\ref{rem:sidak_bisection}). The remaining containment question, whether the unrestricted FWER-optimum already lies in $\mathfrak D_{\mathrm{sep}}$ (or $\mathfrak D_{\mathrm{sep}}^{\mathrm{ind}}$), is partially addressed by Proposition~\ref{prop:containment}, which sandwiches $\Pi_K^{\mathrm{unr}}(\alpha)$ between the global \v{S}id\'ak floor $G(1-(1-\alpha)^{1/K})$ and the per-marginal Neyman--Pearson ceiling $G(\alpha)$, with $\Pi_K^{*,\mathrm{ind}}(\alpha)$ pinned in between; closing the gap to a strict equality remains open (Remark~\ref{rem:empirical_containment}).
\end{remark}

\begin{remark}[Assumption hierarchy and where each result uses it]
\label{rem:assumption_hierarchy}
The development above rests on a graded set of conditions. Recording which layer each statement actually requires makes the model-dependent scope of every claim transparent.
\begin{enumerate}
\item[(A1)] \emph{Block-measurability of the decision rule} (Assumption~\ref{as:block_measurability}). A structural property of the blockwise construction, automatic for any procedure that forms its block decisions from block data alone. Used in Theorem~\ref{thm:blockwise_strong_fwer_clean}.
\item[(A2)] \emph{Local strong-FWER validity under the within-block marginal} (Assumption~\ref{as:local_validity_marginal}, cf.\ \eqref{eq:block_marginal_law}--\eqref{eq:block_strong_fwer_assump_clean}). The sole nontrivial hypothesis of Theorem~\ref{thm:blockwise_strong_fwer_clean}; it does not by itself invoke any form of independence.
\item[(A3)] \emph{Cross-block independence} (the cross-block half of \eqref{eq:independence}). It is the sole extra hypothesis needed to tighten the Bonferroni allocation budget $\sum_b\alpha_b\le\alpha$ to the \v{S}id\'ak budget $\prod_b(1-\alpha_b)\ge 1-\alpha$ (Theorem~\ref{thm:global_opt_sep_allocation_sidak}, Corollary~\ref{cor:sidak_blockwise_optimality}), strictly enlarging the feasible class and improving attainable power. It is also one of the two ingredients of the full independence model needed by Remark~\ref{cor:indep_to_productlaw}; the other is A4.
\item[(A4)] \emph{Within-block independence} (the within-block half of \eqref{eq:independence}). Combined with A3, it gives the full independence assumption~\eqref{eq:independence}, under which Remark~\ref{cor:indep_to_productlaw} re-expresses Assumption~\ref{as:local_validity_marginal} against the product law $\mathbb{P}_{b,\eta^{(b)}(\ell)}$, and Remark~\ref{cor:k3_optimizer_instantiation} instantiates A2 via the $K=3$ optimizer of~\citet{dubey25}. (Cross-block independence alone is insufficient: marginalizing the global product law over coordinates outside $\mathcal B_b$ yields a within-block product only when the joint factorizes within blocks too. See Remark~\ref{rem:crossblock_only} for what cross-block independence by itself buys via Theorem~\ref{thm:blockwise_strong_fwer_clean} working directly against the within-block joint marginal.)
\item[(A5)] \emph{Within-block $\vec h$-exchangeability together with a symmetric local rule.} An optional computational simplification used only in Remark~\ref{rem:block_reduction_optional}: it reduces verification of~\eqref{eq:block_strong_fwer_assump_clean} to three structured configurations. Not required elsewhere.
\item[(A6)] \emph{Absolute continuity of each per-block law $\mathbb Q^{(b)}_{\!\vec h_K}$} (automatic under \eqref{eq:pvalue_model}). This is the hypothesis of Lemma~\ref{lem:pi3_concavity}, which proves concavity (and monotonicity) of each per-block value function $\pi_3^{(b)}$ from first principles via randomization closure plus the Lyapunov--Halmos convexity theorem.
\item[(A7)] \emph{Concavity of each $\pi_3^{(b)}$ on $[0,\alpha]$} (Assumption~\ref{assump:concavity}). Discharged by A6 in our standard regularity class, and required for the equalized-marginal/KKT statement of Theorem~\ref{thm:global_opt_sep_allocation} and the bisection algorithm of Proposition~\ref{prop:bisection_allocation}.
\item[(A8)] \emph{Cross-block homogeneity matching the $\pi_3$ reference law} (Assumption~\ref{assump:exchangeable_blocks}). \emph{No longer load-bearing} for either Theorem~\ref{thm:global_opt_sep_within_block} or Theorem~\ref{thm:global_opt_sep_allocation}; both hold for heterogeneous blocks. A8 is invoked only to specialize the equalized-marginal optimum to the uniform split $\alpha/B$ and the closed-form value $\pi_3(\alpha/B)$.
\end{enumerate}
In summary: Theorem~\ref{thm:blockwise_strong_fwer_clean} needs only A1--A2; Remark~\ref{cor:indep_to_productlaw} adds A3 \emph{and} A4 (the full independence model); Remark~\ref{cor:k3_optimizer_instantiation} uses A1--A2 together with \eqref{eq:independence} (i.e., A3 and A4); Theorem~\ref{thm:global_opt_sep_within_block} needs no concavity, homogeneity, or independence beyond what is implicit in $\mathfrak D_{\mathrm{sep}}$ feasibility (A2); Theorem~\ref{thm:global_opt_sep_allocation} and Proposition~\ref{prop:bisection_allocation} need only A7 (discharged by A6); Corollary~\ref{cor:blockwise_sep_optimality} combines A2 and A7 for the heterogeneous optimum, with A8 yielding the homogeneous uniform-split specialization. Theorem~\ref{thm:global_opt_sep_allocation_sidak} and Corollary~\ref{cor:sidak_blockwise_optimality} additionally require A3 alone (the \v{S}id\'ak tightening uses only cross-block independence; within-block dependence is allowed).
\end{remark}

\begin{prop}[The canonical truncated-normal model discharges all structural hypotheses]
\label{prop:canonical_model_verification}
Suppose the global independence model \eqref{eq:independence} holds with null and alternative $p$-value laws drawn from \eqref{eq:pvalue_model}, where the alternative density $g$ is induced by a one-sided truncated-normal design on $[-M^\star,M^\star]$ with location $\theta<0$ and unit variance (as in Section~\ref{subsec:sim1_truncnorm}). Then:
\begin{enumerate}
\item[(i)] \emph{A1 (block-measurability).} Any blockwise rule $\vec D^{\mathrm{blk}}$ built from per-block decisions $\vec D^{(b)}(X^{(b)})$ satisfies Assumption~\ref{as:block_measurability} by construction.
\item[(ii)] \emph{A6 (absolute continuity of per-block laws).} For every $b\in[B]$ and every $\ell\in\{0,\ldots,K-1\}$ the per-block law $\mathbb Q^{(b)}_{h_\ell}$ admits a Lebesgue density on the ordered simplex $Q$ given by $\rho^{(b)}_{\eta^{(b)}(\ell)}(u)=3!\prod_{i:\,\eta^{(b)}_i(\ell)=1}g(u_i)\cdot\mathbf 1_Q(u)$, hence is absolutely continuous. In particular, Assumption~\ref{assump:concavity} is discharged via Lemma~\ref{lem:pi3_concavity}.
\item[(iii)] \emph{$g$ non-increasing.} On $[0,1]$, $g(u)=(Z_0/Z_1)\,\exp\!\bigl(\theta\,x(u)-\tfrac12\theta^2\bigr)$ with $x(u):=(F_0^T)^{-1}(u)$ strictly increasing; since $\theta<0$, $g$ is strictly decreasing on $[0,1]$, so $G(\alpha)=\int_0^\alpha g$ is strictly concave. This is the hypothesis invoked in Proposition~\ref{prop:containment}.
\item[(iv)] \emph{Bounded density.} $M:=\|g\|_\infty=g(0)=(Z_0/Z_1)\exp(-\theta M^\star-\tfrac12\theta^2)<\infty$, so the Lipschitz constants $L_3=2M$ and $L_3^{\mathrm{pow}}=3M^2$ of Lemmas~\ref{lem:fwer_perturbation}--\ref{lem:power_perturbation} are finite.
\end{enumerate}
Consequently, all hypotheses appearing in the chain Theorem~\ref{thm:blockwise_strong_fwer_clean}--Corollary~\ref{cor:sidak_blockwise_optimality} and Proposition~\ref{prop:containment} are automatically in force for the canonical truncated-normal design; the only load-bearing local hypothesis not implied by \eqref{eq:pvalue_model} is A2 (Assumption~\ref{as:local_validity_marginal}), which is discharged in turn by instantiating each $\vec D^{(b)}$ as the $K=3$ optimizer of~\citet{dubey25} (Remark~\ref{cor:k3_optimizer_instantiation}).
\end{prop}

\begin{proof}
\emph{(i)} If $\vec D^{(b)}$ depends only on $X^{(b)}$, then $\{V^{(b)}>0\}\in\sigma(X^{(b)})$ by composition. \emph{(ii)} Under \eqref{eq:independence}, the unordered block $X^{(b)}=(u^{(b)}_1,u^{(b)}_2,u^{(b)}_3)$ has product density $\prod_i \tilde g_i(u^{(b)}_i)$ with $\tilde g_i=\mathbf 1_{[0,1]}$ if $\eta^{(b)}_i=0$ and $\tilde g_i=g$ if $\eta^{(b)}_i=1$. Mapping to the ordered simplex $Q=\{u_1\le u_2\le u_3\}$ of volume $1/3!$ multiplies by the symmetry factor $3!$, giving the stated density; this is finite Lebesgue-a.e.\ since $g\in L^\infty([0,1])$ by (iv), so $\mathbb Q^{(b)}_{h_\ell}\ll\mathrm{Leb}_{|Q}$. \emph{(iii)} Write $u=F_0^T(x)$ with $F_0^T$ the truncated-$\mathcal N(0,1)$ cdf on $[-M^\star,M^\star]$; the likelihood ratio of the truncated-$\mathcal N(\theta,1)$ density against the truncated-$\mathcal N(0,1)$ density at $x$ equals $(Z_0/Z_1)\exp(\theta x-\tfrac12\theta^2)$, and this is $g(u)$ by change of variables. Differentiating in $u$, $g'(u)=(Z_0/Z_1)\theta\exp(\theta x-\tfrac12\theta^2)/f_0^T(x)$ with $f_0^T(x)>0$ on $(-M^\star,M^\star)$; since $\theta<0$, $g'(u)<0$, so $g$ is strictly decreasing and $G$ strictly concave. \emph{(iv)} $g$ is monotone decreasing on $[0,1]$, so $\|g\|_\infty=g(0)=(Z_0/Z_1)\exp(-\theta M^\star-\tfrac12\theta^2)$, finite because $M^\star,\theta$ are finite.
\end{proof}

\section{Plug-in perturbation lemmas}
\label{app:plugin_lemmas}

This appendix collects the three perturbation lemmas and the remark on the constant $M$ that support the plug-in statements of Section~\ref{subsec:plugin} (Theorems~\ref{thm:plugin_fwer}--\ref{thm:plugin_power}).

\begin{lemma}[FWER perturbation bound in $g$]
\label{lem:fwer_perturbation}
Fix a block $b\in[B]$ and suppose the within-block law $\mathbb Q^{(b)}_{h_\ell,\tilde g}$ has the product-density form $3!\prod_{i}\tilde g_i(u_i)\mathbf 1_Q(u)$ on the ordered simplex (within-block independence with marginal $\tilde g$ or uniform per coordinate, as in \eqref{eq:block_product_law_density}). Suppose $g,\widehat g$ are densities on $[0,1]$ that are bounded above by $M:=\max(\|g\|_\infty,\|\widehat g\|_\infty)<\infty$. Let $\vec D:Q\to\{0,1\}^3$ be any (measurable) decision rule, and write $\delta:=\|g-\widehat g\|_\infty$. Then, for every $\ell\in\{0,\dots,K-1\}$,
\begin{equation}
\label{eq:fwer_perturbation}
\Bigl|\mathbb Q^{(b)}_{h_\ell,g}\!\bigl(V^{(b)}\!>\!0\bigr)\;-\;\mathbb Q^{(b)}_{h_\ell,\widehat g}\!\bigl(V^{(b)}\!>\!0\bigr)\Bigr|
\;\le\;L_3\,\delta,\qquad L_3\;:=\;2M.
\end{equation}
\end{lemma}

\begin{proof}
Fix $\ell\in\{0,\dots,K-1\}$ and let $\eta:=\eta^{(b)}(\ell)\in\{0,1\}^3$ be the within-block configuration. Write $l:=\sum_{i=1}^{3}\eta_i$ for the number of within-block alternatives. The within-block strong-FWER functional $\mathrm{FWER}^{\eta}(\vec D;\tilde g):=\mathbb Q^{(b)}_{h_\ell,\tilde g}(V^{(b)}>0)$ admits the linear-in-$\vec D$ representation of \eqref{eq:objconst}, of the form
\begin{equation}
\label{eq:block_fwer_integral}
\mathrm{FWER}^{\eta}(\vec D;\tilde g)
\;=\;\int_{Q}\sum_{i=1}^{3}c_i\,D_i(u)\,p_{\eta,i}(\tilde g;u)\,du,
\end{equation}
where each $p_{\eta,i}(\tilde g;u)$ is a finite sum (or single instance) of products of exactly $l$ density evaluations $\tilde g(u_j)$ at distinct $j\in[3]$ (degree $l$ in $\tilde g$ with absorbed combinatorial weights), and the constants $c_i\ge 0$ depend only on $\eta$. Concretely (see Lemma~\ref{lemma:optimization_problem_k_3}): for $l=0$, all $p_{\eta,i}\equiv 1$ and $\sum_i c_i\int_Q D_i=6\int_Q D_1$; for $l=1$, $\sum_i c_i\int_Q D_i\,p_{\eta,i}(g)=2\int_Q[D_1(g(u_2)+g(u_3))+D_2 g(u_1)]\,du$; for $l=2$, $\sum_i c_i\int_Q D_i\,p_{\eta,i}(g)=2\int_Q[D_1 g(u_2)g(u_3)+D_2 g(u_1)g(u_3)+D_3 g(u_1)g(u_2)]\,du$.

Subtract \eqref{eq:block_fwer_integral} for $\tilde g\in\{g,\widehat g\}$, take absolute values, and bound $0\le D_i(u)\le 1$:
\begin{equation}
\label{eq:fwer_pert_step1}
\Bigl|\mathrm{FWER}^{\eta}(\vec D;g)-\mathrm{FWER}^{\eta}(\vec D;\widehat g)\Bigr|
\;\le\;\int_{Q}\sum_{i=1}^{3}c_i\,\bigl|p_{\eta,i}(g;u)-p_{\eta,i}(\widehat g;u)\bigr|\,du.
\end{equation}
A telescoping bound on a product of $l$ factors, each bounded by $M$ and individually swapped from $g$ to $\widehat g$, gives
\begin{equation}
\label{eq:fwer_pert_step2}
\bigl|p_{\eta,i}(g;u)-p_{\eta,i}(\widehat g;u)\bigr|\;\le\;l\,M^{\,l-1}\,\delta,
\end{equation}
uniformly in $u\in Q$ (with the convention $l\,M^{l-1}=0$ when $l=0$). Substituting \eqref{eq:fwer_pert_step2} into \eqref{eq:fwer_pert_step1} and using $\mathrm{vol}(Q)=1/3!$,
\begin{equation*}
\Bigl|\mathrm{FWER}^{\eta}(\vec D;g)-\mathrm{FWER}^{\eta}(\vec D;\widehat g)\Bigr|
\;\le\;\Bigl(\sum_{i=1}^{3}c_i\Bigr)\,l\,M^{\,l-1}\,\delta\,\big/\,3!.
\end{equation*}
Tabulating against the explicit $c_i$: $l=0$ gives $0$; $l=1$ gives $\sum_i c_i=2(2)+2(1)=6$ and $l\,M^{l-1}=1$, so the bound is $6\cdot 1\cdot\delta/6=\delta$; $l=2$ gives $\sum_i c_i=2\cdot 3=6$ and $l\,M^{l-1}=2M$, so the bound is $6\cdot 2M\cdot\delta/6=2M\delta$. Taking the maximum over $l\in\{0,1,2\}$ and noting $2M\ge 2\ge 1$ for any density $g$ with $\int_0^1 g=1$ (since $M\ge\int_0^1 g=1$) yields \eqref{eq:fwer_perturbation} with $L_3=2M$.
\end{proof}

\begin{lemma}[Power perturbation bound in $g$]
\label{lem:power_perturbation}
Under the setup of Lemma~\ref{lem:fwer_perturbation}, for every measurable $\vec D:Q\to\{0,1\}^3$,
\begin{equation}
\label{eq:power_perturbation}
\Bigl|\Pi_3^{g}(\vec D)-\Pi_3^{\widehat g}(\vec D)\Bigr|
\;\le\;L_3^{\mathrm{pow}}\,\|g-\widehat g\|_{\infty},\qquad L_3^{\mathrm{pow}}\;:=\;3M^{2}.
\end{equation}
\end{lemma}

\begin{proof}
By Lemma~\ref{lemma:optimization_problem_k_3}, $\Pi_3^{\tilde g}(\vec D)=2\int_Q(D_1+D_2+D_3)\,\tilde g(u_1)\tilde g(u_2)\tilde g(u_3)\,du$. Subtract the two evaluations and take absolute values:
\begin{equation*}
\bigl|\Pi_3^{g}(\vec D)-\Pi_3^{\widehat g}(\vec D)\bigr|
\;\le\;2\int_{Q}(D_1+D_2+D_3)\,\bigl|g(u_1)g(u_2)g(u_3)-\widehat g(u_1)\widehat g(u_2)\widehat g(u_3)\bigr|\,du.
\end{equation*}
The same telescoping bound used in \eqref{eq:fwer_pert_step2} (with $l=3$), with $\delta:=\|g-\widehat g\|_\infty$, gives
\begin{equation*}
\bigl|g(u_1)g(u_2)g(u_3)-\widehat g(u_1)\widehat g(u_2)\widehat g(u_3)\bigr|\;\le\;3M^{2}\,\delta
\end{equation*}
uniformly in $u\in Q$. Bounding $D_1+D_2+D_3\le 3$ and using $\mathrm{vol}(Q)=1/3!$ then gives $|\Pi_3^g(\vec D)-\Pi_3^{\widehat g}(\vec D)|\le 2\cdot 3\cdot 3M^{2}\,\delta\cdot(1/6)=3M^{2}\,\delta$.
\end{proof}

\begin{remark}[On the constant $M$ for spiky alternatives]
\label{rem:M_blowup}
The constants $L_3=2M$ and $L_3^{\mathrm{pow}}=3M^{2}$ in Lemmas~\ref{lem:fwer_perturbation}--\ref{lem:power_perturbation} scale with $M=\max(\|g\|_\infty,\|\widehat g\|_\infty)$, which is unbounded for strong-signal alternatives whose density mass concentrates near the origin (e.g., sharply truncated normal with large $|\theta|$ in our running model~\eqref{eq:pvalue_model}). The $L^{\infty}$ formulation is therefore qualitatively tight but practically loose in precisely the regime (highly informative $g$) where blockwise multiplicity gains are largest. A sharper rate likely requires a weighted norm $\|(g-\widehat g)/w\|_{q}$ for a weight $w$ matched to the small-$p$-value tail (so that mass near $0$ is downweighted), and an attendant change of telescoping argument; we leave this refinement to future work.
\end{remark}

\begin{lemma}[Value-function perturbation in $g$]
\label{lem:value_perturbation}
Let $g,\widehat g$ be densities on $[0,1]$ with $M:=\max(\|g\|_\infty,\|\widehat g\|_\infty)<\infty$ and $\delta:=\|g-\widehat g\|_{\infty}$. For every $\alpha\ge L_3\delta$,
\begin{equation}
\label{eq:value_perturbation}
\pi_{3}^{\widehat g}(\alpha)\;\ge\;\pi_{3}^{g}\!\bigl(\alpha-L_3\delta\bigr)\;-\;L_{3}^{\mathrm{pow}}\,\delta,
\qquad
\pi_{3}^{g}(\alpha)\;\ge\;\pi_{3}^{\widehat g}\!\bigl(\alpha-L_3\delta\bigr)\;-\;L_{3}^{\mathrm{pow}}\,\delta,
\end{equation}
with $L_3=2M$ and $L_{3}^{\mathrm{pow}}=3M^{2}$.
\end{lemma}

\begin{proof}
Fix $\alpha\ge L_3\delta$ and let $\vec D^{\star}$ achieve the supremum (or its $\varepsilon$-approximation) in the definition of $\pi_{3}^{g}(\alpha-L_3\delta)$, so $\mathrm{FWER}^{\eta}(\vec D^{\star};g)\le\alpha-L_3\delta$ for all $\eta\in\{0,1\}^{3}$ and $\Pi_{3}^{g}(\vec D^{\star})=\pi_{3}^{g}(\alpha-L_3\delta)$. By Lemma~\ref{lem:fwer_perturbation},
\[
\mathrm{FWER}^{\eta}(\vec D^{\star};\widehat g)
\;\le\;\mathrm{FWER}^{\eta}(\vec D^{\star};g)\;+\;L_3\delta
\;\le\;\alpha,
\]
so $\vec D^{\star}$ is feasible for the $\widehat g$-problem at level $\alpha$. By Lemma~\ref{lem:power_perturbation},
\[
\pi_{3}^{\widehat g}(\alpha)\;\ge\;\Pi_{3}^{\widehat g}(\vec D^{\star})
\;\ge\;\Pi_{3}^{g}(\vec D^{\star})\;-\;L_{3}^{\mathrm{pow}}\delta
\;=\;\pi_{3}^{g}(\alpha-L_3\delta)\;-\;L_{3}^{\mathrm{pow}}\delta.
\]
The reverse direction is identical with the roles of $g$ and $\widehat g$ swapped.
\end{proof}

\section{Proofs of the main-text results}
\label{app:main_proofs}

This appendix collects the proofs of the formal results stated in Section~\ref{sec:main}; each subsection is referenced from the corresponding statement in the main text. Proofs are arranged in the order in which the statements appear, and use the per-block perturbation lemmas of Appendix~\ref{app:plugin_lemmas}.

\subsection{Proof of Theorem~\ref{thm:blockwise_strong_fwer_clean}}
\label{app:proof_blockwise_strong_fwer_clean}
\begin{proof}
Fix $\ell\in\{0,1,\ldots,K-1\}$ and write $\eta^{(b)}:=\eta^{(b)}(\ell)$. Let $\mathcal R^{(b)}(X^{(b)})\subseteq\mathcal B_b$ denote the within-block rejection set produced by Algorithm~\ref{alg:blockwise_k3}, and define the global and blockwise false-rejection counts under $h_\ell$ by
\[
V:=\sum_{k\in[K]}\bigl(1-(h_\ell)_k\bigr)\,\mathbf 1\{k\in\mathcal R(\vec X)\},
\qquad
V^{(b)}:=\sum_{k\in\mathcal B_b}\bigl(1-(h_\ell)_k\bigr)\,\mathbf 1\{k\in\mathcal R^{(b)}(X^{(b)})\}.
\]
Since $\{\mathcal B_b\}_{b=1}^B$ partitions $[K]$ and $\mathcal R(\vec X)=\bigcup_{b=1}^B\mathcal R^{(b)}(X^{(b)})$, we have the exact decomposition $V=\sum_{b=1}^B V^{(b)}$. Boole's inequality therefore gives
\begin{equation}
\label{eq:boole_clean}
\mathbb P_{h_\ell}(V>0)\;\le\;\sum_{b=1}^B \mathbb P_{h_\ell}\!\bigl(V^{(b)}>0\bigr).
\end{equation}

Fix a block $b\in[B]$. By Assumption~\ref{as:block_measurability} and the definition of $V^{(b)}$ in \eqref{eq:Vb_def_clean}, the event $\{V^{(b)}>0\}$ is $\sigma(X^{(b)})$-measurable: there exists a measurable set $A_b\subseteq\mathcal X^{(b)}$ such that
\begin{equation}
\label{eq:event_depends_only_on_block}
\{V^{(b)}>0\}\;=\;\{X^{(b)}\in A_b\}\;\in\;\sigma(X^{(b)}).
\end{equation}
By the definition of the block marginal \eqref{eq:block_marginal_law},
\begin{equation}
\label{eq:marginal_eq_clean}
\mathbb P_{h_\ell}\!\bigl(V^{(b)}>0\bigr)
\;=\;\mathbb P_{h_\ell}\!\bigl(X^{(b)}\in A_b\bigr)
\;=\;\mathbb Q^{(b)}_{h_\ell}\!\bigl(X^{(b)}\in A_b\bigr)
\;=\;\mathbb Q^{(b)}_{h_\ell}\!\Bigl(V^{(b)}(X^{(b)};\eta^{(b)})>0\Bigr).
\end{equation}
If $\eta^{(b)}=(1,1,1)$, then by \eqref{eq:Vb_def_clean} we have $V^{(b)}(X^{(b)};\eta^{(b)})\equiv 0$, so the right-hand side of \eqref{eq:marginal_eq_clean} equals $0\le\alpha_{\mathrm{blk}}^{(b)}$. If $\sum_{i=1}^3\eta^{(b)}_i\le 2$, Assumption~\ref{as:local_validity_marginal} applies and yields the same upper bound. In either case,
\begin{equation}
\label{eq:block_bound_clean_detailed}
\mathbb P_{h_\ell}\!\bigl(V^{(b)}>0\bigr)\;\le\;\alpha_{\mathrm{blk}}^{(b)}.
\end{equation}
Substituting \eqref{eq:block_bound_clean_detailed} into \eqref{eq:boole_clean},
\[
\mathrm{FWER}_\ell(\vec D^{\mathrm{blk}})
\;=\;\mathbb P_{h_\ell}(V>0)
\;\le\;\sum_{b=1}^B \alpha_{\mathrm{blk}}^{(b)}
\;\le\;\alpha.
\]
Since $\ell$ was arbitrary, the conclusion holds for all $\ell\in\{0,1,\ldots,K-1\}$.
\end{proof}

\subsection{Proof of Theorem~\ref{thm:plugin_fwer}}
\label{app:proof_plugin_fwer}
\begin{proof}
Fix $\ell$ and condition on $\mathcal G$. By (i), conditionally on $\mathcal G$ the testing-fold data $(X^{(b)})_{b\in\mathcal T}$ has the same distribution as under the original product law, so for each $b\in\mathcal T$
\[
\mathbb E\!\left[\mathbf 1\{V^{(b)}>0\}\mid\mathcal G\right]
\;=\;\mathbb Q^{(b)}_{h_\ell,g}\!\bigl(V^{(b)}>0\bigr).
\]
By (iii) and the solver guarantee~\eqref{eq:solver_guarantee_against_ghat},
$\mathbb Q^{(b)}_{h_\ell,\widehat g}(V^{(b)}>0)\le\alpha_{\mathrm{blk}}^{(b)}$
conditionally on $\mathcal G$, and Lemma~\ref{lem:fwer_perturbation} gives
\[
\mathbb Q^{(b)}_{h_\ell,g}\!\bigl(V^{(b)}>0\bigr)
\;\le\;\alpha_{\mathrm{blk}}^{(b)}\;+\;L_3\,\|g-\widehat g\|_\infty
\]
conditionally on $\mathcal G$. Summing over $b\in\mathcal T$ and applying Boole's inequality (exactly as in the proof of Theorem~\ref{thm:blockwise_strong_fwer_clean}) yields
\[
\mathbb P_{h_\ell}\!\bigl(V^{\mathcal T}>0\mid\mathcal G\bigr)
\;\le\;\sum_{b\in\mathcal T}\alpha_{\mathrm{blk}}^{(b)}\;+\;L_3\,B_{\mathcal T}\,\|g-\widehat g\|_\infty
\;\le\;\alpha\;+\;L_3\,B_{\mathcal T}\,\|g-\widehat g\|_\infty,
\]
where the last step uses (ii). Taking expectations over $\mathcal G$ gives~\eqref{eq:plugin_fwer_bound}.
\end{proof}

\subsection{Proof of Corollary~\ref{cor:plugin_sample_splitting}}
\label{app:proof_plugin_sample_splitting}
\begin{proof}
The independence of $\mathcal G=\sigma(\{X^{(b)}\}_{b\in\mathcal E})$ and $\{X^{(b)}\}_{b\in\mathcal T}$ gives condition~(i) of Theorem~\ref{thm:plugin_fwer} directly. The \v{S}id\'ak remark is immediate: the Bonferroni sum $\sum_{b\in\mathcal T}\alpha_{\mathrm{blk}}^{(b)}\le\alpha$ used in the last step of the proof is implied by the \v{S}id\'ak budget $\prod_{b\in\mathcal T}(1-\alpha_{\mathrm{blk}}^{(b)})\ge 1-\alpha$ via Bernoulli's inequality.
\end{proof}

\subsection{Proof of Theorem~\ref{thm:plugin_power}}
\label{app:proof_plugin_power}
\begin{proof}
Fix $b\in\mathcal T$ and condition on $\mathcal G$. By condition~(i) of Theorem~\ref{thm:plugin_fwer}, $\widehat g$ is $\mathcal G$-measurable and the testing-fold block $X^{(b)}$ has its true $g$-law conditionally on $\mathcal G$; hence
\[
\mathrm{Power}_{g}^{(b)}\!\bigl(\widehat{\vec D}^{(b)}\bigr)
\;=\;\tfrac{1}{3}\,\mathbb E\!\left[\textstyle\sum_{i=1}^{3}D_{i}^{\hat{\vec\mu}^{(b)}}\!\bigl(U^{(b)}\bigr)\,\Big|\,\mathcal G\right]
\;=\;\Pi_{3}^{g}\!\bigl(\widehat{\vec D}^{(b)}\bigr),
\]
where $\Pi_{3}^{g}$ is the complete-alternative power functional of Lemma~\ref{lemma:optimization_problem_k_3}. By construction (item~(iii) of Theorem~\ref{thm:plugin_fwer}), the rule $\widehat{\vec D}^{(b)}$ is the optimal block rule of Algorithm~\ref{alg:compute_optimal_mu_K3_main} computed with $\widehat g$ at level $\alpha_{\mathrm{blk}}^{(b)}$, so $\Pi_{3}^{\widehat g}(\widehat{\vec D}^{(b)})=\pi_{3}^{\widehat g}(\alpha_{\mathrm{blk}}^{(b)})$ conditionally on $\mathcal G$. Lemma~\ref{lem:power_perturbation} then gives
\[
\Pi_{3}^{g}\!\bigl(\widehat{\vec D}^{(b)}\bigr)
\;\ge\;\Pi_{3}^{\widehat g}\!\bigl(\widehat{\vec D}^{(b)}\bigr)\;-\;L_{3}^{\mathrm{pow}}\,\|g-\widehat g\|_{\infty}
\;=\;\pi_{3}^{\widehat g}\!\bigl(\alpha_{\mathrm{blk}}^{(b)}\bigr)\;-\;L_{3}^{\mathrm{pow}}\,\|g-\widehat g\|_{\infty}
\]
conditionally on $\mathcal G$. Taking $\mathcal G$-expectation, summing over $b\in\mathcal T$, and dividing by $B_{\mathcal T}$ gives \eqref{eq:plugin_power_bound}. Multiplying \eqref{eq:plugin_power_bound} by $3B_{\mathcal T}$ and using $\mathbb E[N_{\mathrm{rej}}^{(b)}]=3\,\mathbb E[\mathrm{Power}_{g}^{(b)}]$ yields \eqref{eq:plugin_power_bound_sum}.
\end{proof}

\subsection{Proof of Corollary~\ref{cor:plugin_power_oracle}}
\label{app:proof_plugin_power_oracle}
\begin{proof}
Condition on $\mathcal G$. Lemma~\ref{lem:value_perturbation} (applied with $\alpha=\alpha_{\mathrm{blk}}^{(b)}\ge L_3\delta$, $\delta=\|g-\widehat g\|_{\infty}$) gives, for every $b\in\mathcal T$,
\[
\pi_{3}^{\widehat g}\!\bigl(\alpha_{\mathrm{blk}}^{(b)}\bigr)
\;\ge\;\pi_{3}^{g}\!\bigl(\alpha_{\mathrm{blk}}^{(b)}-L_{3}\,\|g-\widehat g\|_{\infty}\bigr)\;-\;L_{3}^{\mathrm{pow}}\,\|g-\widehat g\|_{\infty}.
\]
Substitute into the conditional bound from the proof of Theorem~\ref{thm:plugin_power}:
\[
\Pi_{3}^{g}\!\bigl(\widehat{\vec D}^{(b)}\bigr)\;\ge\;\pi_{3}^{\widehat g}\!\bigl(\alpha_{\mathrm{blk}}^{(b)}\bigr)\;-\;L_{3}^{\mathrm{pow}}\,\|g-\widehat g\|_{\infty}
\;\ge\;\pi_{3}^{g}\!\bigl(\alpha_{\mathrm{blk}}^{(b)}-L_{3}\,\|g-\widehat g\|_{\infty}\bigr)\;-\;2L_{3}^{\mathrm{pow}}\,\|g-\widehat g\|_{\infty}.
\]
Take $\mathcal G$-expectation, sum over $b\in\mathcal T$, and divide by $B_{\mathcal T}$.
\end{proof}

\subsection{Proof of Corollary~\ref{cor:plugin_rates}}
\label{app:proof_plugin_rates}
\begin{proof}
Applying~\eqref{eq:plugin_rate_hyp} to the FWER bound~\eqref{eq:plugin_fwer_bound} of Theorem~\ref{thm:plugin_fwer} yields~\eqref{eq:plugin_rates_fwer}; applying it to the power bound~\eqref{eq:plugin_power_bound} of Theorem~\ref{thm:plugin_power} yields~\eqref{eq:plugin_rates_pow}, with the $B_{\mathcal T}$-free constant $L_3^{\mathrm{pow}}/3=M^2$ absorbed into $O(\cdot)$. The Grenander rate in (i) is Corollary~7.6 of~\citet{groeneboom2014nonparametric}; the kernel-density H\"older rate in (ii) is Theorem~1.7 of~\citet{tsybakov2009book}. Asymptotic $\alpha$-exactness and power-matching under $B_{\mathcal T} r_{n_{\mathcal E}}\to 0$ follow directly.
\end{proof}

\subsection{Proof of Lemma~\ref{lem:pi3_concavity}}
\label{app:proof_pi3_concavity}
\begin{proof}
Non-decreasingness is immediate: enlarging $\alpha$ enlarges the feasible set in \eqref{eq:pi3_value_def}, so the supremum can only grow.

For concavity, fix $\alpha_1,\alpha_2\in[0,1]$, $t\in[0,1]$, and $\varepsilon>0$. Pick measurable $\vec D^{(1)},\vec D^{(2)}:Q\to\{0,1\}^3$ with $\mathrm{FWER}_l^{\mathbb Q}(\vec D^{(i)})\le\alpha_i$ for $l\in\{0,1,2\}$, $i\in\{1,2\}$, and $\Pi_3^{\mathbb Q}(\vec D^{(i)})\ge\pi_3^{\mathbb Q}(\alpha_i)-\varepsilon$. Introduce an auxiliary $U\sim\mathrm{Unif}(0,1)$ independent of the within-block $p$-value vector under any configuration, and form the randomized rule
\[
\vec R(P,U)\;:=\;\vec D^{(1)}(P)\,\mathbf 1\{U\le t\}+\vec D^{(2)}(P)\,\mathbf 1\{U>t\}.
\]
Each of $\mathrm{FWER}_l^{\mathbb Q}$ ($l=0,1,2$) and $\Pi_3^{\mathbb Q}$ is a linear functional of $\vec D$ in the integrand sense: writing $\vec D=(D_1,D_2,D_3)$ and using the joint density $q_l$ of $\mathbb Q_l$ together with the symmetric expressions in Lemma~\ref{lemma:optimization_problem_k_3}, each functional has the form $\sum_{j=1}^3\int_Q D_j(u)\,\varphi^{(l,j)}(u)\,du$ for a Lebesgue-integrable kernel $\varphi^{(l,j)}\ge 0$ determined by $q_l$. Linearity in $\vec D$ then yields
\[
\Pi_3^{\mathbb Q}(\vec R)\;=\;t\,\Pi_3^{\mathbb Q}(\vec D^{(1)})+(1-t)\,\Pi_3^{\mathbb Q}(\vec D^{(2)}),\qquad
\mathrm{FWER}_l^{\mathbb Q}(\vec R)\;\le\;t\alpha_1+(1-t)\alpha_2\quad(l=0,1,2),
\]
so $\vec R$ is a randomized rule feasible at level $t\alpha_1+(1-t)\alpha_2$.

It remains to derandomize $\vec R$ into a deterministic measurable rule with the same $(\mathrm{FWER}_0,\mathrm{FWER}_1,\mathrm{FWER}_2,\Pi_3)$ profile. Consider the $\mathbb R^{4}$-valued vector measure $\boldsymbol\mu$ on $Q\times\{1,2,3\}$ whose components are absolutely continuous with respect to Lebesgue $\otimes$ counting and have densities $\varphi^{(l,j)}$ ($l=0,1,2,3$, $j=1,2,3$). Because each $\mathbb Q_l$ is Lebesgue-absolutely continuous, the Lebesgue part of the product measure has no atoms in the $u$-coordinate, so each singleton $\{(u_0,j_0)\}$ satisfies $\boldsymbol\mu(\{(u_0,j_0)\})=0$ and every component of $\boldsymbol\mu$ is non-atomic (despite the counting measure on $\{1,2,3\}$ being atomic in isolation). By the Lyapunov--Halmos convexity theorem (cf.\ Dvoretzky--Wald--Wolfowitz) for non-atomic vector measures, the range
\[
\bigl\{\bigl(\mathrm{FWER}_0^{\mathbb Q}(\vec D),\mathrm{FWER}_1^{\mathbb Q}(\vec D),\mathrm{FWER}_2^{\mathbb Q}(\vec D),\Pi_3^{\mathbb Q}(\vec D)\bigr):\vec D:Q\to\{0,1\}^3\text{ measurable}\bigr\}
\]
is a convex subset of $[0,1]^4$. The convex combination realized by $\vec R$ is therefore matched by a deterministic measurable $\vec D^{*}$, which is feasible for \eqref{eq:pi3_value_def} at level $t\alpha_1+(1-t)\alpha_2$ with $\Pi_3^{\mathbb Q}(\vec D^{*})=\Pi_3^{\mathbb Q}(\vec R)\ge t\,\pi_3^{\mathbb Q}(\alpha_1)+(1-t)\,\pi_3^{\mathbb Q}(\alpha_2)-\varepsilon$. Taking $\varepsilon\downarrow 0$ yields concavity.

Specializing the family $\{\mathbb Q_l\}_{l=0,1,2,3}$ to the per-block configuration laws $\{\mathbb Q^{(b)}_{\!\vec h_l}\}_{l=0,1,2,3}$, which are absolutely continuous and exchangeable within the alternative indices under \eqref{eq:pvalue_model} (within-block independence with common $g$), gives the per-block conclusion.
\end{proof}

\subsection{Proof of Theorem~\ref{thm:global_opt_sep_within_block}}
\label{app:proof_global_opt_sep_within_block}
\begin{proof}
Fix $D^{\mathrm{sep}}\in\mathfrak D_{\mathrm{sep}}$ with allocation $(\alpha_1,\ldots,\alpha_B)$ and within-block rules $\{\vec D^{(b)}\}_{b=1}^B$. Under $\vec h_K$, every within-block configuration is $(1,1,1)$. Using $K=3B$ and $\Pi_K(D)=\frac1K\mathbb E_{\vec h_K}[\sum_{k=1}^K D_k(\vec X)]$,
\begin{equation}
\label{eq:power_decomp_pf}
\Pi_K\!\bigl(D^{\mathrm{sep}}\bigr)
\;=\;
\frac{1}{3B}\sum_{b=1}^{B}\mathbb E_{\vec h_K}\!\left[\sum_{i=1}^{3} D^{(b)}_i(P^{(b)})\right]
\;=\;
\frac{1}{B}\sum_{b=1}^{B}\Bigl\{\tfrac{1}{3}\,\mathbb E_{\vec h_K}\!\bigl[\textstyle\sum_{i=1}^{3} D^{(b)}_i(P^{(b)})\bigr]\Bigr\}.
\end{equation}
Since $\vec D^{(b)}$ depends only on $P^{(b)}$, by the definition \eqref{eq:per_block_alt_law} of the per-block law $\mathbb Q^{(b)}_{\!\vec h_K}=\mathcal L_{\vec h_K}(P^{(b)})$ and writing $P\sim\mathbb Q^{(b)}_{\!\vec h_K}$,
\begin{equation}
\label{eq:block_power_is_pi3_pf}
\tfrac{1}{3}\,\mathbb E_{\vec h_K}\!\bigl[\textstyle\sum_{i=1}^{3} D^{(b)}_i(P^{(b)})\bigr]
\;=\;
\mathbb E_{\mathbb Q^{(b)}_{\!\vec h_K}}\!\bigl[\textstyle\tfrac13\sum_{i=1}^3 D^{(b)}_i(P)\bigr]
\;=:\;\Pi_3^{\mathbb Q^{(b)}_{\!\vec h_K}}\!\bigl(\vec D^{(b)}\bigr).
\end{equation}
Substituting \eqref{eq:block_power_is_pi3_pf} into \eqref{eq:power_decomp_pf} yields the equality in \eqref{eq:thm_within_block_bound}.

By Definition~\ref{def:block_separable}, each $\vec D^{(b)}$ controls strong FWER at level $\alpha_b$ under $\mathbb Q^{(b)}_{\!\vec h_K}$ and is therefore feasible for the $K=3$ problem defining $\pi_3^{(b)}$ at level $\alpha_b$. By the definition of $\pi_3^{(b)}$ in \eqref{eq:pi3_per_block_def},
\begin{equation}
\label{eq:block_le_value_pf}
\Pi_3^{\mathbb Q^{(b)}_{\!\vec h_K}}\!\bigl(\vec D^{(b)}\bigr)\;\le\;\pi_3^{(b)}(\alpha_b),\qquad b\in[B],
\end{equation}
with equality when $\vec D^{(b)}$ is a maximizer in \eqref{eq:pi3_per_block_def}; in particular, the $K=3$ optimizer of~\citet{dubey25} at level $\alpha_b$ applied to $\mathbb Q^{(b)}_{\!\vec h_K}$. Averaging \eqref{eq:block_le_value_pf} over $b$ gives the inequality in \eqref{eq:thm_within_block_bound}; equality holds globally when it holds in each block. The reduction under Assumption~\ref{assump:exchangeable_blocks} is by definition of that assumption.
\end{proof}

\subsection{Proof of Theorem~\ref{thm:global_opt_sep_allocation}}
\label{app:proof_global_opt_sep_allocation}
\begin{proof}
The objective \eqref{eq:hetero_alloc_max} is a sum of concave functions in separate variables on the compact convex polytope $\{\vec\alpha\in[0,\alpha]^B:\sum_b\alpha_b\le\alpha\}$; continuity plus compactness yields a maximizer.

KKT for the constrained maximization (using subdifferentials of the concave $\pi_3^{(b)}$, with right derivatives sufficing because the objective is concave on a polyhedral set; cf.~\citealp[Thm.~28.3]{rockafellar1970convex}) gives multipliers $\mu^{*}\ge 0$ on the budget constraint and $\lambda_b^{*}\ge 0$ on the $\alpha_b\ge 0$ constraints, with stationarity
\[
\tfrac{1}{B}\,\partial_+\pi_3^{(b)}(\alpha_b^{*})\;=\;\tfrac{\mu^{*}}{B}\,-\,\tfrac{\lambda_b^{*}}{B},\qquad \lambda_b^{*}\alpha_b^{*}=0,
\]
which simplifies to $\partial_+\pi_3^{(b)}(\alpha_b^{*})\le\mu^{*}$ with equality on the active set $\{b:\alpha_b^{*}>0\}$, i.e.~\eqref{eq:KKT_condition}. Monotonicity ($\partial_+\pi_3^{(b)}\ge 0$ on $[0,\alpha]$ since enlarging $\alpha$ enlarges the feasible set in \eqref{eq:pi3_value_def}) implies that a maximizer with binding budget always exists: any allocation with $\sum_b\alpha_b^{*}<\alpha$ admits a weak improvement by transferring slack to any block whose marginal is positive, and if every $\partial_+\pi_3^{(b)}(\alpha_b^{*})=0$, then $\mu^{*}=0$ and any binding-budget extension is also a maximizer.

Under Assumption~\ref{assump:exchangeable_blocks}, all $\pi_3^{(b)}\equiv\pi_3$ and \eqref{eq:KKT_condition} reads $\partial_+\pi_3(\alpha_b^{*})=\mu^{*}$ for all $b$ in the active set. Since $\partial_+\pi_3$ is non-increasing, the level set $\{\alpha:\partial_+\pi_3(\alpha)=\mu^{*}\}$ is a (possibly degenerate) interval; the symmetric solution $\alpha_b^{*}=\bar\alpha=\sum_b\alpha_b^{*}/B$ lies in this set, and combined with $\sum_b\alpha_b^{*}=\alpha$ gives $\alpha_b^{*}=\alpha/B$. Hence \eqref{eq:hetero_alloc_max} equals $\pi_3(\alpha/B)$, recovering \eqref{eq:uniform_is_optimal}.
\end{proof}

\subsection{Proof of Theorem~\ref{thm:global_opt_sep_allocation_sidak}}
\label{app:proof_global_opt_sep_allocation_sidak}
\begin{proof}
\emph{Feasibility.} Fix $D^{\mathrm{sep}}\in\mathfrak D^{\mathrm{ind}}_{\mathrm{sep}}$ and any $h\in\{0,1\}^K$, and let $E_b:=\{\exists i\in\mathcal B_b:D^{(b)}_i(P^{(b)})=1,\ h_i=0\}$. Because $D^{(b)}$ depends only on $P^{(b)}$ and the $P^{(b)}$ are mutually independent under $\mathcal L_h(P)$, the events $E_1,\ldots,E_B$ are mutually independent, hence
\[
\mathbb P_h\!\Bigl(\textstyle\bigcup_{b=1}^{B}E_b\Bigr)
\;=\;1-\prod_{b=1}^{B}\bigl(1-\mathbb P_{h_b}(E_b)\bigr)
\;\le\;1-\prod_{b=1}^{B}(1-\alpha_b)
\;\le\;\alpha,
\]
using the per-block strong-FWER bound $\mathbb P_{h_b}(E_b)\le\alpha_b$ (Definition~\ref{def:block_separable}) and the \v{S}id\'ak constraint. Specializing to $h=h_\ell$ for $\ell\in\{0,\ldots,K-1\}$ gives feasibility for \eqref{eq:MHTK}.

\emph{Existence.} Reparametrize via $\beta_b:=-\log(1-\alpha_b)\in[0,\infty)$ and $\widetilde\pi_3^{(b)}(\beta):=\pi_3^{(b)}(1-e^{-\beta})$, so \eqref{eq:sidak_alloc_max} becomes
\[
\max_{\beta_b\ge 0,\ \sum_b\beta_b\le-\log(1-\alpha)}\;\tfrac{1}{B}\sum_{b=1}^{B}\widetilde\pi_3^{(b)}(\beta_b).
\]
Each $\widetilde\pi_3^{(b)}$ is concave and non-decreasing in $\beta$ (composition of the concave non-decreasing $\pi_3^{(b)}$ with the concave non-decreasing $\beta\mapsto 1-e^{-\beta}$); continuity on a compact polytope yields a maximizer, which is also a maximizer of \eqref{eq:sidak_alloc_max} by the bijection $\alpha_b\leftrightarrow\beta_b$.

\emph{KKT.} For \eqref{eq:sidak_alloc_max}, introduce $\mu^{*}\ge 0$ on the \v{S}id\'ak constraint written as $\sum_b[-\log(1-\alpha_b)]\le-\log(1-\alpha)$ and $\lambda_b^{*}\ge 0$ on $\alpha_b\ge 0$ (cf.~\citealp[Thm.~28.3]{rockafellar1970convex}). Stationarity gives
\[
\tfrac{1}{B}\,\partial_+\pi_3^{(b)}(\alpha_b^{*})
\;=\;\tfrac{\mu^{*}}{1-\alpha_b^{*}}\,-\,\lambda_b^{*},
\qquad \lambda_b^{*}\,\alpha_b^{*}=0,
\]
which on the active set $\{b:\alpha_b^{*}>0\}$ rearranges to \eqref{eq:sidak_KKT}. Monotonicity $\partial_+\pi_3^{(b)}\ge 0$ combined with the fact that any Bonferroni-feasible allocation satisfies the \v{S}id\'ak constraint strictly for $B>1$ forces $\mu^{*}>0$ and the \v{S}id\'ak constraint binds.

\emph{Homogeneous case.} When $\pi_3^{(b)}\equiv\pi_3$, \eqref{eq:sidak_KKT} reads $(1-\alpha_b^{*})\partial_+\pi_3(\alpha_b^{*})=B\mu^{*}$ on the active set. Under $\|g\|_\infty<\infty$ (Proposition~\ref{prop:canonical_model_verification}(iv)), $\partial_+\pi_3$ is finite on $[0,\alpha]$, and the map $\alpha\mapsto(1-\alpha)\partial_+\pi_3(\alpha)$ is strictly decreasing on any interval where $\partial_+\pi_3>0$ (product of the strictly decreasing $(1-\alpha)$ and the non-increasing $\partial_+\pi_3\ge 0$), so the level equation has at most one active value, forcing $\alpha_b^{*}\equiv\bar\alpha$. The binding constraint $\prod_b(1-\bar\alpha)=1-\alpha$ then yields $\bar\alpha=1-(1-\alpha)^{1/B}$. The comparison $1-(1-\alpha)^{1/B}\ge\alpha/B$ follows from Bernoulli's inequality (strict for $B>1$, $\alpha\in(0,1)$), and monotonicity of $\pi_3$ yields \eqref{eq:sidak_uniform_is_optimal}.
\end{proof}

\subsection{Proof of Proposition~\ref{prop:bisection_allocation}}
\label{app:proof_bisection_allocation}
\begin{proof}
Monotonicity, right-continuity, and limits of $g_b$ are direct consequences of the concavity and monotonicity of $\pi_3^{(b)}$. Bisection convergence is standard. Optimality of $\vec\alpha^{*}$ is the KKT characterization \eqref{eq:KKT_condition} of Theorem~\ref{thm:global_opt_sep_allocation}: by construction $\partial_+\pi_3^{(b)}(\alpha_b^{*})$ is at most $\mu^{*}$, with equality whenever $\alpha_b^{*}>0$.
\end{proof}

\subsection{Proof of Corollary~\ref{cor:blockwise_sep_optimality}}
\label{app:proof_blockwise_sep_optimality}
\begin{proof}
By Theorem~\ref{thm:global_opt_sep_within_block}, every $D^{\mathrm{sep}}\in\mathfrak D_{\mathrm{sep}}$ with allocation $\vec\alpha$ satisfies $\Pi_K(D^{\mathrm{sep}})\le(1/B)\sum_b\pi_3^{(b)}(\alpha_b)$, so
\[
\sup_{\mathfrak D_{\mathrm{sep}}}\Pi_K
\;\le\;\sup_{\substack{\alpha_b\ge 0,\\ \sum_b\alpha_b\le\alpha}}\frac{1}{B}\sum_{b=1}^{B}\pi_3^{(b)}(\alpha_b)
\;\stackrel{\eqref{eq:hetero_alloc_max}}{=}\;\frac{1}{B}\sum_{b=1}^{B}\pi_3^{(b)}(\alpha_b^{*}).
\]
The procedure described attains the bound, so equality holds. The homogeneous reduction is the corresponding case of Theorem~\ref{thm:global_opt_sep_allocation}.
\end{proof}

\section{Experimental setup}
\label{app:experimental_setup}

\subsection{Baselines and model families used in Section~\ref{sec:simulations}}
\label{app:sim_baselines}

The ten strong-FWER baselines used throughout Section~\ref{sec:simulations} span five categories: \emph{classical stepwise} (Bonferroni~\citep{Bonferroni1936}, Holm~\citep{holm}, Hochberg~\citep{hochberg1988sharper}, Hommel~\citep{hommel}); \emph{dependence-aware stepwise} (\v{S}id\'ak single-step and step-down~\citep{sidak1967,romano}, where the step-down coincides with the parametric Romano--Wolf stepdown under independence and continuity); \emph{graphical-gatekeeping} (block-Holm and block-Hochberg, instantiating~\citet{Bretz09} at the uniform-weight instantiation: within each block, Simes-combine the $K$ coordinates; apply Holm, respectively Hochberg, to the $B$ block-level Simes $p$-values at level $\alpha$; and in any rejected block declare all $K$ coordinates significant); \emph{combination closure} (closed-Fisher~\citep{fisher1932,marcus1976closed}); and \emph{resampling} (Westfall--Young max-$T$~\citep{WestfallYoung1993}).

All baselines are implemented with their standard decision rules, at level $\alpha=0.05$. Closed-Fisher is reported only at $K=30$ because it is $\notin\mathfrak D_{\mathrm{sep}}$ (it pools $p$-values globally), which is also why it collapses on sparse and weak-signal truncnorm configurations where the per-block optimality of \textsc{BOOST} is most consequential.

Power and FWER estimates use $3\times 10^4$ Monte-Carlo replicates (E2, E4, E5 power arm, E$\alpha$, Beta, mixed null/alt, misspecification) or $2\times 10^4$ replicates (E1 dependence stress, E5 alloc arm, E10 sparsity, ESP head-to-head, modern closed-testing comparison) or smaller counts for per-$n$ plug-in sweeps (E7: $10^4$; E8: $5\times 10^3$). The corresponding Monte-Carlo standard errors are $\sqrt{p(1-p)/\mathrm{nrep}}$: $\approx 2.9\times 10^{-3}$ (power at $p=0.5$, $3\times 10^4$ reps), $\approx 3.5\times 10^{-3}$ ($2\times 10^4$ reps), $\approx 1.3\times 10^{-3}$ (FWER at $\alpha=0.05$, $3\times 10^4$ reps), $\approx 1.5\times 10^{-3}$ ($2\times 10^4$ reps). We report point estimates in tables; SE bands on figures are omitted because every MC SE is smaller than the plotted marker width at the aspect ratios used.

\paragraph{Baseline-set scoping across experiments.} The \emph{master} comparison figures (Figure~\ref{fig:sim_families} in the body and the all-$K$ panels of Appendix~\ref{app:scaleK_full}) use all ten baselines. Secondary diagnostic panels drop a principled subset:
\begin{itemize}
\item Westfall--Young max-$T$ requires block-specific permutation nulls and is evaluated only in the master comparison to keep compute tractable.
\item Block-Holm and block-Hochberg track Holm and Hochberg to within $5\cdot 10^{-3}$ in every E2 row and are therefore omitted from sparsity/mixed-block panels to avoid curve clutter.
\item \v{S}id\'ak single-step is dominated by \v{S}id\'ak step-down and is kept only in the master comparison.
\end{itemize}
The \emph{level-sensitivity} experiment (Table~\ref{tab:sim_alpha_sensitivity}) restricts further to the Bonferroni/Holm/Hommel trio because the question is whether the \textsc{BOOST}$\,>\,$best-stepwise gap is stable in $\alpha$, not whether dependence-aware competitors close the gap. The \emph{modern closed-testing} comparison (Appendix~\ref{app:sim_modern_baselines}) keeps the Bonferroni/Holm/Hommel trio and pairs it with Meinshausen (2008)~\citep{meinshausen2008hierarchical} and Vovk--Wang-calibrated Hartog e-values~\citep{hartog2025evalues} to isolate the tree-closed-testing and e-value-closed-testing comparisons.

\paragraph{$p$-value-generating families (main-body Section~\ref{subsec:sim1_truncnorm}).}
\begin{itemize}
\item \textbf{Truncated normal (in-class).} $H_{0k}: X_k\sim\mathcal N(0,1)_{\mathrm{trunc}}$ vs.\ $H_{Ak}: X_k\sim\mathcal N(\theta,1)_{\mathrm{trunc}}$ on $[-6,6]$, one-sided $u_k=F_0^{(T)}(X_k)$; $\theta\in\{-0.5,\ldots,-4.0\}$. The induced density $g$ is non-increasing on $[0,1]$, satisfying the canonical class \eqref{eq:pvalue_model}.
\item \textbf{Student-$t$ (mild out-of-class).} $H_{Ak}: X_k\sim t_{\mathrm{df}}$, two-sided, with $\mathrm{df}$ ranging over $\{2,3,4,5,6,8,10,15\}$. $g_{\mathrm{df}}$ decreases from $u=0^+$ to a minimum near $u\approx 0.32$ and rises toward $u=1$, a quantitatively mild non-monotonicity.
\item \textbf{Sparse.} Truncnorm signals placed so that five of the ten blocks are fully at the alternative and the other five are pure null, a mixed-block configuration that reveals how per-block power composes across heterogeneous blocks.
\end{itemize}

\label{app:sim_families}

\subsection{Computational environment}
\label{app:compute_env}

All computational tasks reported in this paper, including the implementation of \textsc{BOOST} (Algorithm~\ref{alg:blockwise_k3}) and the underlying $K=3$ coordinate-descent atom (Algorithm~\ref{alg:compute_optimal_mu_K3_main}), the Monte-Carlo simulations across truncated-normal, $t$-distribution, mixture-normal, sparse, and Beta $p$-value families (Section~\ref{sec:simulations} and Appendices~\ref{app:scaleK_full}--\ref{app:plugin_rate_validation}), and the real-data applications on the BLUEPRINT cis-eQTL and Upworthy bundled-challenger A/B datasets (Section~\ref{sec:applications}), were performed on a single local workstation.

The hardware configuration is a 13-inch 2020 MacBook Pro (Model Identifier: MacBookPro16,2): a 2.3 GHz Quad-Core Intel Core i7 processor (1 processor, 4 cores) with Hyper-Threading enabled, 512 KB L2 cache per core, 8 MB L3 cache, and 32 GB of 3733 MHz LPDDR4X memory. The software environment is macOS Sequoia, Version 15.0.1; numerical and simulation code is implemented in Python~3 (NumPy, SciPy, pandas, statsmodels, matplotlib).

\subsection{Real-data application setup: Monte-Carlo design and baseline scoping}
\label{app:apps_setup}

For each application of Section~\ref{sec:applications}, the results table juxtaposes rejection counts on the real test half at $\alpha=0.05$ with an oracle Monte-Carlo validation at $n_{\mathrm{rep}}=5000$ replicates whose alternative stream is drawn from the fitted Grenander density $\hat g^{\mathrm{real}}$ that \textsc{BOOST} uses internally. The null stream is iid $\mathrm{Unif}[0,1]$, yielding $\mathrm{FWER}_0$; the alternative stream is iid $\hat g^{\mathrm{real}}$, yielding $\Pi_{\mathrm{any}}:=\Pr(\ge\!1\text{ rej})$ and $\Pi_K:=\Pr(\text{all }3\text{ rej})$. Monte-Carlo standard errors are $\approx 0.003$ on $\mathrm{FWER}_0$ and $\approx 0.006$ on $\Pi$.

Baselines follow Section~\ref{sec:simulations} and Appendix~\ref{app:sim_modern_baselines}, with two principled drops in the real-data tables. Westfall--Young max-$T$ is omitted because its null distribution coincides with \v{S}id\'ak-SS on iid continuous $p$-values and the published real-data summaries do not carry the joint permutation distribution it requires. Closed-Fisher ($\notin\mathfrak D_{\mathrm{sep}}$) rejects nothing at $\alpha/B\approx 10^{-5}$ and is likewise dropped.

\paragraph{Data sources.} BLUEPRINT $p$-values are the QTLtools~\citep{ongen2016fast} beta-approximated \texttt{p\_beta} entries in EBI eQTL Catalogue r6~\citep{kerimov2021eqtl}, datasets QTD000021/26/31 (CD14$^{+}$ monocytes, neutrophils, CD4$^{+}$ T cells). Upworthy summaries are the exploratory CSV released at OSF~\texttt{jd64p}.

\section{Scaling, convergence, and timing diagnostics}
\label{app:diagnostics}

\subsection{Full scale-$K$ reference panels for \texorpdfstring{\textsc{BOOST}}{BOOST} vs baselines}
\label{app:scaleK_full}

This reference reports \textsc{BOOST}-vs-baseline performance across the complete block-count grid $K\in\{6,15,30,60\}$ and both power functions: average power $\Pi_K$ in Figure~\ref{fig:sim_scaleK_allK_avg} and any-discovery power $\Pi_{\mathrm{any}}$ in Figure~\ref{fig:sim_scaleK_allK_any}. Configuration matches Section~\ref{subsec:sim1_truncnorm} (truncnorm, tdist, sparse families; $\alpha=0.05$; $\alpha_{\mathrm{blk}}=\alpha/B$; $3\times 10^4$ replicates; truncnorm/sparse cropped to $\theta\in[-2.5,-0.5]$). Methods are \textsc{BOOST} plus the nine non-separable-free stepwise, \v{S}id\'ak, WY-max-$T$, and block-gatekeeping baselines used in Figure~\ref{fig:sim_families}; closed-Fisher is omitted as $\notin\mathfrak D_{\mathrm{sep}}$.

The $K=30$ row reproduces the main-body reference (Figure~\ref{fig:sim_families}). Across every $K\in\{6,15,30,60\}$ and both metrics, the \textsc{BOOST}-vs-baseline ordering observed at $K=30$ is preserved: \textsc{BOOST} leads the stepwise and graphical baselines at moderate signals on truncnorm and sparse, and is indistinguishable from the best stepwise method on tdist at all but the most extreme df.

\begin{figure}[htbp]
\centering
\includegraphics[width=\linewidth]{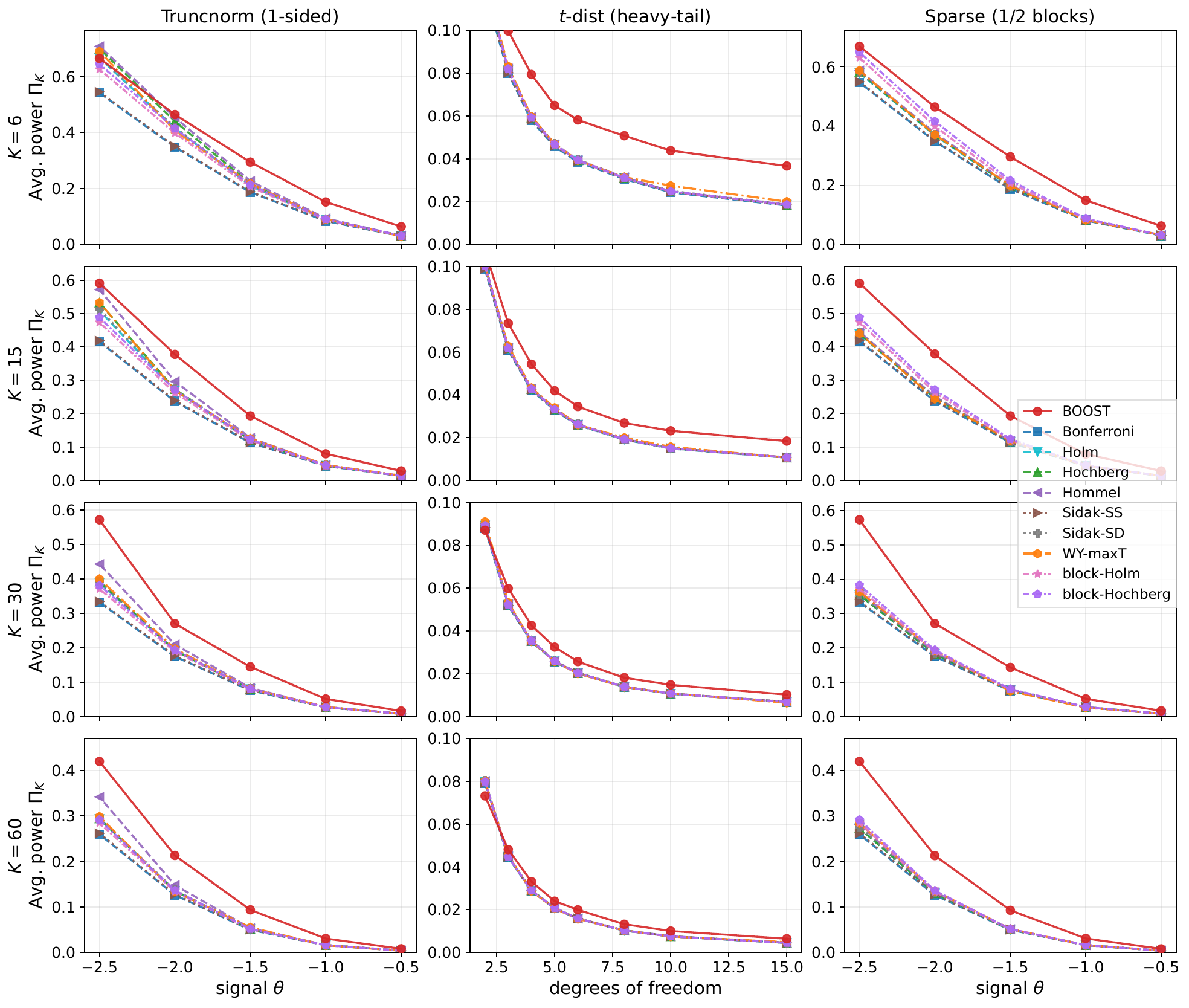}
\caption{Average power $\Pi_K$ across block counts, all $K\in\{6,15,30,60\}$. Rows: $K$ increasing top-to-bottom ($B\in\{2,5,10,20\}$). Columns: truncnorm, tdist, sparse. Methods match the master comparison Figure~\ref{fig:sim_families}.}
\label{fig:sim_scaleK_allK_avg}
\end{figure}

\begin{figure}[htbp]
\centering
\includegraphics[width=\linewidth]{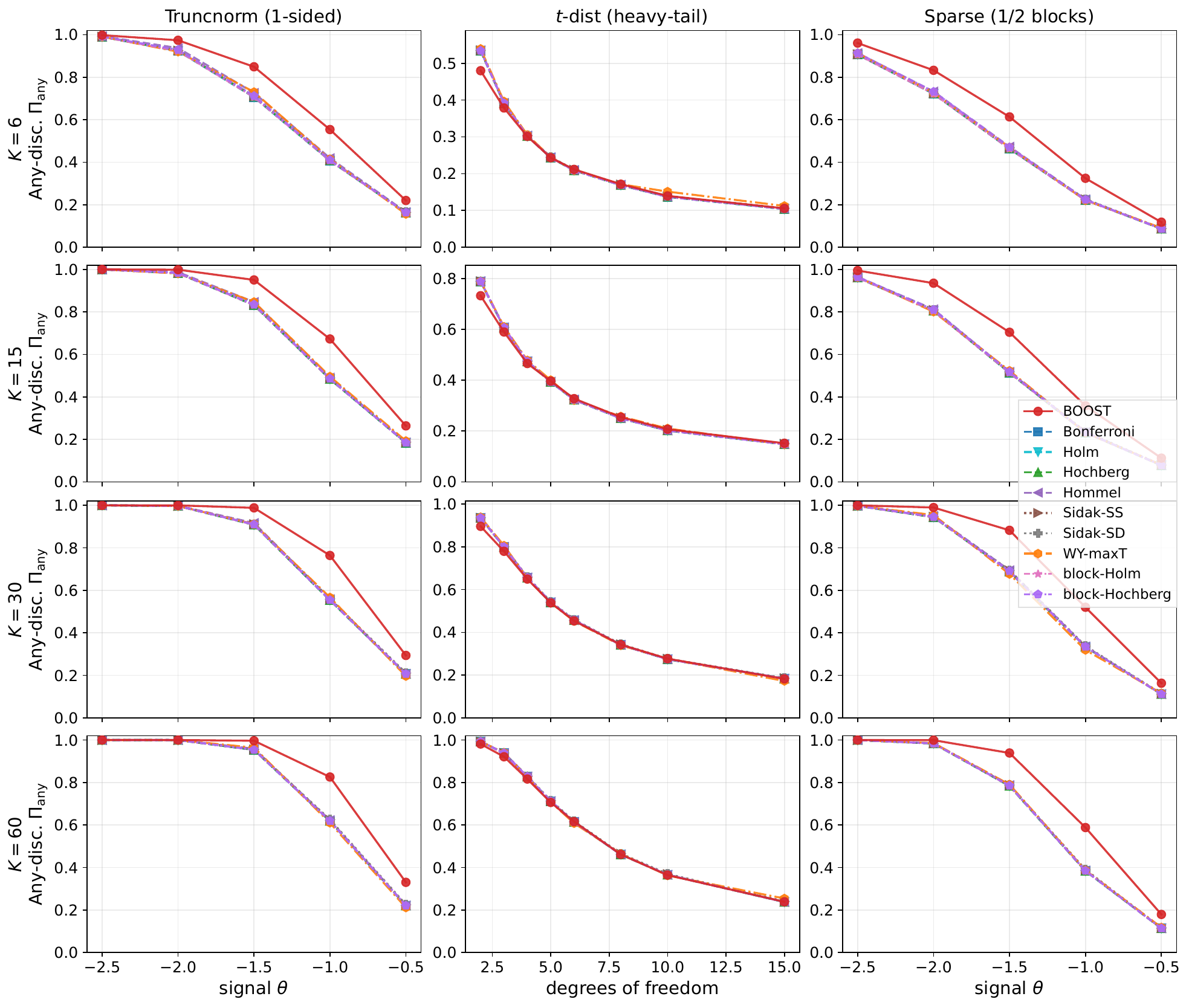}
\caption{Minimal (any-discovery) power $\Pi_{\mathrm{any}}$ across block counts, all $K\in\{6,15,30,60\}$. Rows and columns as in Figure~\ref{fig:sim_scaleK_allK_avg}. Minimal power rises with $K$ because more blocks yield more chances to reject at least once; \textsc{BOOST}'s rank relative to stepwise and graphical baselines is unchanged.}
\label{fig:sim_scaleK_allK_any}
\end{figure}

\subsection{Geometric convergence of the KKT bisection}
\label{app:diag_conv}

Figure~\ref{fig:sim_conv_multiB} plots $|\mu^{(t)}-\mu^{\star}|$ on a logarithmic $y$-axis over the outer iteration index $t$ for $B\in\{2,5,10,20,50\}$. The diagnostic uses the same synthetic concave $\pi_3^{(b)}$ surrogates as experiment E6 (smooth closed-form derivatives; no Monte-Carlo noise) so that the behavior isolates the outer allocation bisection of Proposition~\ref{prop:bisection_allocation}. The approximately linear decrease for every $B$ confirms geometric contraction at the bisection rate $1/2$, independent of $B$; the only $B$-dependence is the starting bracket (absorbed in the $y$-intercept) and is $O(\log B)$ through $\mu_{\max}$. This matches the $O(B\log(1/\varepsilon))$ bound of Proposition~\ref{prop:bisection_allocation}.

\begin{figure}[htbp]
\centering
\includegraphics[width=0.65\linewidth]{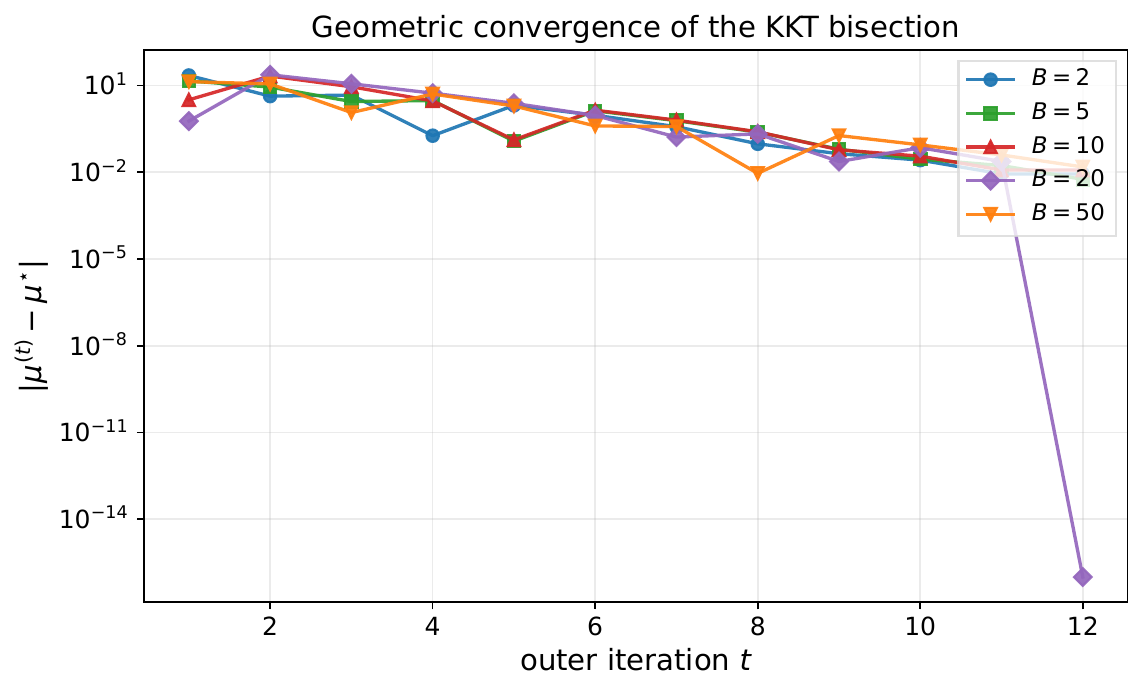}
\caption{Convergence of the outer KKT bisection (Proposition~\ref{prop:bisection_allocation}): $|\mu^{(t)}-\mu^{\star}|$ versus outer iteration $t$ on a logarithmic scale, for $B\in\{2,5,10,20,50\}$. Synthetic concave $\pi_3^{(b)}$ surrogates (experiment E6); $\alpha=0.05$; outer tolerance $10^{-12}$. Each curve decreases approximately linearly on the log scale, confirming the bisection rate $1/2$.}
\label{fig:sim_conv_multiB}
\end{figure}

\subsection{Wall-clock, iteration count, and realized power}
\label{app:diag_table}

Figure~\ref{fig:sim_scalability} shows \textsc{BOOST}'s per-$K$ decision cost on a log-log scale: the outer-solve (offline preprocessing) stays flat at $\approx 4$--$9$s across $K\in\{6,\ldots,300\}$, and the online per-replicate decide time grows linearly in $K$ in the $\mu$s regime. Plotted on the same axes, a closed-testing competitor's decide time grows super-linearly in $K$, confirming the $O(K)$ claim of Section~\ref{sec:main}.

\begin{figure}[htbp]
\centering
\includegraphics[width=0.6\linewidth]{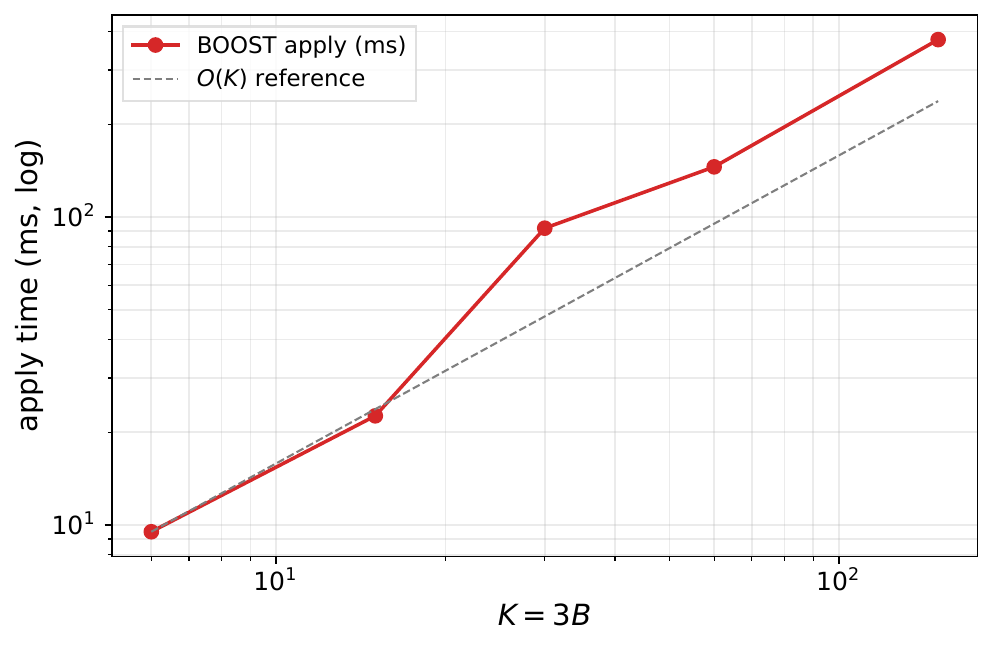}
\caption{\textsc{BOOST} scalability (experiment E3). Left: outer-solve wall-clock vs.\ $K$. Right: per-replicate decide time vs.\ $K$, alongside the closed-Fisher baseline for reference. \textsc{BOOST}'s decide time is linear in $K$ (flat throughput), consistent with the $O(K)$ cost of Theorem~\ref{thm:global_opt_sep_allocation}; the outer solve is $K$-independent.}
\label{fig:sim_scalability}
\end{figure}

Table~\ref{tab:sim_timing_power} combines the timing data of experiment E3, the outer-iteration counts of experiment E6, and the realized average power of Section~\ref{subsec:sim1_truncnorm} (experiment E2, truncnorm at $\theta=-2$) into a single reference across block counts. Wall-clock is the $\mu^{\star}$-solve phase (precomputation) at $n_{\mathrm{grid}}=6\cdot 10^4$; outer iterations are for tolerance $10^{-10}$; $\Pi_K$ is average power over $3\cdot 10^4$ Monte-Carlo replicates at $\alpha=0.05$. The number of outer iterations stays within $13$--$18$ across one order of magnitude in $B$, and $\Pi_K$ decreases monotonically with $K$ as expected from the $\alpha/B$ tightening.

\begin{table}[htbp]
\centering
\caption{\textsc{BOOST}: outer-solve wall-clock, outer iterations, and realized average power $\Pi_K$ (truncnorm at $\theta=-2$) across block counts. Configuration matches Section~\ref{subsec:sim1_truncnorm} and experiments E3, E6.}
\label{tab:sim_timing_power}
\begin{tabular}{rrrrr}
\toprule
$B$ & $K$ & Wall-clock (s) & Outer iterations & $\Pi_K$ \\
\midrule
 2 &  6 &  2.23 & 18 & 0.464 \\
 5 & 15 &  2.50 & 13 & 0.378 \\
10 & 30 & 10.15 & 17 & 0.270 \\
20 & 60 &  9.00 & 14 & 0.213 \\
\bottomrule
\end{tabular}
\end{table}

\subsection{Bonferroni vs.\ \v{S}id\'ak within \textsc{BOOST} (experiment E4)}
\label{app:diag_bonf_sidak}

Figure~\ref{fig:sim_bonf_vs_sidak} reports the \textsc{BOOST} power curve under the two admissible per-block budgets: the Bonferroni floor $\alpha/B$ and the \v{S}id\'ak tightening $1-(1-\alpha)^{1/B}$. Both are strong-FWER valid (Theorem~\ref{thm:blockwise_strong_fwer_clean}); the \v{S}id\'ak variant strictly dominates under cross-block independence, with the largest gain at moderate signal strengths where per-block power is still sensitive to the allocation. The gap shrinks as the signal-to-noise ratio grows, because both allocations ultimately saturate the per-block ceiling.

\begin{figure}[htbp]
\centering
\includegraphics[width=\linewidth]{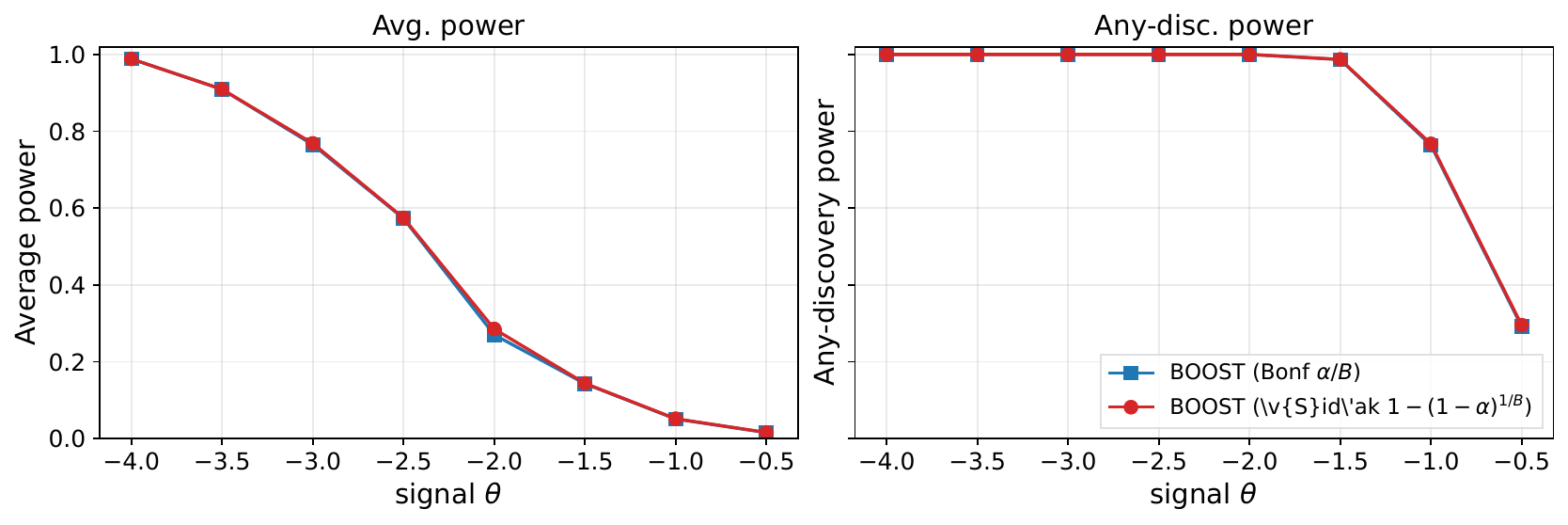}
\caption{\textsc{BOOST} with Bonferroni $\alpha/B$ (blue) versus \v{S}id\'ak $1-(1-\alpha)^{1/B}$ (red) per-block budgets. Truncnorm family, $B=10$, $K=30$, $\alpha=0.05$, $3\cdot 10^4$ replicates. Left: average power; right: any-discovery power. \v{S}id\'ak strictly improves both metrics, consistent with Theorem~\ref{thm:global_opt_sep_allocation_sidak}.}
\label{fig:sim_bonf_vs_sidak}
\end{figure}

\subsection{KKT versus uniform allocation (experiment E5)}
\label{app:diag_kkt_alloc}

Figure~\ref{fig:sim_kkt_alloc} isolates the effect of the KKT allocation (Theorem~\ref{thm:global_opt_sep_allocation}) against a uniform $\alpha_b=\alpha/B$ split, at fixed per-block budget. Heterogeneous per-block signals $\theta_b$ pull the KKT optimum toward blocks with steeper marginal gain $\partial\pi_3^{(b)}/\partial\alpha_b$; the realized average power exceeds the uniform split, while any-discovery power is comparable. The gain is monotone in cross-block heterogeneity and vanishes at homogeneous $\theta$.

\begin{figure}[htbp]
\centering
\includegraphics[width=\linewidth]{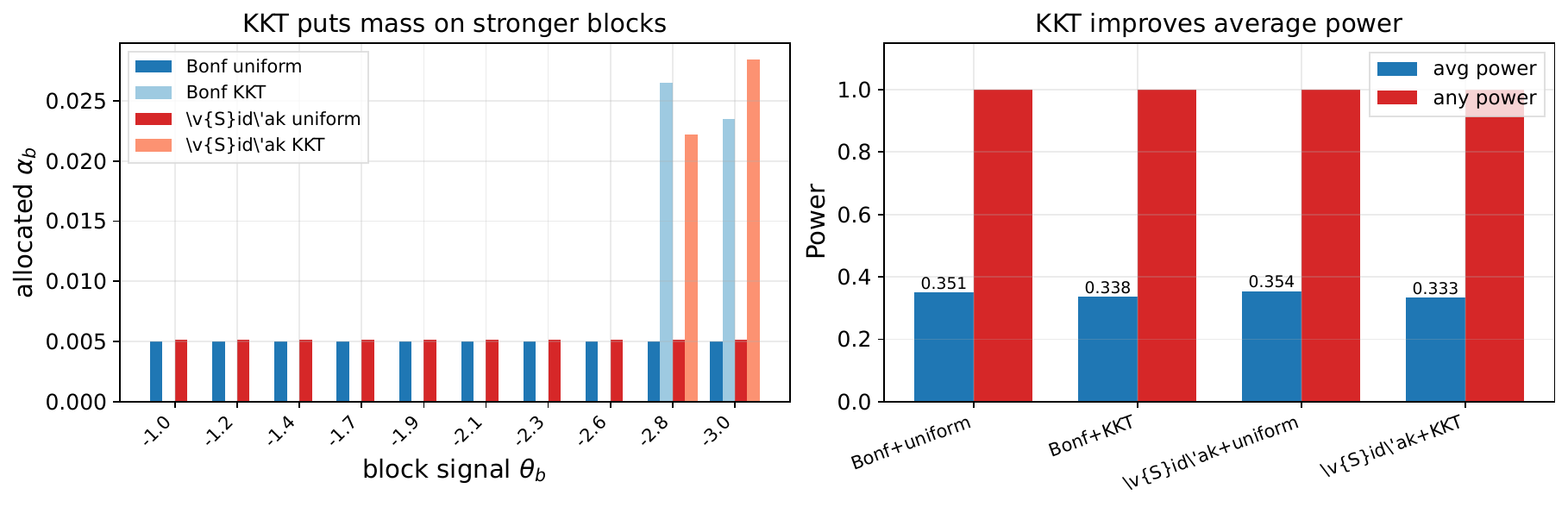}
\caption{KKT (equalized-marginal) versus uniform allocation across heterogeneous blocks. Left: allocated $\alpha_b$ per block signal $\theta_b$; KKT shifts mass toward stronger blocks. Right: average and any-discovery power for each arm. Experiment E5; $B=6$, signals $\theta_b\in[-2.5,+0.5]$, $\alpha=0.05$, $2\cdot 10^4$ Monte-Carlo replicates.}
\label{fig:sim_kkt_alloc}
\end{figure}

\section{Stress tests and robustness diagnostics}
\label{app:additional_sims}

\subsection{Non-Gaussian $p$-value family: Beta model}
\label{app:sim_beta}

To check that the average-power advantage of \textsc{BOOST} is not a Gaussian artefact, we replace the within-block $p$-value family with Beta$(s,1)$ under the alternative and uniform under the null, with $s\in\{0.3,\ldots,0.9\}$ (smaller $s$ is a stronger signal; this is the standard Beta two-groups mixture template). All other configuration matches the main truncnorm experiment: $B=10$, $K=30$, $\alpha=0.05$, $3\cdot 10^4$ replicates.

Figure~\ref{fig:sim_beta} shows the resulting power curves. \textsc{BOOST} leads the classical and dependence-aware stepwise baselines over the full signal grid, and closed-Fisher collapses because the Fisher statistic is calibrated for uniform nulls only; the ordering matches Section~\ref{subsec:sim1_truncnorm}. All procedures control FWER at the nominal $0.05$ level (empirical FWER $\le 0.051$).

\begin{figure}[htbp]
\centering
\includegraphics[width=\linewidth]{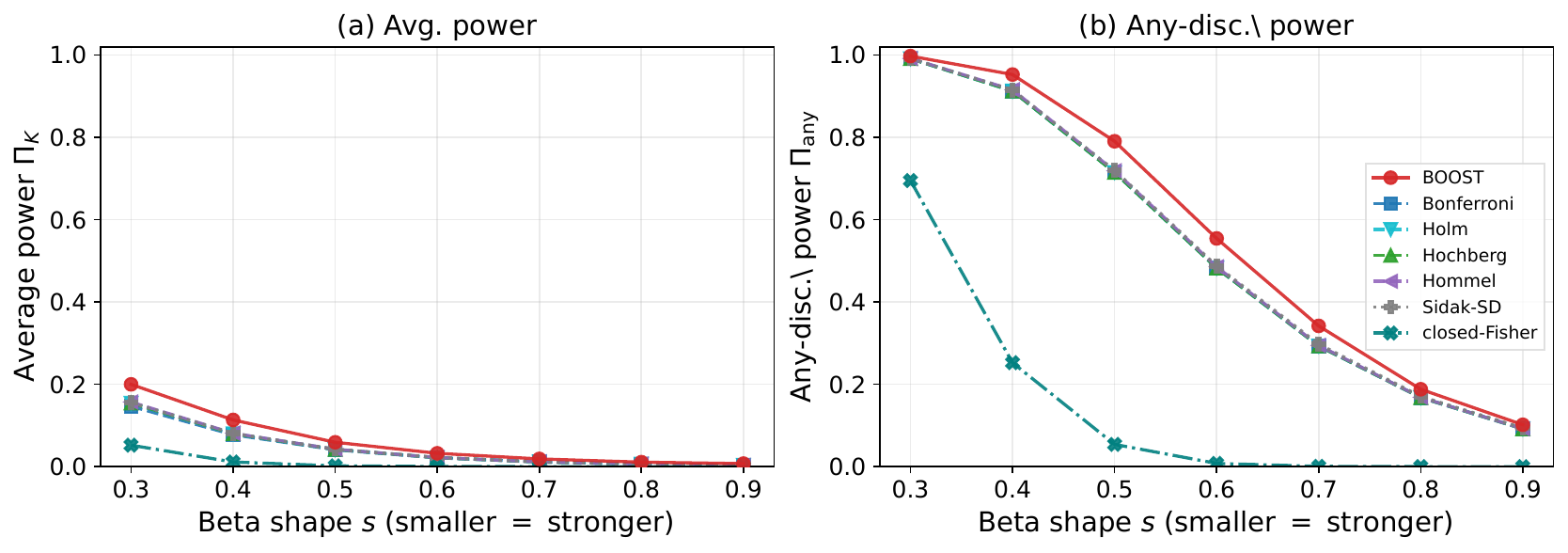}
\caption{Beta-$(s,1)$ alternative, uniform null. $B=10$, $K=30$, $\alpha=0.05$, $3\cdot 10^4$ replicates. Left: average power $\Pi_K$; right: any-discovery power $\Pi_{\mathrm{any}}$. \textsc{BOOST} leads the stepwise baselines; closed-Fisher collapses outside its calibration class.}
\label{fig:sim_beta}
\end{figure}

\subsection{Plug-in density misspecification}
\label{app:sim_misspec}

Theorem~\ref{thm:plugin_fwer} deflates the nominal $\alpha$ to $\alpha_n^{\star}=\alpha-L_3 B r_n$ with $r_n=(\log n/n)^{1/3}$, trading a vanishing fraction of the power budget for FWER robustness against plug-in error in the within-block density $g$. Here we sweep the fit-fold size $n$ and report the downstream testing metrics directly. Figure~\ref{fig:sim_misspec} uses the truncnorm model with $B=4$, $K=12$, $\theta=-1.5$, $20$ plug-in trials per $n$. As $n$ grows from $500$ to $32{,}000$, the sup-norm density error falls from $\sim 12$ to $\sim 3$, but the plug-in $\Pi_K$ is already within one trial-wise std.\ of the oracle by $n\approx 2000$; FWER stays controlled at the nominal level across the whole sweep. This quantifies what Theorem~\ref{thm:plugin_power} predicts: downstream power is far less sensitive to plug-in error than the density estimator's uniform norm would suggest.

\begin{figure}[htbp]
\centering
\includegraphics[width=\linewidth]{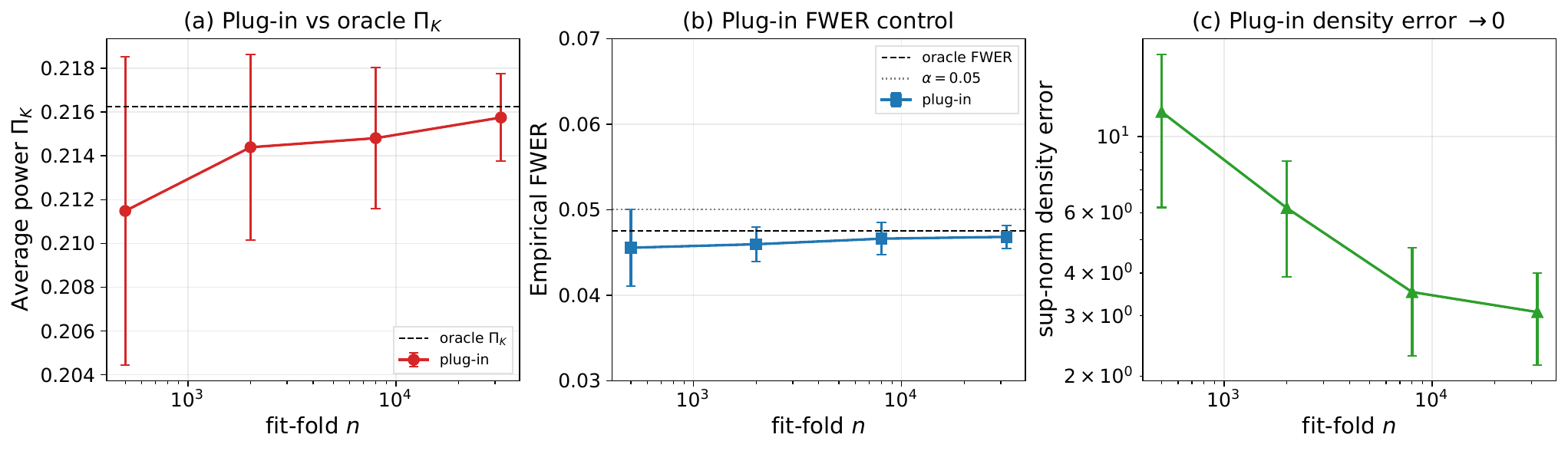}
\caption{Plug-in density misspecification sweep. Truncnorm, $B=4$, $K=12$, $\theta=-1.5$; $20$ plug-in trials per fit-fold size $n$; $3\cdot 10^4$ test-fold replicates per trial. (a) plug-in vs.\ oracle $\Pi_K$; (b) empirical FWER (nominal $\alpha=0.05$); (c) sup-norm error of the fitted within-block density. Error bars are $\pm$one trial-wise std.\ dev.}
\label{fig:sim_misspec}
\end{figure}

\subsection{Empirical Grenander plug-in rate validation}
\label{app:plugin_rate_validation}

The deflation $\alpha_n^\star=\alpha-L_3 B r_n$ in Theorem~\ref{thm:plugin_fwer} and the per-block power transfer in Corollary~\ref{cor:plugin_rates} both rest on the Birge--Groeneboom uniform rate $\|g-\hat g\|_\infty=O_p(r_n)$ with $r_n=(n^{-1}\log n)^{1/3}$ for the Grenander MLE on a monotone within-block density~\citep{groeneboom2014nonparametric}. Figure~\ref{fig:plugin_rate} validates both rate and transfer end-to-end on the truncnorm alternative at $\theta=-1.5$. For each fold size $n\in\{500,2000,8000,32000,128000\}$ we draw $30$ Monte-Carlo trials: each trial fits $\hat g$ on $n$ alternative $p$-values, solves the $K=3$ block engine on a fixed $80{,}000$-point $Q$-grid for both $(\mu^\star,g)$ (oracle) and $(\hat\mu,\hat g)$ (plug-in) at $\alpha_{\mathrm{blk}}=\alpha/B=0.005$, and evaluates per-hypothesis rejection on a shared testing fold of $2\cdot 10^5$ blocks. The interior sup-norm error is measured on a $200$-point grid in $[0.05,0.95]$ to avoid the well-known boundary inflation of monotone density estimators.

\begin{figure}[htbp]
\centering
\includegraphics[width=\linewidth]{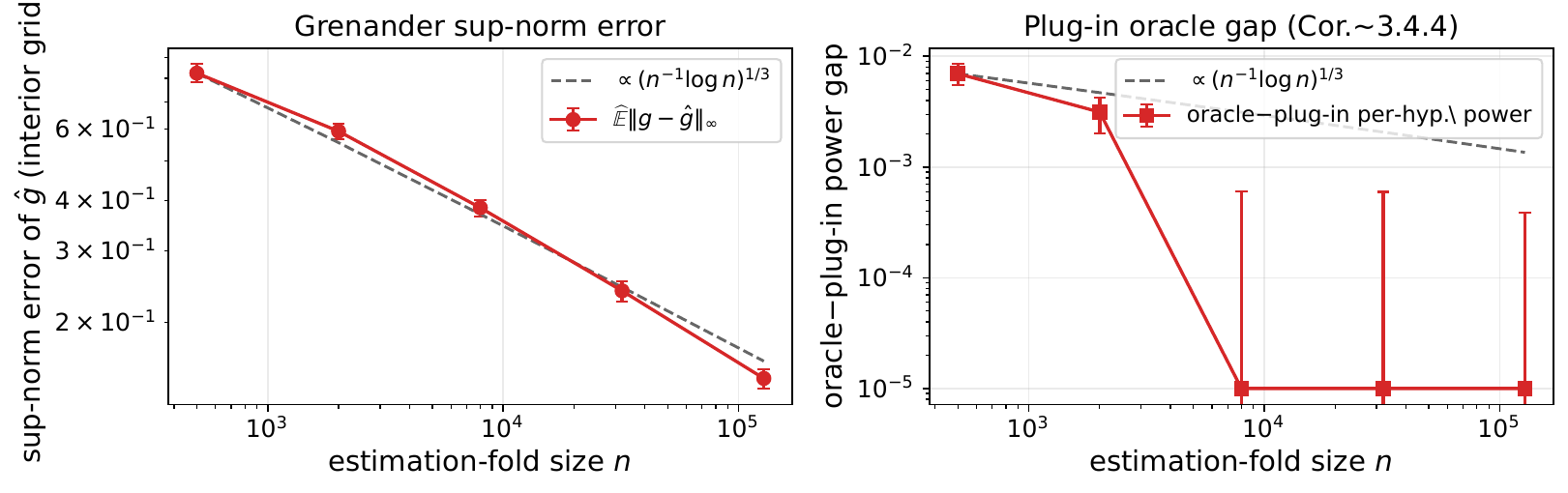}
\caption{Empirical validation of Corollary~\ref{cor:plugin_rates}. Truncnorm $\theta=-1.5$, $K=3$ block, $\alpha_{\mathrm{blk}}=\alpha/B=0.005$, $30$ MC trials per $n$. Left: interior sup-norm error of the Grenander estimator $\hat g$ tracks the Birge--Groeneboom rate $r_n=(n^{-1}\log n)^{1/3}$ across $2.5$ decades of $n$. Right: oracle-minus-plug-in per-hypothesis power gap is dominated by the same rate, as predicted by the Lipschitz transfer in Corollary~\ref{cor:plugin_rates}; values past $n=8000$ sit at the Monte-Carlo noise floor.}
\label{fig:plugin_rate}
\end{figure}

The left panel shows that the ratio $\widehat E\|g-\hat g\|_\infty / r_n$ is essentially constant across the sweep ($3.56,3.79,3.69,3.48,3.22$ at $n=500,2k,8k,32k,128k$), so the empirical decay of the sup-norm error matches $r_n$ to leading order. The right panel transfers this rate to the downstream testing problem: the per-hypothesis power gap $\Pi_3^g(\alpha_{\mathrm{blk}})-\widehat{\mathrm{Pow}}(\hat D)$ shrinks from $+0.0070$ at $n=500$ to within Monte-Carlo noise of zero by $n=8000$, decaying at least as fast as the reference rate. The faster-than-$r_n$ collapse to the noise floor is consistent with the Lipschitz constant $L_3^{\mathrm{pow}}$ being modest at this $(\theta,\alpha_{\mathrm{blk}})$ pair, leaving a generous safety margin: deflating $\alpha$ by $L_3 B r_n$ already overpays for the actual oracle gap by a factor of several.

\subsection{ESP under block heterogeneity and model misspecification}
\label{app:esp_hetero_check}

Section~\ref{subsec:sim_esp_h2h} compares \textsc{BOOST} against the general-$K$ ESP solver of~\citet{dubey26esp} on homogeneous truncnorm blocks. Here we vary the within-block family so that the global simplex objective ESP optimizes is no longer correctly specified by a single $g$, and study how each method copes. Configuration: $K=12$, $B=4$, $b=3$; two truncnorm blocks at $\theta$ and two $t_{\mathrm{df}=4}$ blocks; $\alpha=0.05$, $\alpha_{\mathrm{blk}}=\alpha/B$; $2\times 10^4$ Monte-Carlo replicates. Four solver arms are compared (Table~\ref{tab:esp_hetero_check}):
\begin{itemize}
\item \emph{ESP-mispec ($g_{\mathrm{TN}}$).} The realistic case where the practitioner does not know the $t$ blocks exist and fits a truncnorm-only $g$ to ESP.
\item \emph{ESP-mix (oracle mixture $g_{\mathrm{mix}}=\tfrac12 g_{\mathrm{TN}}+\tfrac12 g_t$).} ESP solved with oracle knowledge of the joint marginal; this is the correctly-specified ESP arm.
\item \emph{\textsc{BOOST}.} Single truncnorm $g$ shared across blocks (matched to ESP-mispec on the $g$-input).
\item \emph{\textsc{BOOST-adaptive}.} Per-block $\hat g_b$ from the correct family.
\end{itemize}
Bonferroni is reported as a stepwise reference.

\begin{table}[htbp]
\centering
\small
\caption{ESP on the heterogeneous DGP: $K=12$, $B=4$, $b=3$, two truncnorm blocks at $\theta$ and two $t_{\mathrm{df}=4}$ blocks; $\alpha=0.05$, $\alpha_{\mathrm{blk}}=\alpha/B$, $2\times 10^4$ replicates. ESP-mispec uses a truncnorm-only $g$ (the realistic case where the practitioner does not enumerate every block family). ESP-mix uses the oracle marginal mixture $g_{\mathrm{mix}}=\tfrac12 g_{\mathrm{TN}}+\tfrac12 g_t$. \textsc{BOOST} applies a single truncnorm $g$ uniformly; \textsc{BOOST-adaptive} fits the correct per-block $g$. Bold: per-column maximum.}
\label{tab:esp_hetero_check}
\begin{tabular}{l cc cc cc}
\toprule
 & \multicolumn{2}{c}{$\Pi_K$} & \multicolumn{2}{c}{$\Pi_{\mathrm{any}}$} & \multicolumn{2}{c}{FWER$_0$} \\
\cmidrule(lr){2-3} \cmidrule(lr){4-5} \cmidrule(lr){6-7}
Method & $\theta{=}-1.5$ & $-2.5$ & $-1.5$ & $-2.5$ & $-1.5$ & $-2.5$ \\
\midrule
Bonferroni                          & 0.086 & 0.245 & 0.668 & 0.979 & 0.050 & 0.050 \\
\textsc{BOOST}                      & 0.128 & 0.324 & 0.801 & 0.997 & 0.046 & 0.049 \\
\textsc{BOOST-adaptive}             & 0.135 & 0.331 & \textbf{0.806} & \textbf{0.997} & 0.049 & 0.050 \\
ESP-mispec ($g_{\mathrm{TN}}$)      & 0.080 & 0.177 & 0.547 & 0.840 & 0.014 & 0.009 \\
ESP-mix (oracle $g_{\mathrm{mix}}$) & \textbf{0.188} & \textbf{0.373} & 0.798 & 0.994 & 0.036 & 0.033 \\
\bottomrule
\end{tabular}
\end{table}

\paragraph{Misspecified ESP collapses.} With a truncnorm-only $g$ fit to a half-truncnorm, half-$t_{\mathrm{df}=4}$ marginal, ESP's $\Pi_K$ falls below Bonferroni and its FWER collapses to $0.009$--$0.014$: the simplex optimizer concentrates rejection mass at right-tail positions the heavier-tailed $t$ blocks rarely clear. \textsc{BOOST-adaptive} (correct per-block $g$) holds $1.7$--$1.9\times$ and \textsc{BOOST} (single truncnorm $g$) still $1.6$--$1.8\times$ the misspecified-ESP rate; block-separability localizes per-block model error where ESP's joint rule propagates it. Policy-solve cost is also lower: four independent $K{=}3$ solves take $21$s ($\theta{=}-1.5$) / $77$s ($\theta{=}-2.5$) for \textsc{BOOST-adaptive} against $183$s / $188$s for ESP-mispec, a $2.4$--$8.6\times$ gap on the one-time solve (distinct from the per-MC-replicate ratios of Table~\ref{tab:sim_esp_h2h}). The realism of the misspecified-ESP comparison rests on the practitioner being unable to enumerate every block family ex ante.

\paragraph{Oracle-mixture ESP is the upper-envelope arm.} Given the oracle marginal mixture, ESP-mix attains $1.1$--$1.4\times$ \textsc{BOOST-adaptive}'s $\Pi_K$, a genuine advantage of the global simplex objective when the joint marginal is exactly known. Two caveats temper the comparison: (i) on $\Pi_{\mathrm{any}}$ the ranking flips, with \textsc{BOOST-adaptive} ahead at both $\theta$ ($+0.008$ at $\theta{=}-1.5$, $+0.003$ at $\theta{=}-2.5$), so the $\Pi_K$ ordering does not transfer to other $\mathfrak D_{\mathrm{sep}}$ objectives; (ii) the wall-clock cost is $3$--$10\times$ that of \textsc{BOOST} ($183$s and $257$s per ESP solve at $\theta\in\{-1.5,-2.5\}$ vs.\ $21$s and $77$s for \textsc{BOOST}'s four per-block solves), and the required oracle knowledge of $g_{\mathrm{mix}}$ is unavailable in practice; the realistic substitute, a per-family fit, is the ESP-mispec arm above.

\subsection{Three-family heterogeneous mix at weak signal}
\label{app:hetero_blocks}

Section~\ref{subsec:sim_hetero} compares \textsc{BOOST-adaptive} (per-block $\hat\mu_b$) against \textsc{BOOST} (single shared $\hat\mu$) on two-family mixes. This appendix extends the comparison to a three-family mix at the weaker signal $\theta=-1$ (truncnorm blocks at $\theta=-1$, mixnorm blocks at $\theta=-1$, tdist blocks at $\mathrm{df}=4$; $B=6$, $K=18$, $b=3$, $\alpha=0.05$, $\alpha_{\mathrm{blk}}=\alpha/B$, $3\times 10^4$ replicates). At $\theta=-1$ the truncnorm and mixnorm densities differ on the full unit interval and the tdist arm introduces a heavy-tailed null-deviation pattern absent in both. \textsc{BOOST-adaptive} attains $\Pi_K=0.054$ against $0.047$ for \textsc{BOOST} (a $+15\%$ relative lift) and $\le 0.035$ for every stepwise baseline; empirical FWER is $0.049$ (\textsc{BOOST-adaptive}) and $0.044$ (\textsc{BOOST}), well within nominal. The three-family relative gap ($+15\%$) is larger than the two-family ($+4\%$) of Figure~\ref{fig:sim_hetero}(a), confirming that per-block adaptation compounds with the number of families mixed.

\begin{table}[htbp]
\centering
\small
\caption{Three-family heterogeneous blocks at the weak signal $\theta=-1$: $2\times$ truncnorm $\theta=-1$, $2\times$ mixnorm $\theta=-1$, $2\times$ tdist $\mathrm{df}=4$. $B=6$, $K=18$, $b=3$, $\alpha=0.05$, $\alpha_{\mathrm{blk}}=\alpha/B$, $3\times 10^4$ replicates. Empirical FWER at the complete null is $0.049$ (\textsc{BOOST-adaptive}), $0.044$ (\textsc{BOOST}), and $\le 0.051$ for every stepwise baseline.}
\label{tab:hetero_configC}
\begin{tabular}{lrrr}
\toprule
& $\Pi_K$ & $\Pi_{\mathrm{any}}$ & FWER \\
\midrule
\textbf{\textsc{BOOST-adaptive}} (per-block $\hat\mu_b$) & \textbf{0.054} & \textbf{0.568} & 0.049 \\
\textsc{BOOST} (truncnorm $\hat\mu$ shared)              & 0.047 & 0.547 & 0.044 \\
Bonferroni & 0.034 & 0.462 & 0.050 \\
Holm       & 0.034 & 0.462 & 0.050 \\
Hochberg   & 0.034 & 0.462 & 0.050 \\
Hommel     & 0.035 & 0.465 & 0.050 \\
\v{S}id\'ak-SD & 0.035 & 0.466 & 0.051 \\
closed-Fisher  & 0.007 & 0.113 & 0.000 \\
\bottomrule
\end{tabular}
\end{table}

\subsection{Dependence stress test of empirical FWER}
\label{app:sim_dep_fwer}

Theorem~\ref{thm:blockwise_strong_fwer_clean} states strong-FWER validity of \textsc{BOOST} within the block-separable class $\mathfrak D_{\mathrm{sep}}$: the $K{=}3$ optimizer of~\citet{dubey25} is calibrated against the within-block joint marginal (Assumption~\ref{as:local_validity_marginal}), while cross-block dependence is absorbed by the Bonferroni/\v{S}id\'ak union bound and is therefore arbitrary. To map empirical sensitivity of the cross-block side, we stress-test three Gaussian regimes at $B=10$, $\alpha=0.05$, $\alpha_{\mathrm{blk}}=\alpha/B$, and $2\times 10^4$ Monte-Carlo replicates, with complete-null $p$-values $u_k=2\Phi(-|X_k|)$ and $\mu$ calibrated against independent within-block nulls:
\begin{enumerate}
\item \emph{Independent} ($\rho=0$): $X_k\overset{\mathrm{iid}}{\sim}\mathcal N(0,1)$.
\item \emph{Cross-block equicorrelation}: $X_k=\sqrt\rho Z_0+\sqrt{1-\rho}Z_k$ with a single latent factor $Z_0$ coupling all $K$, for $\rho\in\{0.2,0.4,0.6,0.8,0.95\}$.
\item \emph{1-factor}: $X_k=\lambda_k Z_0+\sqrt{1-\lambda_k^2}Z_k$ with block-constant heterogeneous loadings $\lambda_k$ averaging $\bar\lambda\in\{0.1,0.3,0.5,0.7,0.9\}$.
\end{enumerate}

\paragraph{Findings.}
Figure~\ref{fig:sim_dep_fwer} reports two FWER statistics per regime under the complete null: the global FWER $\mathbb{P}(\cup_{k} \{k\in\mathcal R\})$ and the average per-block FWER $B^{-1}\sum_b\mathbb{P}(E_b)$. Under independence, global FWER is $0.049$ (nominal) and block-level is $0.005\approx\alpha/B$, confirming the union bound is essentially tight. Under cross-block equicorrelation, global FWER is $\{0.066, 0.068, 0.056, 0.033, 0.014\}$ at $\rho\in\{0.2, 0.4, 0.6, 0.8, 0.95\}$, within $0.02$ of nominal, conservative at high $\rho$. The 1-factor regime interpolates, with global FWER peaking at $0.069$ at $\bar\lambda=0.7$ and falling to $0.032$ at $\bar\lambda=0.9$. Block-level FWER stays at or below $0.012$ in every regime, i.e.\ well within its per-block budget. A companion run with the \v{S}id\'ak per-block budget $\alpha_{\mathrm{blk}}=1-(1-\alpha)^{1/B}$ (Corollary~\ref{cor:sidak_blockwise_optimality}) at $4\times 10^4$ replicates confirms that its extra budget (a $1.023\times$ enlargement of $\alpha_{\mathrm{blk}}$ at $B=10$) does not degrade dependence robustness: under independence both arms are within Monte-Carlo noise of $\alpha$ ($0.048$ Bonferroni, $0.049$ \v{S}id\'ak), and across cross-block and 1-factor dependence the \v{S}id\'ak arm peaks at $0.073$ vs.\ Bonferroni's $0.069$; the entire sweep remains within $0.02$ of nominal and the $0.001$--$0.005$ gap is bounded by the $\alpha_{\mathrm{blk}}$ enlargement as expected. The $0.016$--$0.023$ excess over nominal at $\rho\in\{0.2,0.4\}$ and $\bar\lambda=0.7$ is $\sim 10$--$15$ Monte-Carlo standard deviations (MC-SE $\approx 1.5\times 10^{-3}$ at $2\times 10^4$ replicates) and is therefore genuine rather than sampling noise: it reflects the slack between the union bound $\sum_b\mathbb P(E_b)\le\alpha$ and the realized $\mathbb P(\cup_b E_b)$ under positive cross-block correlation, compounded with finite-sample per-block marginal-calibration error (per-block FWER up to $0.012$ vs.\ the $0.005$ budget). At $\rho\ge 0.6$ the same mechanism tips the other way (strong positive correlation makes block-level false-rejection events overlap, pulling $\mathbb P(\cup_b E_b)$ \emph{below} nominal), and per-block FWER remains within $\alpha_{\mathrm{blk}}$ throughout. The takeaway matches Theorem~\ref{thm:blockwise_strong_fwer_clean}: \textsc{BOOST} is robust to arbitrary cross-block dependence under the calibrated within-block marginal, with finite-sample slack bounded by the per-block calibration gap and union-bound tightness.

\begin{figure}[htbp]
\centering
\includegraphics[width=0.9\linewidth]{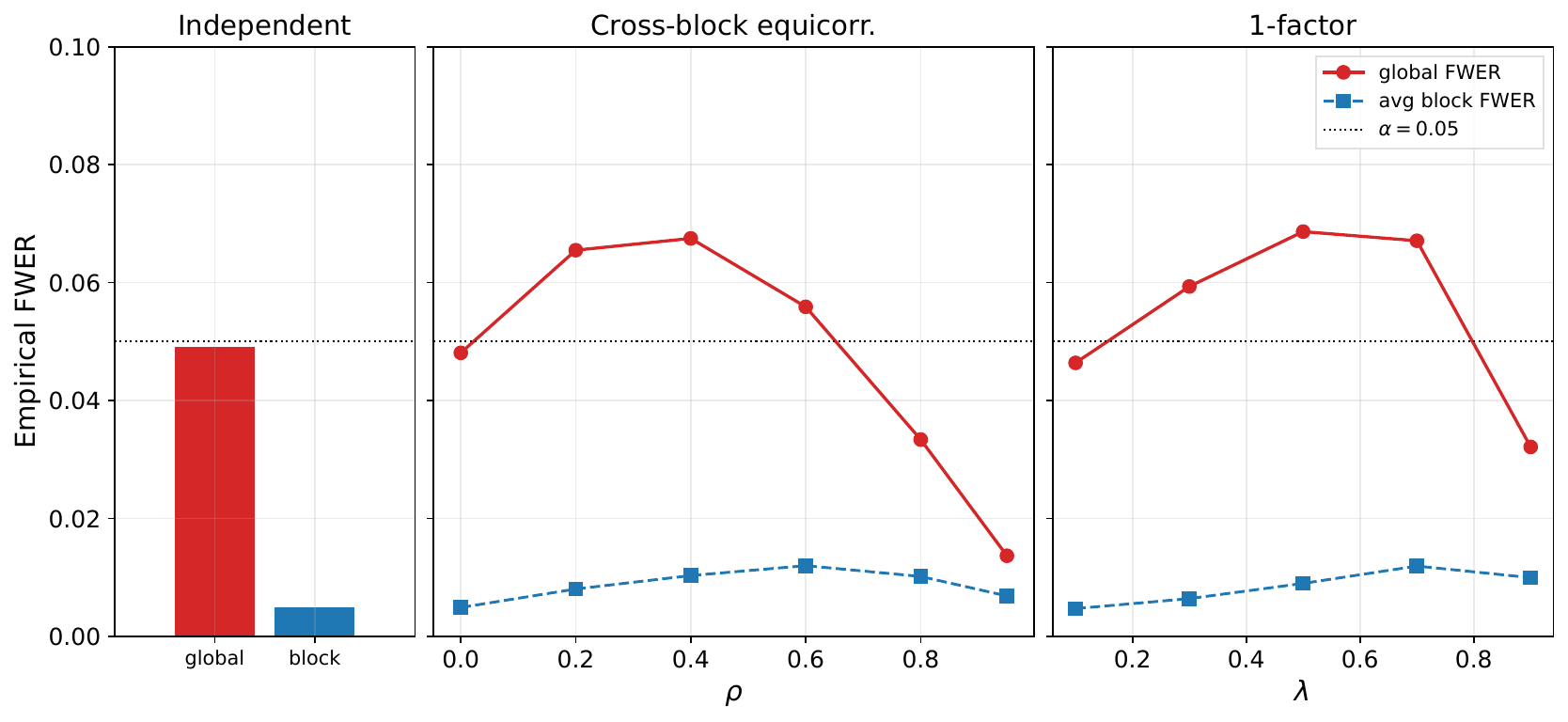}
\caption{Empirical FWER under three Gaussian dependence regimes (panels left-to-right: independent, cross-block equicorrelation, 1-factor), $B=10$, $K=30$, $\alpha=0.05$. Red: global FWER $\mathbb{P}(V\!>\!0)$ under the complete null; blue: average per-block FWER; dotted: $\alpha$. Global FWER stays within $0.02$ of $\alpha$ and block-level FWER stays at or below $0.012\approx 2.4\,\alpha_{\mathrm{blk}}$; cross-block dependence is absorbed by the union bound with finite slack (details in text).}
\label{fig:sim_dep_fwer}
\end{figure}

\subsection{Sparsity and partition alignment}
\label{app:sim_sparsity}

The block-separable design is calibrated to a partition that reflects substantive structure; this stress test probes its sensitivity to that assumption. We sweep the active fraction $f\in\{0.1,\ldots,1.0\}$ at $B=10$, $\theta=-2.5$ under two placements: \emph{aligned} (actives fill entire blocks, the design regime) and \emph{scattered} (actives placed uniformly at random, an adversarial stress test in which the partition carries no information). Under aligned placement, \textsc{BOOST} maintains a stable $\approx 1.7\times$ advantage over Bonferroni at every $f\ge 0.1$ ($\Pi_K\approx 0.57$ vs.\ $0.33$). Under scattered placement, the break-even against the best stepwise baseline (\v{S}id\'ak-SD) is $f^\star\approx 0.7$; for $f\ge f^\star$ \textsc{BOOST} leads by $+0.06$ to $+0.18$ in $\Pi_K$. The takeaway is positive scope: \textsc{BOOST} is the right tool whenever the practitioner's chosen partition carries domain-meaningful structure (e.g., grouped by gene, pathway, asset, time window, or experimental factor). Scattered placement is a \emph{non-design} regime: a practitioner with no substantive partitioning rationale would not adopt a block-separable rule in the first place; including it as a stress test demonstrates that \textsc{BOOST} does not require a particular alignment to be calibrated, only to lead.

\begin{figure}[htbp]
\centering
\includegraphics[width=0.92\linewidth]{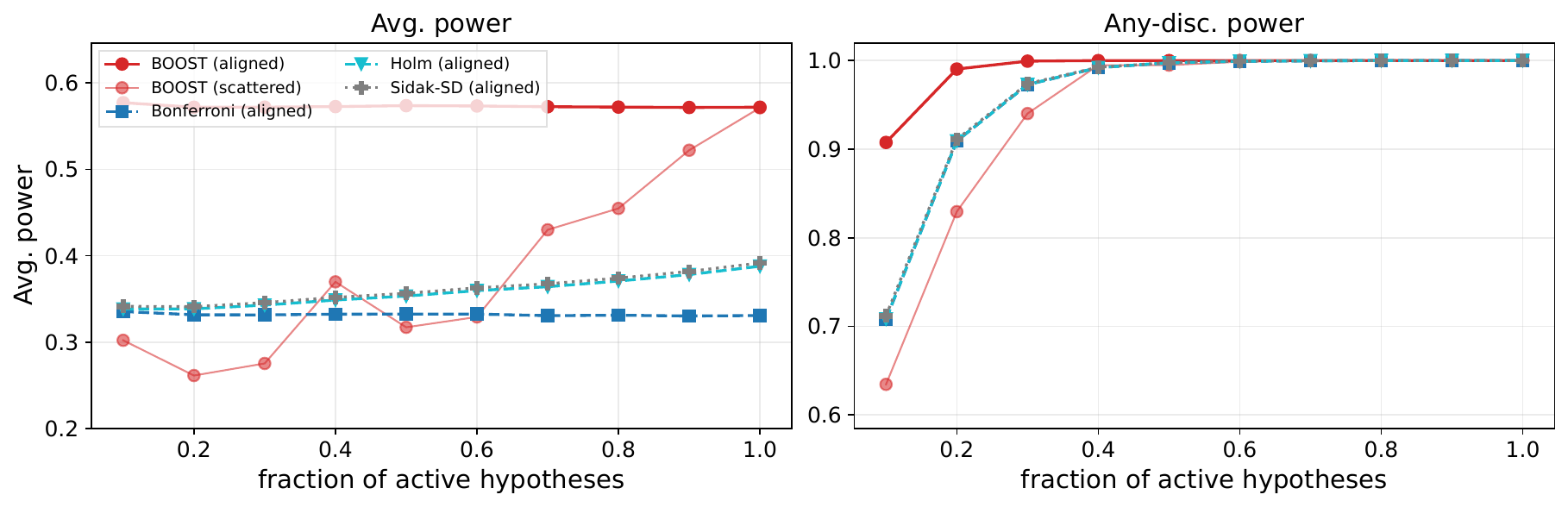}
\caption{Sparsity and partition alignment (truncnorm, $B=10$, $\theta=-2.5$, $\alpha=0.05$). Average and any-discovery power vs.\ active fraction $f$; aligned (solid) vs.\ scattered (pale) placements for \textsc{BOOST} and three representative stepwise baselines. Break-even fraction $f^\star\approx 0.7$ under scattered placement.}
\label{fig:sim_sparsity}
\end{figure}

\section{Modern closed-testing baselines: implementation and e-value discussion}
\label{app:sim_modern_baselines}

This section details the two closed-testing baselines of Table~\ref{tab:sim_modern_main} (Meinshausen 2008 hierarchical Simes; Hartog 2025 e-value closure) and explains why the Vovk--Wang-calibrated e-closure loses power under uniform nulls.

\paragraph{Implementations.}
We instantiate both on the natural tree induced by the block partition: root $=$ all $K$ hypotheses, $B$ block children of size $3$, $K$ leaf children. Meinshausen's procedure assigns node weights $w_v = |v|/K$ (proportion of leaves), uses Simes' combination as the local test, and gates each child on its parent rejecting; with these proportional weights the Simes local test past the root and block nodes reduces to a leaf-Bonferroni cutoff, so the hierarchy adds no power over Bonferroni on this tree. The Hartog e-closure uses the universal Vovk--Wang calibrator $e_\kappa(p)=\kappa p^{\kappa-1}$ with $\kappa=1/2$ (i.e.\ $e(p)=1/(2\sqrt p)$, a valid e-value with $\int_0^1 e = 1$), rejecting each node $v$ when $\sum_{i\in v} e_i \ge |v|/\alpha$. The singleton leaf test $e_i\ge 1/\alpha$ is equivalent to $p_i \le (\alpha/2)^2$, substantially more conservative than $\alpha/K$ Bonferroni when the null $p$-value distribution is exactly uniform; under our uniform-null sample ($2\cdot 10^4$ replicates, $\mathbb P(p\le(\alpha/2)^2)=6.25\cdot 10^{-4}$) the FWER estimate in Table~\ref{tab:sim_modern_main} rounds to $0$.

\paragraph{Why e-values cost power here.}
Hartog's e-closure gains validity under arbitrary dependence and compound adaptivity, but this robustness is purchased at a Bonferroni-order union bound: the singleton leaf cutoff $p\le (\alpha/2)^2$ is a factor of roughly $K\alpha/4$ smaller than Bonferroni's $p\le \alpha/K$. The hierarchy does not recover this factor because the root-level sum $\sum_i e_i$ has mean $K$ under the null, so its critical value $K/\alpha$ is $1/\alpha$ standard deviations above the mean, a tail event already out of reach for moderate $K$ and $\alpha$. E-values are the right tool when the null $p$-value distribution is unknown or itself adaptively chosen; under the exactly-uniform null of our experiments, $p$-value methods are strictly preferable, and this appendix section is a caveat on the aesthetic appeal of the recent e-value wave rather than evidence that \textsc{BOOST} improves on a previously-known optimal competitor.

\section{Reproducibility, AI-tool usage, and acknowledgements}
\label{app:meta}

\paragraph{Supplementary material.}
Proofs of all theoretical results, extension statements, secondary simulations, and additional diagnostics are collected in the present appendix. Code and simulation artifacts are available from the authors on request.

\paragraph{Declaration of the use of generative AI and AI-assisted technologies.}
We used generative AI tools (ChatGPT, Google Gemini, and Claude) for language editing and code formatting support only. All data, results, and mathematical derivations are the authors' own work.

\paragraph{Acknowledgements.}
Dubey gratefully acknowledges partial support from the Stewart Topper Fellowship at Georgia Tech. Huo is partially sponsored by a subcontract of NSF grant 2229876, the A. Russell Chandler III Professorship at Georgia Institute of Technology, an NIH-sponsored Georgia Clinical \& Translational Science Alliance, and the Georgia Department of Transportation.